\documentclass[a4paper,10pt,notitlepage,nofootinbib,floats,onecolumn,tightenlines,
superscriptaddress,longbibliography,accepted=2021-10-29]{quantumarticle}
\pdfoutput=1
\usepackage[utf8]{inputenc}
\usepackage{amsmath}
\usepackage{amsthm}
\usepackage{amssymb}
\usepackage{amsfonts}
\usepackage{dsfont} 
\usepackage{graphicx}
\usepackage{latexsym}
\usepackage{mathrsfs}
\usepackage{multirow}
\usepackage{enumitem}

\usepackage{soul}
\usepackage{color}
\usepackage{xcolor}
\usepackage[unicode,colorlinks=true,bookmarks=true]{hyperref}
\hypersetup{
	citecolor={green!50!black},
	linkcolor={blue!70!black},
	urlcolor={blue!70!black}}

\newtheorem{theorem}{Theorem}
\newtheorem{definition}[theorem]{Definition}
\newtheorem{proposition}[theorem]{Proposition}
\newtheorem{lemma}[theorem]{Lemma}

\newtheorem{algorithm}[theorem]{Algorithm}

\def\al#1{\begin{align}#1\end{align}}

\newcommand{\Tr}{\mathrm{Tr}}
\newcommand{\Id}{I}
\newcommand{\ket}   [1]{\left|{#1}\right\rangle}
\newcommand{\sket}  [1]{|{#1}\rangle}
\newcommand{\bra}   [1]{\left\langle{#1}\right|}
\newcommand{\sbra}  [1]{\langle{#1}|}
\newcommand{\braket}[2]{\langle{#1}|{#2}\rangle}

\newcommand{\proj}  [1]{\left|{#1}\right\rangle\!\!\left\langle{#1}\right|}

\newcommand{\norm}  [1]{\left|\left|{#1}\right|\right|}
\newcommand{\snorm} [1]{||#1||}
\newcommand{\bnorm} [1]{\big|\big|#1\big|\big|}

\renewcommand{\b}{\textbf{b}}
  
\renewcommand{\u}{\boldsymbol{1}}
  \newcommand{\x}{\textbf{x}}
  \newcommand{\A}{{\mathcal{A}}}
  \newcommand{\B}{{\mathcal{B}}}
  \newcommand{\D}{{\mathcal{D}}}
  \newcommand{\N}{{\mathcal{N}}}
\renewcommand{\O}{{\mathcal{O}}}
\renewcommand{\P}{{\mathcal{P}}}
\renewcommand{\S}{{\mathcal{S}}}
  \newcommand{\T}{{\mathcal{T}}}
\renewcommand{\U}{{\mathcal{U}}}

  \newcommand{\hil} {\mathcal{H}}
\renewcommand{\t}[1]{\mathrm{#1}}

\newcommand{\defin}[1]{\hyperref[def:#1]{Definition~\ref*{def:#1}}}
  \newcommand{\fig}[1]{\hyperref[fig:#1]{Figure~\ref*{fig:#1}}}
\renewcommand{\sec}[1]{\hyperref[sec:#1]{Section~\ref*{sec:#1}}}
  \newcommand{\thm}[1]{\hyperref[thm:#1]{Theorem~\ref*{thm:#1}}}
  \newcommand{\prop}[1]{\hyperref[prop:#1]{Proposition~\ref*{prop:#1}}}
  \newcommand{\lem}[1]{\hyperref[lem:#1]{Lemma~\ref*{lem:#1}}}
  \newcommand{\app}[1]{\hyperref[app:#1]{Appendix~\ref*{app:#1}}}
  \newcommand{\alg}[1]{\hyperref[alg:#1]{Algorithm~\ref*{alg:#1}}}

\makeatletter
\g@addto@macro\bfseries{\boldmath}
\makeatother

\begin{document}

\author{
	Davide Orsucci
}

\affiliation{
	Institut f\"ur Kommunikation und Navigation, 
	Deutsches Zentrum f\"ur Luft- und Raumfahrt (DLR), 
	M\"unchener Str. 20,
	82234 Weßling, 
	Germany
}

\author{
	Vedran Dunjko
}

\affiliation{
	Leiden University,  
	Niels Bohrweg 1,  
	2333 CA Leiden,  
	The Netherlands
}

\title{On solving classes of positive-definite quantum linear systems with quadratically improved runtime in the condition number}

\date{November 1, 2021}

\maketitle

\begin{abstract}
Quantum algorithms for solving the Quantum Linear System (QLS) problem are among the most investigated quantum algorithms of recent times, with potential applications including the solution of computationally intractable differential equations and speed-ups in machine learning. A fundamental parameter governing the efficiency of QLS solvers is $\kappa$, the condition number of the coefficient matrix $A$, as it has been known since the inception of the QLS problem that for worst-case instances the runtime scales at least linearly in $\kappa$~\cite{HHL}. However, for the case of positive-definite matrices classical algorithms can solve linear systems with a runtime scaling as $\sqrt{\kappa}$, a quadratic improvement compared to the the indefinite case. It is then natural to ask whether QLS solvers may hold an analogous improvement. In this work we answer the question in the negative, showing that solving a QLS entails a runtime linear in $\kappa$ also when $A$ is positive definite. We then identify broad classes of positive-definite QLS where this lower bound can be circumvented and present two new quantum algorithms featuring a quadratic speed-up in $\kappa$: the first is based on efficiently implementing a matrix-block-encoding of $A^{-1}$, the second constructs a decomposition of the form $A = L L^\dag$ to precondition the system. These methods are widely applicable and both allow to efficiently solve \textsf{BQP}-complete problems.
\end{abstract}

\section{Introduction}
\label{sec:intro}

Quantum computation is described using the formalism of linear algebra, suggesting that quantum methods may be intrinsically well-suited to perform linear algebraic tasks. Algorithms solving linear systems of equations, in particular, are at the cornerstone of linear algebra~\cite[Chapter~2]{Press07}, having many direct applications and playing a pivotal role in several computational methods~\cite{Saad03}. In the seminal work of Harrow, Hassadim, and Lloyd (HHL) the so-called Quantum Linear System (QLS) problem was introduced and a quantum algorithm was presented that allows solving the QLS exponentially faster than classical algorithms solving classical linear systems~\cite{HHL}. In subsequent works, several new algorithms have been put forward that solve QLS with further increased efficiency in comparison to the original HHL algorithm, improving the runtime dependence on the condition number~\cite{Ambainis10}, on the precision~\cite{Childs15} and on the sparsity~\cite{Wossnig18}. Recently, a new approach inspired by adiabatic quantum computation has introduced a significantly simpler quasi-optimal solving algorithm~\cite{Subasi18,An19}, significantly narrowing the gap with experimental implementations~\cite{Wen19}, making the algorithm compatible with Near-term Intermediate Scale Quantum (NISQ) devices~\cite{Bravo19, Huang19} and leading to the development of the presently most efficient QLS solvers~\cite{Lin19}.

A key idea underpinning the possibility of achieving large quantum speed-ups in linear algebra tasks is the fact that an exponentially large complex vector can be compactly encoded in the amplitudes of a pure quantum state; e.g., a $n$-qubit state is described via $2^n$ amplitudes. This intuition is indeed correct for the QLS problem, which has been proven to be \textsf{BQP}-complete~\cite{HHL}: any quantum computation can be re-formulated as a QLS with only a polynomial overhead and therefore there exist families of QLS problems that afford super-polynomial speed-ups compared to classical solution methods (unless $\textsf{BPP} = \textsf{BQP}$\footnote{$\textsf{BPP}$ is the class of decision problems that can be solved in bounded-error probabilistic polynomial time, $\textsf{BQP}$ are those that can be solved with bounded-error polynomial time quantum computations. Loosely speaking, $\textsf{BPP} = \textsf{BQP}$ would mean that the power of quantum computers is equal to that of classical computers.}). While this reduction shows that almost certainly there exist families of QLS problems that allow an exponential speed-up compared to all classical methods, the crucial question is whether there are \emph{natural} problems that can be directly formulated and solved as QLS. The ubiquity of linear systems seems to suggests their quantum variant should be broadly applicable as well, but still it is not guaranteed, since in the QLS setting further significant constraints have to be met to obtain exponential speed-ups~\cite{Aaronson15}.

A prominent fundamental bottleneck of QLS solvers is that, to efficiently obtain the solution, it is not sufficient that the coefficient matrix $A$ of the system $A\x = \b$ is invertible, but it also have to be \emph{robustly} invertible, that is \emph{well-conditioned}: it is required that the \emph{condition number} of the matrix $A$, defined as the ratio between the largest and the smallest singular values, is small. In fact, solving a QLS necessarily entails a runtime scaling at least linearly in the condition number (unless $\textsf{BQP} = \textsf{PSPACE}$\footnote{$\textsf{PSPACE}$ is the class of decision problems that can be solved in classical polynomial space. With a Feymann sum-over-paths approach one can show that any quantum computation can be classically simulated in exponential time but with only polynomial space, thus $\textsf{BQP} \subseteq \textsf{PSPACE}$. It is widely believed that $\textsf{BQP} \neq \textsf{PSPACE}$.}) as was already proven in Ref.~\cite{HHL}. Therefore, an exponential quantum speed-up for QLS solving is achievable only when the condition number scales polylogarithmically with the system size and, unfortunately, it seems rather difficult to find natural examples of matrix families that exhibit such mild growth of the condition number~\cite{Montanaro16}. However, polynomial quantum speed-ups for linear system solving should be rather broadly achievable and could still provide a quantum advantage if the degree of the polynomial is large enough~\cite{Babbush20}. In this view, obtaining a further quadratic improvement in the dependence on the condition number for some restricted classes of matrices, that are however of wide practical interest, could be of the uttermost importance for obtaining a broader impact of QLS solving algorithms. A previous publication showed that in the context of quantum algorithms for solving certain Markov chain problems the use of specialised QLS solvers for positive-definite matrices provides an improvement in the condition number dependence~\cite{Chowdhury16}. Exploring the general positive-definite QLS problem is the main focus of this work.

In the rest of the Introduction we motivate why quadratic speed-ups in the condition number in positive-definite QLS may be expected (\sec{quest}), review some related results present in the literature (\sec{literature}) and then proceed to give high-level overviews of our two algorithms (\sec{overview1} and \sec{overview2}). In \sec{notation} we fix the notation and give the main definitions. In \sec{lower_bound} we prove that QLS solvers require, even when restricting to positive-definite matrices, a runtime scaling linearly in the condition number. We then move to the main results of this work, that is, achieving improved runtime scaling for solving certain classes of positive-definite QLS: in \sec{block_encoding} we show a method based on an efficient implementation of a matrix-block-encoding of $A^{-1}$ and in \sec{hamiltonians} a method based on decomposing the coefficient matrix as $A = LL^\dag$ to effectively precondition the system. Finally, an outlook of possible future research directions is given in \sec{discussion}.

\subsection{Positive definite linear systems and quadratic speed-ups}
\label{sec:quest}

In this work, we investigate the efficiency of QLS solving algorithms specialised to the case where the coefficient matrix $A \in \mathds{C}^{N \times N}$ is Hermitian and positive definite (PD). We will call this sub-class the positive-definite Quantum Linear System (PD-QLS) problem.

A first reason to focus specifically on the PD case is that in the classical setting several problems of practical relevance are formulated as linear systems involving PD coefficient matrices. A second important reason is that the few fully worked-out examples in the literature providing strong evidence of polynomial quantum speed-ups for ``natural'' QLS problems involve PD coefficient matrices~\cite{Montanaro16, Clader13}. In particular, the discretization of a partial differential equation (PDE) having a positive-definite kernel (such as, e.g., Poisson's equation) results in a large linear system where the coefficient matrix is positive definite. Montanaro and Pallister perform in Ref.~\cite{Montanaro16} a detailed analysis of how QLS may be employed to solve the finite-element formulation of a PDE and show that the linear dependence on the condition number is the main bottleneck to obtaining large quantum speed-ups. In fact, the discretization of PDEs for functions defined in $\mathds{R}^D$ typically results in positive-definite linear systems where the condition number scales as $\O\!\left(N^{2/D}\right)$~\cite{Montanaro16}.

We highlight that one might have reasonably conjectured that it is possible to have a quadratically better scaling in $\kappa$, the condition number of $A$, in the PD case. First, note that the runtime lower bound of Ref.~\cite{HHL} is proven using a special family of matrices that are indefinite (neither positive nor negative definite) by construction, hence it is not directly applicable to the PD case. A standard method allows to transform an indefinite linear system into a PD one, but having a quadratically larger condition number\footnote{Namely, for a given indefinite matrix $A$, the systems $A \x = \b$ and $A' \x = \b'$, with $\b' := A^\dag \b$ and $A' := A^\dag A$, have the same solution. The matrix $A'$ is positive definite and $\kappa(A') = \kappa(A)^2$.}; hence, the lower bound in Ref.~\cite{HHL} directly yields a $\sqrt{\kappa}$ runtime lower bound for PD-QLS solvers. Second, the conjugate gradient (CG) descent method is the most efficient classical algorithm for solving PD linear systems and requires only $\O\big(\sqrt{\kappa} \log(1/\varepsilon)\big)$ iterations to converge to the correct solution, up to $\varepsilon$ approximation. Each iteration consists in the update of some vectors having $N$ entries, where $N$ is the dimension of the linear system, thus resulting in a total runtime in $\O\big(N\sqrt{\kappa}\log(1/\varepsilon)\big)$~\cite{She94}. Then, it might seem plausible that a quantization of the CG method could yield a quantum algorithm having runtime in $\O\big(\t{polylog}(N)\sqrt{\kappa}\log(1/\varepsilon)\big)$.

This conjectured quadratic speed-up is however not always achievable since, as we prove in \sec{lower_bound}, PD-QLS solvers have runtime scaling linearly in $\kappa$ in worst-case problem instances. But this no-go result can be used as guidance to understand what is preventing us from achieving a better runtime scaling and, conversely, what additional conditions have to be imposed in order to achieve a quadratic speed-up in $\kappa$ for PD-QLS solvers.

\begin{table}[t!]
\hspace{-8mm}
{\renewcommand{\arraystretch}{1.5}
\begin{tabular}{p{25mm} p{55mm} p{69mm}} 
	\hline\hline
	\textbf{ Method} & \textbf{ Result} & \textbf{ Requirements}  \\ 
	\hline\hline
	\begin{tabular}{c}
	Reduction to \\[-7pt] majority problem
	\end{tabular}
	& 
	\begin{tabular}{l}
	$\Omega(\kappa)$ query complexity lower bound \\[-7pt] \prop{lower_bound}
	\end{tabular}
	& 
	\begin{tabular}{l}
	Access to $A$ via a matrix-block-encoding\\[-7pt] or via a sparse-oracle access
	\end{tabular}
 	\\
	\hline 
	\begin{tabular}{c}
	Block-encoding \\[-7pt] of $A^{-1}$
	\end{tabular}
	&
	\begin{tabular}{l}
	$\O(\sqrt{\kappa})$ query and gate complexity \\[-7pt] 
	\prop{Algorithm1} and \prop{VTAA}
	\end{tabular}
	&
	\begin{tabular}{l}
	Access to normalised matrix-block-encoding \\[-7pt]
	of $B = (\Id - \eta A)$ and $\norm{A^{-1}\ket{\b}}$ is large
	\end{tabular}
	\\
	\hline
 	\begin{tabular}{c}
	Decomposition  \\[-7pt] as $A = LL^\dag$ 
	\end{tabular}
	&
	\begin{tabular}{l}
	$\O(\sqrt{\kappa})$ gate complexity \\[-7pt] \prop{Algorithm2}
	\end{tabular}
	& 
	\begin{tabular}{l}
	$A$ is the sum of PD local Hamiltonians, \\[-7pt] 
	$\b$ is sparse and $1/\gamma$ in Eq.~\eqref{eq:gamma_main} is small
	\end{tabular} \\
	\hline\hline
\end{tabular}
}
\caption{Summary of the main results of this work. The results provided in the second column of the table hold under the conditions specified in the third column. The big-$\Omega$ notation is used for runtime lower bounds (in query complexity), the big-$\O$ notation for runtime upper bounds (exhibiting an explicit solving algorithm).
}
\label{table:1}
\end{table}

\subsection{Previous related results}
\label{sec:literature}

In context of general QLS solvers, i.e.\ solvers applicable also to indefinite or non-Hermitian matrices, the best algorithms have a runtime scaling quasi-linearly in $\kappa$, i.e.\ scaling as $\O\big(\kappa\, \t{polylog} (\kappa)\big)$, thus almost saturating the linear lower bound~\cite{HHL}. Note that the original HHL algorithm has a worse performance, with a runtime scaling as $\O(\kappa^2 / \varepsilon)$ where $\varepsilon$ is the target precision. The first algorithm to achieve a quasi-linear scaling in $\kappa$ was proposed by Ambainis in Ref.~\cite{Ambainis10}, which introduces a technique called Variable-Time Amplitude Amplification (VTAA) and employs it to optimize to the HHL algorithm. Subsequently, Childs, Kothari and Somma~\cite{Childs15} introduced polynomial approximations of $A^{-1}$ to exponentially improve the runtime dependence on the approximation error to $\O\big(\kappa^2\, \t{polylog}(\kappa/\varepsilon) \big)$; they show, furthermore, that the VTAA-based optimization can be used also for this algorithm, thus yielding a $\O\big(\kappa\, \t{polylog}(\kappa/\varepsilon) \big)$ runtime. Later, Chakraborty et al.\ showed that also the pseudo-inversion problem, whereby the matrix $A$ may be non-invertible and even non-square, can be solved with a runtime in $\O\big(\kappa\, \sqrt{\gamma} \,\t{polylog}(\kappa/\varepsilon) \big)$, where $\gamma$ parametrises the overlap of $\b$ with the subspace where $A$ is non-singular~\cite{Chakraborty18}. Finally, the current state-of-the-art methods for general QLS solving is given in Ref.~\cite{Lin19}, which does not rely on VTAA but instead is based on ideas stemming from adiabatic quantum computation~\cite{Subasi18}, which result in conceptually simpler algorithms and in a significant improvement of the polylogarithmic factors.

Furthermore, several specialized quantum algorithms have been introduced with the scope of more efficiently solving QLS for particular sub-classes of matrices. A few works, e.g.\ \cite{Clader13,Shao18,Tong20}, have investigated the use of preconditioning to speed-up the solution. The main idea, which is well-known in classical linear system solving methods, is to look for an invertible matrix $B$, a so-called \emph{preconditioner}, such that the matrix $BA$ has a smaller condition number than $A$, and subsequently solve the equivalent linear system $BA = B\,\b$. An algorithm based on a sparse preconditioning matrix was introduced in Ref.~\cite{Clader13}, but it has very little formal guarantees of performance improvement. Another method based on circulant preconditioners was presented in Ref.~\cite{Shao18}, for which it is clearer how to assess when a runtime improvement can be achieved. Runtime improvements have been obtained in Ref.~\cite{Tong20} applying new preconditioning methods to Hamiltonians arising in many-body physics. An entirely different approach, based on hybrid classical-quantum algorithms, has been explored in Ref.~\cite{Wu20}, which yields runtime speed-ups for the case where the rows or columns of the coefficient matrix can be prepared with polylogarithmic-depth quantum circuits. We also mention the result of Ref.~\cite{Vazquez20}, showing that it is possible to solve QLS for the special class of tridiagonal Toeplitz matrix with a runtime that scales polylogarithmically in all parameters, condition number included; note, however, that matrices of this class can be fully specified with just two real parameters. Finally, the PD-QLS problem has been previously considered in Ref.~\cite{Chowdhury16}, where it is suggested that positive-definite systems could be solved more efficiently than indefinite ones.

\subsection{Overview of the method based on a matrix-block-encoding of \texorpdfstring{$A^{-1}$}{A-1}}
\label{sec:overview1}

Our first method for solving PD-QLS with improved runtime is based on implementing as a quantum circuit a unitary matrix $\U_{A^{-1}}$ that encodes in a sub-block a matrix proportional to $A^{-1}$, i.e.\ a so-called matrix-block-encoding~\cite{Low17,Gilyen18}. The solution state $\ket{A^{-1}\b}$ can be subsequently obtained via matrix-vector multiplication, achieved by applying the circuit encoding $A^{-1}$ and projecting onto the correct sub-block. This method is analogous to the one introduced by Childs, Kothari and Somma in Ref.~\cite{Childs15}, where it is shown that exponentially precise polynomial approximations of the inverse function can be constructed, which then allow to implement a matrix-block-encoding of $A^{-1}$ up to exponentially small error and finally solve the QLS problem via matrix-vector multiplication. Here, we show that if $A$ is a PD matrix, an equally good approximation of $A^{-1}$ can be obtained with polynomials having a quadratically smaller degree, leading to the possibility of a quadratic speed-up in PD-QLS solving.

More in detail, Ref.~\cite{Childs15} considers a Hermitian matrix $A$ whose spectrum is by assumption contained in the domain $\mathcal{D}_{\Delta} := [-1,-\Delta] \cup [+\Delta,+1]$, with $\Delta \leq 1/\kappa$. Then, families of real polynomials $P(x)$ are constructed such that $|P(x) - 1/x| \leq \varepsilon$ on the domain $\mathcal{D}_\Delta$ and have a degree $\ell \in \O\big(\kappa \log(1/\varepsilon)\big)$, see panel $(a)$ of \fig{two_plots} for an illustrative example. Using either the Linear Combination of Unitaries (LCU)~\cite{Childs12,Berry14} or the Quantum Signal Processing (QSP)~\cite{Low16,Low17,Gilyen18} method, it is possible to implement a matrix-block-encoding of $P(A)$, up to some rescaling factor $K>0$ such that $P(A)/K$ can be encoded a sub-block of a unitary matrix, and then we have
$
	\norm{
	P(A) - A^{-1}
	}
	\, \leq \, 
	\varepsilon
$
in operator norm. The query complexities of LCU and QSP scale at least linearly in the degree of the polynomial, hence the use of polynomials of low degree is crucial to construct efficient QLS solvers. Moreover, the normalisation factor $K$ of the matrix-block-encoding of $P(A)$ scales linearly in $\kappa$ and enters multiplicatively in the cost of the matrix-vector multiplication step necessary to produce $\ket{A^{-1}\b}$.

In case $A$ is a PD matrix, we can exploit the knowledge that its spectrum is contained in $\mathcal{D}_\Delta^+ := [\Delta,1]$ to perform the following trick: we assume to have access to a \emph{normalised} matrix-block-encoding of $B := \Id - \eta\,A$, for some constant\footnote{We show in \sec{D} that $B$ can be efficiently constructed for $\eta = 1$ when $A$ is diagonally dominant and for $\eta = 1/J$ when $A$ is the sum of $J$ positive semi-definite local Hamiltonian terms.} $\eta \in (0,2]$, so that the spectrum of $B$ is always contained in the interval $\mathcal{D}_B = [-1,1- \eta \Delta]$. We then construct a polynomial $P(x)$ that approximates the function $f(x):=1/(1-x)$ on the domain $\D_B$ up to $\varepsilon$ distance. Using the QSP method to implement a matrix-block-encoding of the matrix $P(B)$ we then have
\al{
	P(B) \approx
	f(B) =
	\frac{1}{I - B} =
	\frac{1}{I - (I-\eta\,A)} =
	\frac{1}{\eta} A^{-1}
}
which is the required matrix inverse, up to a $1/\eta$ rescaling factor. Importantly, our polynomial approximation of $1/(1-x)$ has a degree $\ell \in \O\big(\sqrt{\kappa/\eta}\, \log (\kappa / (\eta\varepsilon)) \big)$, a quadratically better dependence on $\kappa$ with respect to the approximation of $1/x$ given in Ref.~\cite{Childs15}, see panel $(b)$ of \fig{two_plots}. Moreover, the normalisation factor $K$ scales again linearly in $\kappa$, which is the best dependence achievable.

\begin{figure}[t!]
\hspace{-5mm}
\begin{tabular}{ll}
$(a)$ 
& \hspace{-8mm}
$(b)$ \\
\includegraphics[scale=.27]{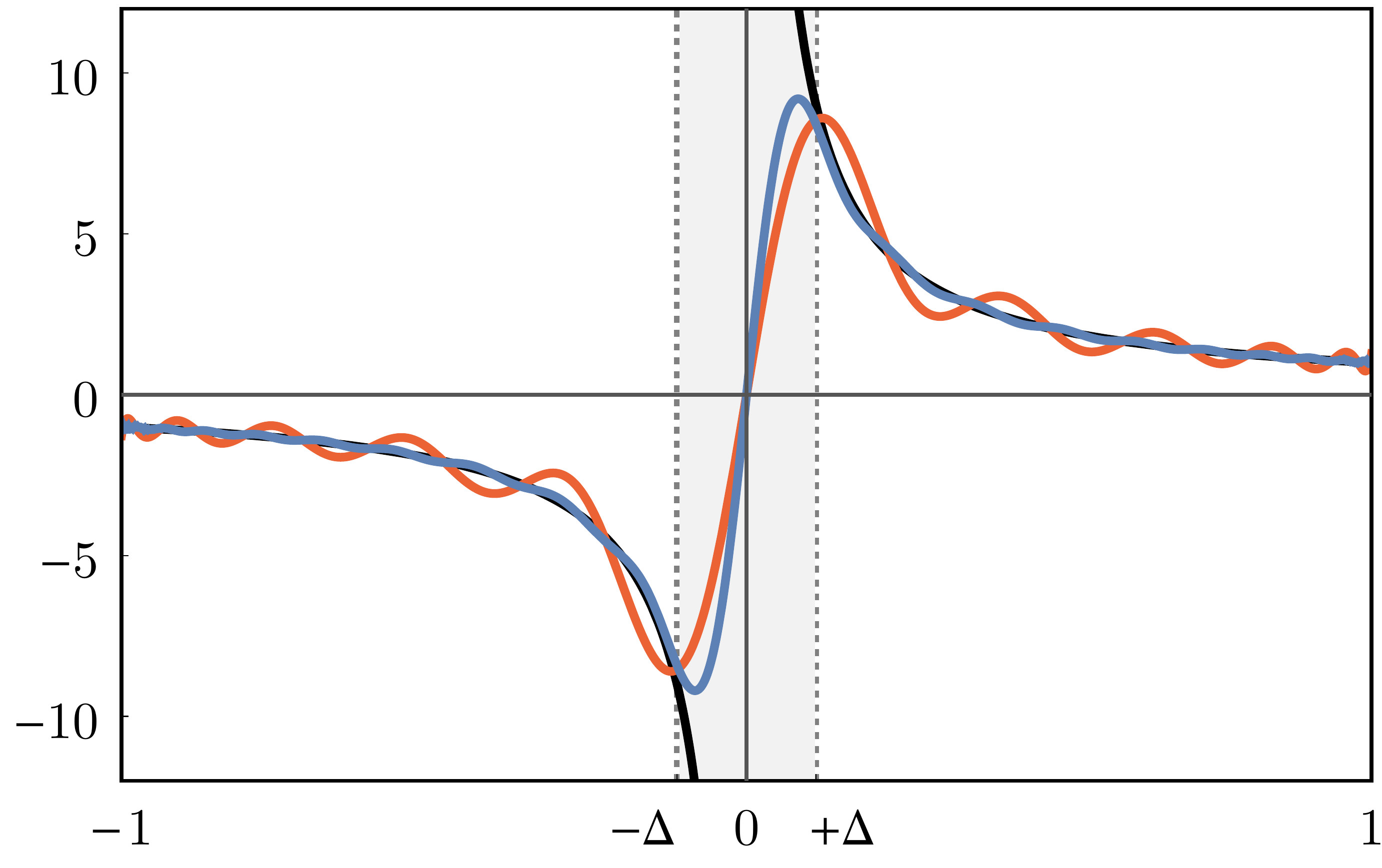}  \hspace{4mm}
& \hspace{-10mm}
\includegraphics[scale=.27]{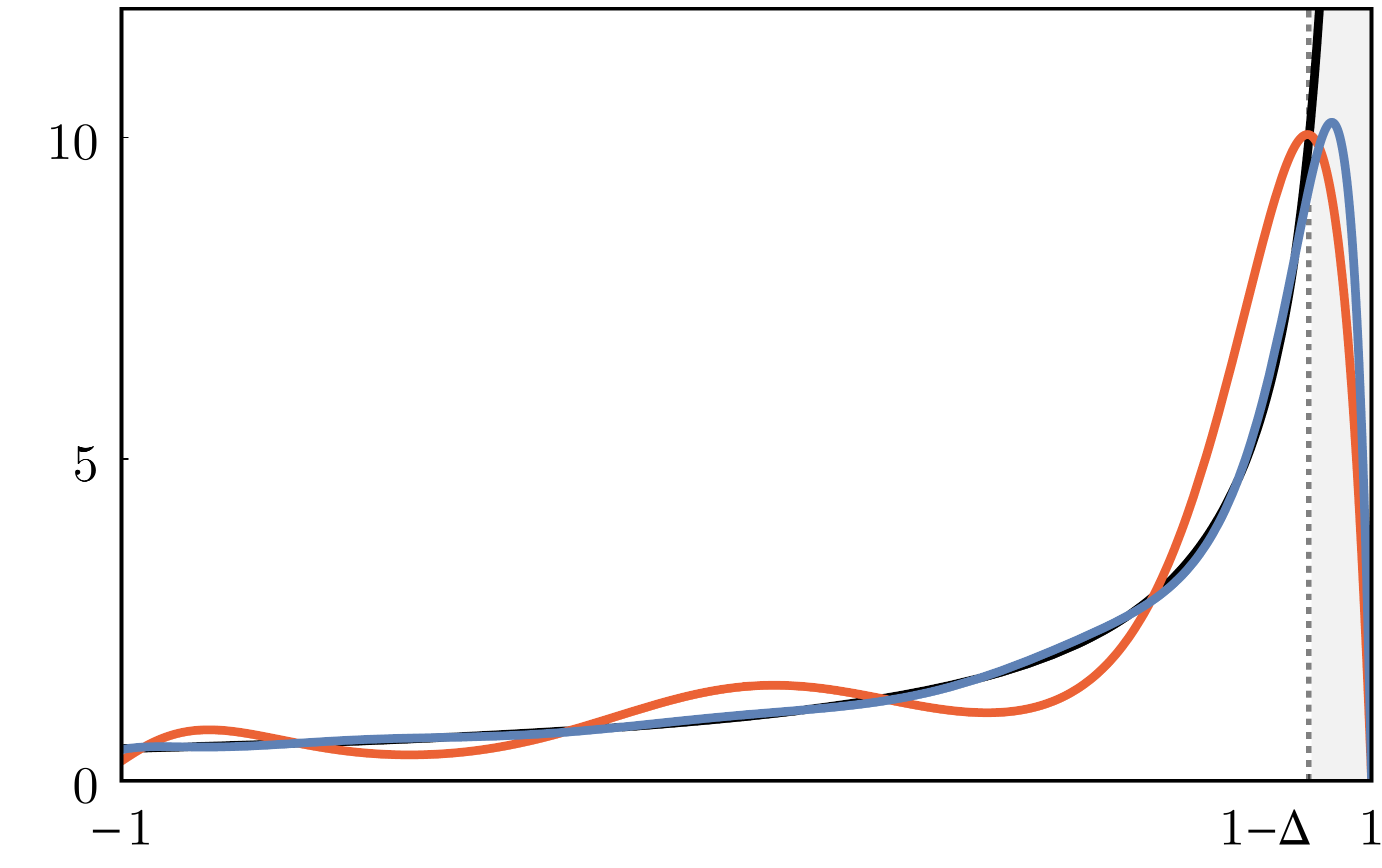}
\end{tabular}
\caption{In panel $(a)$ are shown two polynomials approximating the function $1/x$ (in black) chosen from the family of polynomials given in Ref.~\cite[Lemma~14]{Childs15}, using as parameters $b=200,j_0=10$ (red curve) and $b=200,j_0=20$ (blue curve curve). In panel $(b)$ are shown two polynomials approximating the function $1/(1-x)$ (in black) chosen from the family given in Eq.~\eqref{eq:polynomial}, using as parameters $\ell=6,\kappa=15$ (red curve) and $\ell=10,\kappa=9$ (blue curve). The shaded regions in each panel indicate the intervals where the polynomial approximation is not accurate.}
\label{fig:two_plots}
\end{figure}

From a mathematical standpoint, the possibility of a quadratic reduction of the degree of the approximating polynomial can be interpreted as a consequence of Bernstein's inequality~\cite{Schaeffer41}. This inequality states that in the class of real polynomials $P(x)$ of degree $\ell$ such that $|P(x)| \leq 1$ for all $x \in [-1,+1]$ the derivative $P'(x)$ satisfies $|P'(x)| \leq (1-x^2)^{-1/2} \ell$ for all $-1 < x < 1$, while we have $|P'(x)| \leq \ell^2$ for $x=1$ and $x=-1$ (and these last bounds are saturated by Chebyschev polynomials). Note that polynomials approximations of $1/x$ and of $1/(x-1)$ have high derivative close to the singularities, respectively in $x=0$ and in $x=1$, and because of Bernstein's inequality the latter case allows good polynomial approximations having a quadratically lower degree.

We need next to perform matrix-vector multiplication to obtain the state $\ket{A^{-1}\b}$ and thus solve the QLS; this is obtained by applying the quantum circuit that encodes the matrix $P(B) \propto A^{-1}$ to a quantum state of the form $\ket{0}\ket{\b}$ and then post-selecting the ancilla register to be in $\ket{0}$ or, more efficiently, using amplitude amplification~\cite{Brassard02}. The amplitude amplification step implies a multiplicative overhead of order $1/\kappa$ in the worst case, yielding a total runtime in $\O\big(\kappa^{3/2} \log ( \kappa/\varepsilon )\big)$ for this PD-QLS solver. However, the efficiency the matrix-vector multiplication depends on the post-selection success probability and thus on the specific choice for the vector $\b$. In a best-case scenario the post-selection success probability is constant and thus the overall runtime of the PD-QLS solver is $\O\big(\sqrt{\kappa} \log ( \kappa/\varepsilon )\big)$, while in the same high-success-probability scenario the QLS solver of Ref.~\cite{Childs15} has a runtime in $\O\big(\kappa \log (\kappa/\varepsilon)\big)$, since this is the cost of implementing a polynomial approximation of $A^{-1}$ for indefinite matrices.

We remark that there is a technical assumption that has to be satisfied to allow the realisation of our method, namely, that it is possible to implement a \emph{normalised} matrix block encoding of $B= \Id -\eta\,A$. This is a non-trivial task: standard methods could allow us to prepare, for instance, a normalised matrix-block-encoding of $B/d$ where $d \geq 2$ is the \emph{sparsity} of $B$, i.e.\ the maximum number of non-zero entries in each column of $B$~\cite{Childs15}. Note, however, that we would then need to implement an approximation of the function $g(x) := 1/(1 - x d)$ to obtain $g(B/d) = \frac{1}{\eta} A^{-1}$; the function $g(x)$ has a singularity in $x = 1/d$, which is in the interior of the domain $[-1,+1]$, and thus Bernstein's inequality precludes us from achieving good approximations with low-degree polynomials for this function.

We will show that it is possible to implement normalised matrix-block encodings of $B$ for two special classes of QLS problems: the first is when $A$ is a diagonally dominant matrix; the second, when $A$ is given as the sum of local PD Hamiltonian terms, where by ``local'' we mean that it acts non-trivially only on a small number of qubits. In these two cases it is therefore possible to implement our improved PD-QLS solver. We leave the question whether it is possible in broader generality to prepare normalised block-encodings as an open research question.

\subsection{Overview of the method based on the \texorpdfstring{$A = LL^\dag$}{A=LL†} decomposition}
\label{sec:overview2}

Our second method for solving PD-QLS with improved runtime is based on finding a decomposition $A = L L^\dag$, akin to the classical Cholesky decomposition~\cite{Cholesky} which exists for all PD matrices, and then use $L^{-1}$ as an efficient and effective preconditioner. Note, in fact, that $L^\dag \x = \b'$ for $\b' = L^{-1} \b$ is a linear system equivalent to the original one, but the decomposition $A = L L^\dag$ immediately implies $\kappa(L) = \kappa(L^\dag) =\sqrt{\kappa(A)}$ and thus the new system provably has a quadratically smaller condition number. This decomposition is similar to a \emph{spectral gap amplification} of $A$~\cite{Somma13}.

The method involves thus two main steps: in the first one we use classical computation to efficiently obtain a description of a matrix $L$ such that $LL^\dag = A$ and such that it is possible to efficiently find, using only classical computation, a description of the vector $\ket{\b'} := \ket{L^{-1}\b}$; in the second one, we use an efficient quantum algorithm, having runtime quasi-linear in $\kappa$, to solve $\ket{L^\dag \x} = \ket{\b'}$, which thus gives $\ket{\x} = \ket{(L^\dag)^{-1} L^{-1}\b} = \ket{A^{-1}\b}$. Note that the classical descriptions of $L^\dag$ and of $\ket{\b'}$ should also allow the efficient compilation the quantum algorithm used to solve the preconditioned QLS.

The picture is not yet complete, since we actually use a matrix $L$ that is non-square and thus singular. As a result, the inversion operation has to be substituted by a pseudo-inversion and the condition number by the \emph{effective} condition number, the ratio between the largest and the smallest \emph{non-zero} singular value; when $A$ is invertible the effective condition number of $L$ is quadratically smaller than the condition number of $A$. We also use two different pseudo-inverses in the classical and in the quantum part of the computation: in the quantum step the Moore-Penrose pseudo-inverse $(L^\dag)^+$ is employed and in the classical preconditioning a \emph{generalised} pseudo-inverse $L^g$ chosen such that $(L^\dag)^+ L^g = A^{-1}$; thus, they yield the desired solution $\ket{\x} = \ket{A^{-1}\b}$ when applied in sequence. These extensions to non-square matrices and to different pseudo-inverses are made to give leeway in the design of the classical part of the algorithm, allowing us to meet the efficiency requirements mentioned above. Finally, we will employ the QLS solver of~\cite{Chakraborty18}, which can tackle pseudo-inversion problems and has a runtime quasi-linear in the effective condition number.

We show that a fully suitable decomposition of the form $A=LL^\dag$ can be constructed for the Sum-QLS problem, i.e., when $A$ is a sum of local PD Hamiltonian terms. In this case, the matrix $L$ is formed by several blocks, each constructed from a single Hamiltonian term, while the pseudo-inverse $L^g$ is obtained inverting the individual blocks, operations involving only small matrices and thus classically feasible. We also require that the vector $\b$ is sparse, containing only polynomially many non-zero entries, thus allowing to efficiently compute the description of $\ket{\b'} = \ket{L^g\,\b}$.

As a final technical caveat, we note that because of the mismatch between the pseudo-inverse used in the classical and in the quantum part of the algorithm ($L^g$ and $(L^\dag)^+$), the vector $\b'$ is not entirely contained in the support of $(L^\dag)^+$ and thus amplitude amplification of the component in the correct subspace is required. This incurs in a multiplicative overhead of a factor $1/\sqrt{\gamma}$, where $\sqrt{\gamma} > 0$ is a known bound for the amplitude of the ``good'' component of $\b'$. This method thus has a provable runtime improvement over competing QLS solvers whenever $1/\sqrt{\gamma}$ is sufficiently small.

\section{Notation and definitions}
\label{sec:notation}

We assume knowledge of the main quantum computation concepts, as given for instance in Ref.~\cite{NC00}. A quantum computation is described using a Hilbert space of dimension $2^n$ for some $n$, corresponds to a system of $n$ qubits, having a distinguished computational basis.

\subsection{Linear algebra and asymptotic notation}
\label{sec:linear_algebra}

We consistently denote with $N \in \mathds{N}$ the dimension of the QLS we aim to solve and we define $n := \lceil \log_2 N\rceil$, so that a vector in $\mathds{C}^N$ (possibly padded with zeroes at the end if $N$ is not a power of two) can be described as pure state of $n$ qubits. For any complex matrix $A \in \mathds{C}^{N\times M}$ having $N$ rows and $M$ columns we write its Singular Value Decomposition (SVD) as $A = V \Sigma W^\dag$ where $V$ and $W$ are unitary matrices of size $N\times N$ and $M\times M$ respectively and $\Sigma$ is a real non-negative matrix of size $N \times M$ which is uniquely determined up to reordering of the diagonal entries and contains the singular values of $A$ on the main diagonal. An Hermitian matrix $A$ is positive definite if $\bra{\textbf{v}} A \ket{\textbf{v}} > 0$ for all $\ket{\textbf{v}}$ and is is positive semi-definite if $\bra{\textbf{v}} A \ket{\textbf{v}} \geq 0$. The eigenvalues of $A$ are real, positive (in the definite case) or non-negative (in the semi-definite case) and coincide with its singular values. For a general $A$, $\varsigma_{\min},\varsigma_{\max}$ and $\lambda_{\min},\lambda_{\max}$ denote the minimum and maximum singular values and eigenvalues, respectively. The Moore-Penrose pseudo-inverse of $A$ is obtained by applying to the singular values $\varsigma_i$ of $A$ the function $f: \mathds{R} \mapsto \mathds{R}$ defined as $f(x)=1/x$ for $x\neq 0$ and $f(0)=0$. More precisely, for a diagonal matrix $\Sigma$ we define $\Sigma^+ = f(\Sigma)$, while for a general matrix $A = V \Sigma W^\dag$ the pseudo-inverse is given as $A^+ := W \Sigma^+ V^\dag$. Given $A \in \mathds{C}^{N \times M}$ a \emph{generalised} pseudo-inverse $A^g \in \mathds{C}^{M \times N}$ is any matrix satisfying the equation $A \,A^g A = A$.

In this work we employ the $\ell^2$-norm for vectors $\norm{\textbf{v}}^2 := \sum_{i=1}^N \text{v}_i^2$ and the induced operator norm for matrices $\norm{A} := \max_{\textbf{v} \neq 0} \norm{A\textbf{v}}/\norm{\textbf{v}}$. The condition number of a matrix is given by $\kappa(A):=\norm{A}||A^{-1}||$. Since we have $\norm{A} = \varsigma_{\max}(A)$ and $||A^{-1}|| = 1/\varsigma_{\min}(A)$, the condition number can be also written as $\kappa(A) = \frac{\varsigma_{\max}(A)}{\varsigma_{\min}(A)}$. For a singular matrix $A$ we define the \emph{effective} condition number to be $\kappa_\t{eff}(A) := \norm{A}||A^+||$, which is equal to the ratio between the largest and the smallest non-zero singular value of $A$. A Hermitian matrix $A$ is diagonally dominant if $\sum_{j: j\neq i} |A_{ij}| \leq |A_{ii}|$ for all $i$ and note that $|A_{ii}| = A_{ii} > 0$ when $A$ is positive definite.

We use the standard big-$\O$ and small-$\mathfrak{o}$ notations for asymptotic scaling, together with the following definitions: $f(x) \in \Omega(g(x))$ if and only if $g(x) \in \O(f(x))$, which is used to give lower bounds, and $\Theta(g(x)) = \O(g(x)) \,\cap\, \Omega(g(x))$. We also use the soft-$\O$ and soft-$\Omega$ notations where $f(x) \in \widetilde{O}(g(x))$ means $f(x) \in \O\big(g(x)\, \text{polylog}\,[g(x)]\big)$, and similarly for $\widetilde{\Omega}(g(x))$, which are used to give more coarse-grained bounds.

\subsection{Definition of the Quantum Linear System problem}
\label{sec:definitions}

In this section we introduce the main definitions that are relevant in the contest of the QLS problem, which is a quantum analogue of the classical linear algebra problem of solving the system of equations $A\x = \b$, having solution $\x = A^{-1}\b$ when $A$ is invertible.

We use pure quantum states to encode $N$-dimensional complex vectors. A vector $\textbf{v}$ enclosed in a bra or in a ket is always assumed to be normalised, $\bnorm{\ket{\textbf{v}}} = 1$. In particular we have:
\al{
	& \ket{\b} 
	= 
	\frac{\b}{\norm{\b}} 
	= 
	\frac{\sum_{i=1}^N b_i \ket{i}}
	{\left(\sum_{i=1}^N |b_i|^2\right)^{\!1/2}} \\
	& \ket{A^{-1}\b} =
	\frac{A^{-1}\ket{\b}}{\norm{A^{-1}\ket{\b}}} \;.
}

We now give a general definition of the QLS problem. The formulation is similar to the one provided in Ref.~\cite{Subasi18} and is intentionally not specifying the access models employed for the coefficient matrix $A$ and the known-term vector $\b$, for sake of generality.

\begin{definition}[Quantum Linear System]
\label{def:QLS}
Suppose we have access to a vector $\b \in \mathds{C}^N \setminus \{0\}$ and to a non-singular matrix $A \in \mathds{C}^{N \times N}$ (access is given via quantum oracles, or some kind of implicit or explicit description). We are given two real positive parameters $\varsigma_*$ and $\varsigma^*$ such that $\varsigma_* \leq \varsigma_{\min}(A)$ and $\varsigma_{\max}(A) \leq \varsigma^*$, i.e.\ the singular values of $A$ are contained in the interval $\D_A = \big[ \varsigma_*,\varsigma^* \big]$; equivalently, we are given two parameters $\kappa > 1$ and $\alpha > 0$ that provide the upper bounds $\kappa(A) \leq \kappa$ and $\norm{A}\leq \alpha$. We are also given a target precision $\varepsilon > 0$. 

The QLS problem then consists in preparing a density matrix $\rho_\x$ which is $\varepsilon$-close in trace distance to the solution vector $\ket{\x} = \ket{A^{-1}\b}$; that is, we require that
\al{
\label{eq:trace_distance}
	\bnorm{\,\rho_\x - \proj{\hspace{.5pt}\x\,}}_{\Tr}
	\leq \varepsilon
}
where the trace norm is defined as $\norm{X}_{\Tr}:= \frac{1}{2} \Tr\left(\sqrt{XX^\dag}\right)$. 
\end{definition}

\begin{definition}[Positive-definite Quantum Linear System]
\label{def:PD-QLS} 
A PD-QLS problem is a QLS problem, as given in \defin{QLS}, where the coefficient matrix $A$ is Hermitian and positive definite.
\end{definition}

We note that a more commonly employed definition requires that the QLS solver outputs a state $|\widetilde{\x}\rangle$ such that $\bnorm{\,|\widetilde{\x}\rangle - \ket{\x}} \leq \varepsilon$. We prefer to use the trace distance since it is operationally motivated, as it equals the probability of distinguishing two copies of two quantum states when optimizing over all possible measurements, see~\cite[Section 9.2.1]{NC00}. The trace distance then also bounds the maximum relative error that can be introduced when estimating the expectation value $\Tr( \proj{\x} M )$ for a given observable $M$ and computing expectation values was the end goal of Ref.~\cite{HHL}. Finally, for pure normalised states $\ket{\psi}$ and $\ket{\phi}$ the trace distance simplifies as
$
	d_\Tr (\psi,\phi) :=
	\bnorm{\,\proj{\psi} - \proj{\phi}}_\Tr 
	= 
	\sqrt{1 - |\braket{\psi}{\phi}|^2}
$
and using
$
	|\braket{\psi}{\phi}|
	\geq 
	1 - \frac{1}{2} \big|\big| \ket{\psi} - \ket{\phi} \big|\big|^2
$
we obtain the inequality
\al{
\label{eq:trace_ineq}
	d_\Tr (\psi,\phi) 
	\leq 
	\big|\big| \ket{\psi} - \ket{\phi} \big|\big| \,,
}
thus the trace distance  between $|\widetilde{\x}\rangle$ and $\ket{\x}$ is at least as small as their $\ell^2$ distance.

It is customary to assume that $A$ has been rescaled by a factor $\alpha \geq \norm{A}$, so that we take, without loss of generality, $\norm{A} \leq 1$ and the only parameter that needs to be specified is an upper bound to the condition number. This will be also our convention, unless otherwise specified.

\begin{figure}[t!]
\includegraphics[scale=1.4]{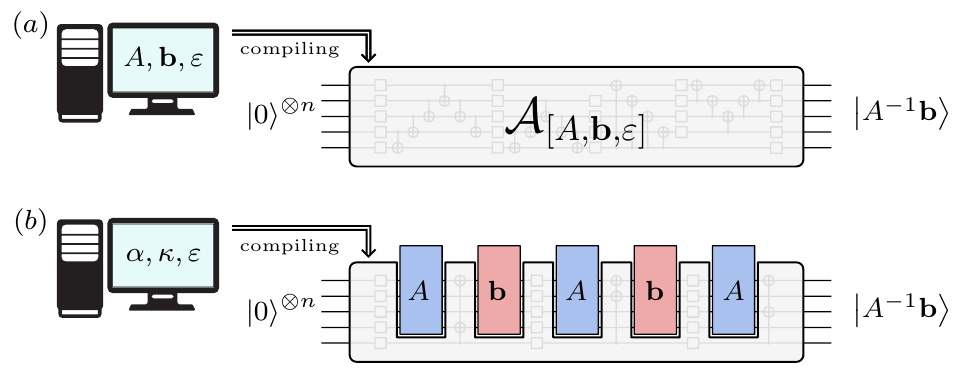}
\caption{Different access models for QLS algorithms. Panel (a) illustrates the case where a classical description of $A$ and $\b$ (and the target precision $\varepsilon$) is provided and then used to compile a quantum algorithm $\mathcal{A}$, solving the QLS for the given $A$ and $\b$. The description of $A$ and $\b$ does not need to be fully explicit: it is sufficient that it allows to efficiently compute the sequence of elementary quantum gates of $\mathcal{A}$. Panel (b) illustrates the case of a \emph{relativising} quantum algorithm; in this case, the sequence of elementary quantum gates only depends on a few parameters (e.g., $\alpha, \kappa, \varepsilon$ as in \defin{QLS}) while two fixed sub-routines specify $A$ and $\b$ and these sub-routines can be treated as black-boxes.
}
\label{fig:access_models}
\end{figure}

\subsection{Oracles for quantum linear system solving}
\label{sec:oracles}

We now define and discuss a few different access models for $A$ and $\b$, since the results we present for the PD-QLS solvers crucially depend on which access model is assumed. A fundamental distinction is between oracular and non-oracular (a.k.a.\ relativising and non-relativising) algorithms, see \fig{access_models}. For instance, in oracular settings it is often possible to establish unconditional query complexity lower bounds, while in non-oracular settings non-trivial runtime lower bounds typically can be proven only under some (reasonable) complexity theory assumptions, such as $\textsf{P} \neq \textsf{NP}$. Note that most of the literature on QLS assumes oracular access to $A$ and $\b$ \cite{HHL, Ambainis10, Childs15, Wossnig18, Subasi18, Lin19, Chakraborty18}.

We start defining the access model for $\b$ that we assume throughout this paper, except where differently specified.

\begin{definition}[State preparation oracle, as in~\cite{HHL}] 
\label{def:state_prep}
Given a vector $\b \in \mathds{C}^N$ we say that we have quantum access to a \emph{state preparation oracle} for $\b$ if there is a unitary operator $\U_\b$ such that $\U_\b \ket{0^n} = \ket{\b'}$, where $\b'$ is obtained by padding $\b$ with zeroes until its size is a power of $2$.
\end{definition}

As already noted by HHL, a general setting where it is possible to efficiently implement $\U_B$ is the one presented by Grover and Rudolph~\cite{Grover02}; another possibility is that $\U_\b$ is directly encoded in a qRAM~\cite{Giovannetti08}, as may be required in quantum machine learning contexts.

Next, we define two models for access to $A$, which we denote as $\P_A$ and $\U_A$ and which correspond to a sparse-matrix-access and matrix-block-encoding, respectively.

\begin{definition}[Sparse-matrix-access, as in~\cite{Childs15}] 
\label{def:sparse_access}
Given a Hermitian matrix\footnote{Since $A$ is Hermitian, access by rows and by columns are equivalent. The definition can be extended to non-Hermitian matrices, but we need to assume access both by rows and by columns.} $A$ which is $d$-sparse (i.e., has at most $d$ non-zero entries in each row and column) a quantum \emph{sparse-matrix-access} $\P_A$ is given by a pair of oracular functions
\al{
	\P_A := (\P_A^{pos}, \P_A^{val}) 
}
where $\P_A^{pos}$ and $\P_A^{val}$ specify the positions of the (potentially) non-zero entries of $A$ and the values of those entries, respectively, i.e.
\al{
	\label{eq:pos}
	\P_A^{pos}: \ket{i,\nu} \ \mapsto \ \ket{i,j(i,\nu)} 
	\hspace{5.5mm}
	& \qquad
	\forall ~ i,j\in \{1,\ldots,N\} 
	~ \t{and} ~ \nu \in \{1,\ldots,d\} \\[2mm]
	\P_A^{val}: \ket{i,j,z} \mapsto \ket{i,j,A_{i,j} \!\oplus\! z} 
	& \qquad
	\forall ~ z \in \{0,1\}^*
}
where $A_{i,j} \!\oplus\! z \in \{0,1\}^*$ denotes a bit string of arbitrary length that encodes the value $A_{i,j} \in \mathds{C}$. 
\end{definition}

In order to keep the presentation simple, we assume here and throughout the paper that numeric representations of complex numbers can be specified exactly or with a sufficiently high number of digits of precision.

\begin{definition}[Matrix block encoding, as in~\cite{Low16,Gilyen18}]
\label{def:U_A}
A unitary operator $\U_A$ acting on $n+a$ qubits is called an $(\alpha, a, \varepsilon)$-matrix-block-encoding of a $n$-qubit operator $A$ if \footnote{The circuit $\U_A$ may also act on other ancilla qubits, which are in $\ket{0}$ both before and after the application of $\U_A$. For a given $(\cdot, a, \cdot)$-matrix-block-encoding we only count the $a$ ancilla qubits that require post-selection to $\bra{0^a}$ to obtain the encoding of $A$.} 
\al{
	\bnorm{
	\, A - \alpha\,  (\bra{0^a} \otimes I) \,\U_A\, (\ket{0^a} \otimes I) 
	\,} 
	\leq 
	\varepsilon
}
which can be expressed also as:
\al{
	\U_A = 
	\left(
	\begin{array}{cc}
	\widetilde{A} /\alpha & * \\
	*                     & *
	\end{array}
	\right)
	\quad \t{with} ~ \norm{\widetilde{A} - A} \leq \varepsilon \,,
}
where the asterisks ($*$) denote arbitrary matrix blocks of appropriate dimensions.

We call $\alpha$ the \emph{normalization factor} of the matrix-block-encoding and we say in the special case where $\alpha = 1$ that $\U_A$ is a \emph{normalised} matrix-block-encoding of $A$.  
\end{definition}

A technique introduced by Childs allows to implement a $(d,1,0)$-matrix-block-encoding of $A$, where the normalisation constant $d$ is equal to the sparsity of $A$, using only a constant number of accesses to $\P_A$ and $\O\big(\t{poly}(n)\big)$ extra elementary gates, see Ref.~\cite{Childs15} and references therein. In short, we have the reduction:
\al{
\label{eq:ChildsReduction}
	& \P_A
	\ \Longrightarrow \  
	\U_{A}
}
where the arrows means that having access to an oracle of first type allows to efficiently implement the oracle of the second type.

We also note that other access models to $A$ have been considered in the literature; for example in Ref.~\cite{Kerenidis16} it is assumed that is possible to efficiently prepare quantum states that are proportional to each one of the columns of $A$.

\subsection{Quantum linear systems in non-oracular settings}
\label{sec:non-oracular}

We consider in \sec{hamiltonians} (and also briefly in \sec{D}) a case that we call the Sum-of-Hamiltonians QLS (Sum-QLS) problem, which is not formulated as an oracular algorithm but is based, instead, on a classical description of $A$ and $\b$. In order to obtain efficient Sum-QLS solving algorithms, it is thus necessary that the descriptions of $A$ and $\b$ are compact, depending at most on $\O\big(\t{poly}(n)\big)$ real or complex parameters.

For the known term vector $\b$, we will simply assume that it is a sparse vector in the computational basis, with at most $\O\big(\t{poly}(n)\big)$ non-zero entries, whose position is also known. Hence, a fully explicit classical description of $\b$ can be provided and this also enables efficiently implementing a state preparation circuit $\U_\b$.

For the coefficient matrix $A$ we give a more implicit description: the entries of $A$ are not specified one-by-one (which would be inefficient, as the matrix size $N \in \Theta(2^n)$ is assumed to be very large) but rather can be computed from only a relatively small set of parameters, scaling polynomially in the number of qubits. Specifically, we assume that $A$ is given as the sum of polynomially many local Hamiltonian terms:
\al{
	A =
	\sum_{j=1}^J
	H_{(j)} 
	\qquad 
	\forall j \ H_{(j)} \ \textup{is positive (semi)-definite,} 
}
where the number of terms is $J\in\O\big(\t{poly}(n)\big)$ and each Hamiltonian $H_{(j)}$ acts on a small number of qubits, namely, on at most $\O\big(\log(n)\big)$ qubits. This case has been previously considered in Ref.~\cite{Chowdhury16}.

\section{Query complexity lower bound}
\label{sec:lower_bound}

In this section we prove a $\Omega(\kappa)$ lower bound on the runtime of QLS which is alternative to the ones given by HHL in Ref.~\cite{HHL}. The main innovation we introduce is that our lower bound applies also to the PD-QLS case, while the proofs given by HHL only yield a $\widetilde{\Omega}(\sqrt{\kappa})$ lower bound when specialised to PD matrices. More precisely, we have the following result.

\begin{proposition}[Query complexity lower bound]
\label{prop:lower_bound}
Consider oracular quantum algorithms that solve the PD-QLS problem as presented in \defin{PD-QLS} for different access models to $A$ and $\b$. Namely, access to $\b$ is given via a state preparation oracle $\U_\b$ (\defin{state_prep}), while access to $A$ is given either via a sparse-matrix oracle $\P_A$ (\defin{sparse_access}) or via a matrix-block-encoding $\U_A$ (\defin{U_A}). Then, PD-QLS solving algorithms reaching a constant precision $\varepsilon \in \O(1)$ have query complexities $Q[\U_\b], Q[\U_A], Q[\P_A]$ all in $\Omega\big(\min(\kappa, N) \big)$.
\end{proposition}

The proof of these lower bounds is rather technical and can be found in \app{lowerbounds}. We now present a weaker result that, however, can be easily proven as a consequence of the optimality of Grover search~\cite{Boyer96}; namely, we show that PD-QLS solving has a linear scaling in $\kappa$ for all $\kappa \leq \sqrt{N}$.

\begin{proposition}
\label{prop:lower_bound_Grover}
Consider a PD-QLS problem as presented in \defin{PD-QLS} and suppose that access to $A$ is given by a sparse-matrix oracle $\P_A$ (\defin{sparse_access}), with no assumption on the access model for $\b$. Then, a quantum algorithm that solves the QLS up to any constant precision $\varepsilon \in [0,1)$ must make $\Omega\big(\min(\kappa,\sqrt{N})\big)$ accesses to $\P_A$.
\end{proposition}

\begin{proof}

Consider the search problem of finding an element $j \in \mathcal{S}$, where $\S \subseteq \{1,\ldots, N\}$ is a subset containing $M$ elements. The membership of a element $j$ in $\mathcal{S}$ is encoded as a quantum oracle $\P_\mathcal{S}$ which flips a ancilla qubit if $j \in \S$ and leaves the ancilla unchanged if $j \notin \S$.

Consider, next, a matrix $A$ that is diagonal (and thus 1-sparse) having entries
\al{
	A_{j,j}
	\; = \;
	\begin{cases}
	\alpha = \sqrt{\frac{N-M}{N}} & \t{if}~ j \notin \S \\
	\beta  = \sqrt{\frac{M}{N}}   & \t{if}~ j    \in \S \,.
	\end{cases}
	\label{eq:A_parameters}
}
The sparse-matrix oracle $\P_A = (\P_A^{pos}, \P_A^{val})$ for this diagonal matrix $A$ can be implemented with exactly one access to the membership oracle $\P_\S$. In fact, $\P_A^{pos}$ can be implemented without any access to $\P_\S$, since it is known that the non-zero entries are on the diagonal, while $\P_A^{val}$ can be implemented with a single access to $\P_\S$, assuming that $M$ and $N$ are known: by definition we have $\P_\S \ket{j,0} = \ket{j,0}$ if $j \notin \S$  and $\P_\S\ket{j,0} = \ket{j,1}$ if $j \in \S$, thus it is sufficient to apply on the ancilla system the transformation $\ket{0} \mapsto \ket{\alpha}$ and $\ket{1} \mapsto \ket{\beta}$, where the quantum register contains a binary representation of the numbers $\alpha$ and $\beta$, to implement $\P_A^{val}$.

Now notice that $A$ is a matrix having condition number $\kappa = \frac{\alpha}{\beta} = \sqrt{\frac{N-M}{M}} \in \Theta \Big(\sqrt{\frac{N}{M}} \,\Big)$ assuming $M \leq N/2$. Moreover, $A^{-1}$ is also diagonal, with entries
\al{
	(A^{-1})_{j,j}
	\; = \;
	\begin{cases}
	\sqrt{\frac{N}{N-M}} & \text{if}~ j \notin \S \\
	\sqrt{\frac{N}{M}}   & \text{if}~ j    \in \S \,.
	\end{cases}
}
Solving the QLS for the known-term vector $\ket{\b} = \ket{\u_N} = \frac{1}{\sqrt{N}} \sum_{j=1}^N \ket{j}$ yields
\al{
\ket{\x}\ 
& = \
	\frac{A^{-1} \ket{\u_N}}{\norm{A^{-1} \ket{\u_N}}} \\
& = \
	\frac{1}{\sqrt{2N}}
	\bigg(
	\sqrt{\frac{N}{N-M}}	
	\sum_{j \notin \S} \ket{j} +
	\sqrt{\frac{N}{M}}	
	\sum_{j \in \S} \ket{j}
	\bigg)	\\
& \equiv \
	\frac{1}{\sqrt{2}}
	\Big(
	\ket{j \notin \S} + \ket{j \in \S}
	\Big) , \label{eq:Grover}
}
where in the last line we have introduced the normalised vectors $\ket{j \notin \S} := \frac{1}{\sqrt{N-M}}	\sum_{j \notin \S} \ket{j}$ and $\ket{j \in \S} := \frac{1}{\sqrt{M}}	\sum_{j \in \S} \ket{j}$. Measuring $\ket{\x}$ in the computational basis therefore solves the search problem with probability $1/2$.

Suppose now that exists a quantum algorithm $\mathcal{A}$ that solves the QLS problem exactly ($\varepsilon = 0$) for PD matrices and that queries $\mathfrak{o}\big(\kappa\big)$ times the oracle $\P_A$. Applying $\mathcal{A}$ to the diagonal matrix $A$ and $\ket{\b} = \ket{\u_N}$ would then require only $\mathfrak{o}(\sqrt{N/M}\,)$ calls to $\P_A$, and hence to $\P_\S$, in order to produce $\ket{\x}$. This means that using $\mathcal{A}$ we can solve an unstructured search problem using $\mathfrak{o}(\sqrt{N/M}\,)$ queries to $\P_\S$, violating the optimality of Grover search.

Next, suppose that the algorithm $\mathcal{A}$ is an approximate PD-QLS solver, i.e.\ that it produces a state $\rho_\x$ such that $\bnorm{\,\rho_\x - \proj{\hspace{.5pt}\x\,}}_{\Tr} \leq \varepsilon$ for some constant $\varepsilon < 1/2$. Apply a projective measurement to $\rho_\x$ that projects it on the space spanned by $\{\ket{j}\}_{j\in\S}$ with probability $p$ and projects it onto the orthogonal subspace with probability $q=1-p$; in the ideal case ($\varepsilon=0$) we would have $p=q=1/2$. By the operational definition of the trace distance, the probability distribution $(p,q)$ must be at most $\varepsilon$-distinguishable from $(1/2,1/2)$, i.e.\ $\max \{|p-1/2|, |q-1/2|\} \leq \varepsilon$. Thus, the success probability is a constant $p \geq 1/2 - \varepsilon > 0$.

Finally, this argument can be extended to any constant precision $\varepsilon < 1$. It is sufficient to define a new diagonal matrix $\widehat{A}$ by changing the values $\alpha$ and $\beta$ in Eq.~\eqref{eq:A_parameters}, so that the probability $\widehat{p}$ of finding an element $j \in \S$ when measuring $|\widehat{\x} \rangle = |\widehat{A}^{-1} \u_N \rangle$ satisfies $\widehat{p} > 1 - \varepsilon$.

\end{proof}

Notice that in the proof the vector $\ket{\b} = \ket{\u_N}$ is fixed and easy to produce and hence the access model for $\ket{\b}$ plays no role in our reduction. We also remark that this proof can be straightforwardly modified to prove that the operation of quantum matrix-vector multiplication (i.e., obtaining a state proportional to $A\ket{\b}$) must also have a linear cost in $\kappa$. Moreover, since a sparse oracle access $\P_A$ allows to efficiently implement also a matrix-block encoding $\U_A$~\cite{Childs15}, the same reduction immediately rules out oracular algorithms that use $\mathfrak{o}(\kappa)$ accesses to $\U_A$ (for $\kappa \leq \sqrt{N}$). Finally, a simple argument shows that a $\Omega(\kappa)$ lower bound holds also for the $\U_\b$-query complexity: an initial small difference between two input states $\b$ and $\b'$ can be magnified $\kappa$-fold in the corresponding outputs $\ket{A^{-1}\b}$ and $\ket{A^{-1}\b'}$, which is impossible unless one uses at least $\kappa$ accesses to $\U_\b$~\cite{PrivateComm}.

\section{Method based on low-degree polynomial approximations of \texorpdfstring{$A^{-1}$}{A-1}}
\label{sec:block_encoding}

We start this section giving a few details on how to use the Quantum Signal Processing (QSP) method to implement polynomial functions of a matrix that we can access through a matrix-block-encoding (\sec{B}) and then we provide the explicit definition of the approximating polynomials of the inverse function (\sec{C}). Next (\sec{D}) we discuss two cases where we can implement a matrix-block-encoding of $B = \Id - \eta\,A$, as required to achieve a quadratic reduction in the degree of the approximating polynomials. Finally (\sec{from-to}) we discuss the cost of matrix-vector multiplication and summarise the cost of PD-QLS solving via this approach.

\subsection{The quantum signal processing method}
\label{sec:B}

We employ QSP as a tool to implement $\U_{A^{-1}}$, a matrix-block-encoding of $A^{-1}$, assuming that we have access to a normalised matrix-block-encoding $\U_B$ of 
\al{
	B := \Id - \eta\,A \,
	\label{eq:B_definition}
}
for some $\eta>0$. We assume that the spectrum of $B$ is contained in the interval $\mathcal{D}_B = [-1,1-\eta/\kappa]$, where $\kappa$ is an upper bound to the condition number of $A$. The QSP method can be stated, already specialised to our case of interest, as follows~\cite[Theorem 56]{Gilyen18}.

\begin{theorem}
\label{thm:Signal}
Consider a $(\beta,b,\epsilon)$-block-encoding $\U_B$ of a Hermitian matrix $B$ and let $P(x)$ be a degree-$\ell$ real polynomial with $|P(x)| \leq 1/2$ for all $x \in [-1, 1]$. Then there is a quantum circuit $\U_{P(B/\beta)}$ which is a $(1, b+2, 4\ell \sqrt{\epsilon/\beta})$-block-encoding of $P(B/\beta)$, and requires $\ell$ applications of $U_B$ and $U_B^\dag$, a single application of controlled-$U_B$ and $\O\big((n+b)\ell\big)$ elementary quantum gates. Moreover, the same result holds for polynomials satisfying $|P(x)| \leq 1$ if $P(x)$ has defined parity, i.e., $P(-x) = P(x)$ or $P(-x) = -P(x)$.
\end{theorem}

Importantly, there are explicit classical algorithms that can efficiently compute a parametrization of $\,\U_{P(B/\beta)}$ for any polynomial $P$ and then compile an explicit quantum circuit that implements it, see Refs.~\cite{Chao20,Dong20} for the current state-of-the-art.

We now further motivate the need to choose $\beta=1$, i.e., that the matrix-block-encoding of $B$ is normalised; equivalently, the part proportional to the identity in the definition~\eqref{eq:B_definition} must not be rescaled. We remind that, as presented in \sec{overview1}, our goal is to implement an polynomial $P(B/\beta)$ approximating $A^{-1}$, that is  
\al{
	P(B/\beta) \approx
	f(B) =
	\frac{1}{I - B} =
	\frac{1}{I - (I-\eta\,A)} =
	\frac{1}{\eta} A^{-1} .
}
Note, however, that in this expression actually we have $P(x) \approx \frac{1}{1-\beta  x}$, a function that has a singularity in $x = 1/\beta \leq 1$. As a consequence of Bernstein's inequality~\cite{Schaeffer41} this function may have polynomial approximations with quadratically better degree only when the singularity is in $x=1$, i.e.\ when we have $\beta =1$.

\subsection{Polynomial approximation of $1/(1-x)$ }
\label{sec:C}

In this section we show analytical polynomial approximations of the function $f(x) = 1/(1-x)$ so that we can use the QSP method to implement it for the matrix $B = \Id - \eta\,A$ as in Eq.~\eqref{eq:B_definition}. To keep the notation simple we assume $\eta =1$ and that the spectrum of $A$ is contained in $\mathcal{D}_A = \left[\, \frac{1}{\kappa}, \, 2\right]$, while we can account for any value $\eta < 1$ simply by rescaling $\kappa$ to $\kappa/\eta$. Consequently, it is only necessary for the polynomial $P(x)$ to be a good approximation of our target function in the domain $\D_B = \left[\, -1, \, 1\! -\! \frac{1}{\kappa}\,\right]$.

Our starting point will be the polynomial $\hat{\T}_{\ell,\kappa}(x)$, a shifted and rescaled version of $\T_\ell(x)$, the $\ell$-degree Chebyshev polynomial of the first kind,
\al{
\label{eq:Cheb_def}
	\hat{\T}_{\ell,\kappa}(x)
	:=
	\frac{\T_\ell \left(\frac{x + \frac{1}{2\kappa}}{1-\frac{1}{2\kappa}} \right)}
	     {\T_\ell \left(\frac{1 + \frac{1}{2\kappa}}{1-\frac{1}{2\kappa}} \right)} 
	\,,
}
which is the solution of the following minimax problem (see Ref.~\cite[Theorem 6.25]{Saad03}):
\al{
	\hat{\T}_{\ell,\kappa}(x) =
	\underset{\substack{P \in \mathds{R}_\ell, \\ P(1) =1}}{\t{argmin}}
	\max_{x\in[-1,1-\frac{1}{\kappa}]}
	\big| P(x) \big| \,.
}

Chebyschev polynomials satisfy the property that $|\T_\ell(x)| \leq 1$ for all $\ell \in \mathds{N}$ and all $x \in [-1,+1]$, while $\T_\ell(1 + \delta) \geq \frac{1}{2} e^{\ell \sqrt{\delta}}$ for $0 \leq \delta \leq 1/6$, see e.g.~\cite[Lemma 13]{Lin19}. Using the changes of variables 
\al{
\label{eq:change}
	y(x) := \frac{x + \frac{1}{2\kappa}}{1-\frac{1}{2\kappa}} 
	\qquad \quad
	\delta := 
	\frac{1 + \frac{1}{2\kappa}}{1-\frac{1}{2\kappa}} - 1
	= \frac{1}{\kappa - 1/2 }
}
we can rewrite the definition in Eq.~\eqref{eq:Cheb_def} as
\al{
	\hat{\T}_{\ell,\kappa}(x)
	=
	\frac{\T_\ell \big( y(x) \big)}
	     {\T_\ell \big( 1 + \delta \big)} \,.
}
Then, the numerator satisfies $\big| \T_\ell \big( y(x) \big) \big| \leq 1$ for all $x \in \D_B = [-1, 1 - \frac{1}{\kappa}]$, while the denominator is $\T_\ell \big( 1+\delta \big) \geq \frac{1}{2} e^{\ell \sqrt{\delta}}$. This means that it is sufficient to choose $\ell \geq \frac{1}{\sqrt{\delta}}  \log \!\left( \frac{2}{\epsilon}\right)$, i.e.\ $\ell \in \Theta\big(\sqrt{\kappa} \log(1/\epsilon)\big)$, to obtain $\big|\hat{\T}_{\ell,\kappa}(x)\big| \leq \epsilon$ on the interval $\D_B$.

We then use the following $(2\ell -1)$-degree polynomial as our approximation of $f(x) = \frac{1}{1-x}$:
\al{
\label{eq:poly_approx}
	P_{2\ell-1,\kappa}(x) := \frac{1}{1-x} \left[1 - \hat{\T}_{\ell,\kappa}(x) \right]^2 .
}
To see that $P_{2\ell-1,\kappa}(x)$ is indeed a polynomial, note that $[1 - \hat{\T}_{\ell,\kappa}(x)]$ has a twofold root in $x = 1$, thus is exactly divisible by $1-x$ and moreover $P_{2\ell-1,\kappa}(1) = 0$. This last property is useful because it allows us to implement the pseudo-inverse $A^+$ for a singular matrix $A$; i.e., in the case in which $A$ has some eigenvalues that are equal to $0$ (equivalently, $B = I-A$ has eigenvalues equal to $1$) these will be mapped to $0$ and if moreover all the non-zero eigenvalues are separated from zero by a gap $1/\kappa$, the the polynomial in Eq.~\eqref{eq:poly_approx} is a close approximation of the mapping $(1-\lambda) \mapsto 1/\lambda$ for all $\lambda \neq 0$ in the spectrum of $A$. We choose the degree $\ell \in \Theta\big(\sqrt{\kappa} \log(\kappa/\varepsilon)\big)$ in such a way that $\big|\hat{\T}_{\ell,\kappa}(x)\big| \leq \varepsilon / (3\kappa)$ for all $x \in \D_B$ and thus we obtain from Eq.~\eqref{eq:poly_approx} 
\al{
\label{eq:P_bound}
	\left\vert
	P_{2\ell-1,\kappa}(x) - \frac{1}{1-x}
	\right\vert
	\; \leq \;
	\kappa 
	\left(
	2\,\frac{\varepsilon}{3\kappa} + 
	\frac{\varepsilon^2}{9\kappa^2}
	\right) 
	\; \leq \;
	\varepsilon
	\qquad \quad
	\forall x \in \D_B\,,
} 
that is, we have an $\varepsilon$-close polynomial approximation on the interval $\D_B = [-1, 1 - \frac{1}{\kappa}]$. This directly implies $\norm{P_{2\ell-1,\kappa}(B) - A^{-1}} \leq \varepsilon$ in operator norm.

Finally, the polynomial in Eq.~\eqref{eq:poly_approx} has to be normalised so that it becomes compatible with the QSP method. Therefore we define
\al{
\label{eq:polynomial}
	\hat{P}_{2\ell-1,\kappa}(x) :=
	\frac{P_{2\ell-1,\kappa}(x)}{K}\,,
	\qquad 
	\t{where} ~ 
	K := \,2\!\! \max_{x\in [-1,+1]} |P_{2\ell-1}(x)| \;.
}
With this definition we have $|\hat{P}_{2\ell-1}(x)| \leq \frac{1}{2}$ for $x\in[-1,+1]$, as required. The normalisation constant satisfies $K \in \Theta(\kappa)$ for $\ell \in \Omega(\sqrt{\kappa})$, see \app{polynomial} for the proof of this bound. In conclusion, the QSP method allows us to implement a $(K, b+2, \varepsilon)$-matrix-block-encoding of $A^{-1}$, where $b$ is the number of ancilla qubits required in the block-encoding of $B$.

\subsection{Implementing normalised matrix-block-encodings of \texorpdfstring{$B = I$ - $\eta A$}{B=I-ηA}}
\label{sec:D}

In this subsection we show two methods that, under different assumptions, allow us to efficiently implement a normalised matrix-block-encoding of $B = \Id - \eta\,A$. Preliminarily, we remark that this is a non-trivial task, as we argue with the following three considerations.

First, any Hermitian matrix $M \in \mathds{C}^{N \times N}$ satisfying $\norm{M} \leq 1$ can be implemented as a normalised sub-block of a unitary, since an explicit construction is given by~\cite{Low16}
\al{
	\U_M
	=
	\left(
	\begin{array}{cc}
	M & -\sqrt{\Id - M^2}\\
	\sqrt{\Id - M^2} & M
	\end{array}
	\right).
}
Implementing the circuit corresponding to $\U_M$ in general requires up to $\O(N^2)$ elementary quantum gates~\cite{Shende06} and is thus inefficient; however, efficient constructions are possible in specialised cases.

Second, standard quantum methods allow us to efficiently implement \emph{sub-normalised} matrix-block-encodings of $B$. As a first example, Childs' quantum walk operator uses $\O(1)$ accesses to a sparse-matrix oracle $\P_B$ to implement $\U_{M}$ for $M = B/d$, where $d$ is the column-sparsity of $B$~\cite{Childs15}. A second example is to assume that we have access to $\U_{A}$, a block-encoding of $A$ with $\norm{A}\leq 1$, and then use the LCU lemma~\cite{Childs12} to implement a linear combination of $\Id$ and $-\U_{A}$, yielding a normalised block-encoding of $p \,\Id - (1-p)A \equiv p\,B$, for some $p \in (0,1)$. However, any ``black-box'' method that aims at amplifying a block-encoding of $B / \beta$ (with $\beta > 1$) to a normalised block-encoding of $B$ is in general inefficient. This can be proven, for instance, applying the lower bound in Ref.~\cite[Theorem 73]{Gilyen18} to the function $f(x) =  \beta\,x$.

Third, it is currently an open problem whether it is possible or not to implement a normalised matrix-block-encoding $\U_M$ given access to a sparse-matrix oracle $\P_M$ with $\norm{M} \leq 1$. In absence of general results of this kind, we then turn to developing specialised methods to efficiently implementing a normalised block-encoding $\U_B$, with $B = \Id - \eta\,A$, in the cases where $(i)$ $A$ is diagonally dominant or $(ii)$ $A$ is the sum of positive semi-definite local Hamiltonians.

\subsubsection{Diagonally-dominant coefficient matrix}

In this section, we implement a normalised block-encoding of $B = \Id -A$ employing the method described in Ref.~\cite[Lemma~47]{Gilyen18}, which we report here in \lem{Gram}. Our construction requires the preparation of some families of states $\{\ket{\psi_i}\}_i, \{\ket{\phi_j}\}_j$ that are well-defined only when $A$ is diagonally-dominant, while attempts at extending the method to $B = \Id - \eta\,A$ for non-diagonally-dominant PD matrices results in non-normalisable states for any $\eta > 0$. We also remark that the problem of solving linear systems involving diagonally dominant PD matrices (which includes the noteworthy class of Laplacian matrices \cite{Merris94}) is well studied in classical linear algebra: for these classes of matrices there are classical algorithms that can solve a linear system substantially faster than what is possible for more general matrices \cite{Spielman10}.

\begin{lemma} 
\label{lem:Gram} 
Suppose that we have access to two ``state preparation'' unitaries $U_L$ and $U_R$ (left and right) acting on $a+s$ qubits such that
\al{
	U_L : \ket{0^a} \ket{i}\, & \mapsto \ket{\psi_i} \\ 
	U_R : \ket{0^a} \ket{j} & \mapsto \ket{\phi_j} ,
}
for any $i,j \in \{1, \ldots, 2^s\}$ and for some families of states $\{\ket{\psi_i}\}_i$ and $\{\ket{\phi_j}\}_j$. Then, it is immediate to see that $U_L^\dag U_R$ is a $(1,a,0)$-matrix-block-encoding of the Gram matrix $H$ such that $H_{ij} = \braket{\psi_i}{\phi_j}$.
\end{lemma}

Let $A$ be a Hermitian $d$-sparse diagonally-dominant matrix, i.e.\ $\sum_{j\neq i}|A_{ij}| \leq A_{ii} \leq 1$ for all $i$. By Gershgorin theorem~\cite{Gersh}, diagonal dominance of a Hermitian matrix is sufficient to guarantee  positive semi-definiteness, i.e.\ $\lambda_{\min}(A)\geq 0$, and both $\lambda_{\min}(A) = 0$ and $\lambda_{\min}(A) \neq 0$ are possible\footnote{If $A$ is \emph{strictly} diagonally dominant we have $\lambda_{\min}(A) \geq \min_i \big\lbrace A_{ii} - \sum_{j\neq i}|A_{ij}| \big\rbrace > 0$ and then the condition number is immediately bounded by $\kappa(A) \leq 1/\lambda_{\min}(A)$ when $\norm{A} \leq 1$.}. Consider then the states
\al{
\label{eq:psi_state}
	& \ket{\psi_i} := 
	\sum_{l \in \t{supp}(A_i)}
	\sqrt{\delta_{il} - A_{il}} \ket{l}
	\; + \; \sqrt{r_i} \ket{N+1} 
	\\
	& \ket{\phi_j} := 
	\sum_{k \in \t{supp}(A_j)}
	\sqrt{\vphantom{A_{il}} \smash[b]{\delta_{jk} - A_{jk}^*} } \ket{k}
	\; + \; \sqrt{r_j} \ket{N+1} 
}
where $\t{supp}(A_i)$ are the position of the (at most) $d$ non-zero entries of the column vector $A_i$ and the value $0 \leq r_i \leq 1$ can be computed so that the states are normalised, since we have
\al{
\label{eq:normalization}
	|r_i| =
	1 - \sum_{l\in \t{supp}(A_i)} 
	\left|\sqrt{\delta_{il} - A_{il}}\right|^2
	=
	A_{ii} - \sum_{l\neq i} \big| A_{il} \big| 
   \geq 0
}
where we have used $A_{ii} \leq 1$ and the diagonal dominance of $A$. We then define the following state-preparation unitaries:
\al{
\label{eq:U_L}
	U_L : \ket{0^a} \ket{i} & 
	\; \mapsto \; 
	|0^b\rangle \ket{i} \ket{\psi_i^*}  \\
	U_R : \ket{0^a} \ket{j} & 
	\; \mapsto \; 
	|0^b\rangle \ket{\phi_j} \ket{j}, 
}
for certain numbers $a$ and $b$ of ancilla qubits initialised in $\ket{0}$, and where $\ket{\psi_i^*}$ is the complex conjugate of the state $\ket{\psi_i}$ w.r.t.\ the computational basis. Then, $U_L^\dag U_R$ is a normalised encoding of the matrix $B = I - A$, as one can verify using $A_{ji}^* = A_{ij}$:
\al{
	B_{ij} = \braket{0^b,i,\psi_i^*}{0^b,\phi_j,j} 
	= 
	\underbrace
	{ \sqrt{\vphantom{A_{il}} \smash[b]{\delta_{ji}-A_{ji}^*}} }_
	{\braket{i}{\phi_j}}
	\underbrace
	{ \sqrt{\delta_{ij}-A_{ij}} }_
	{\braket{\psi_i^*}{j}}
	= 
	\delta_{ij}-A_{ij} \;.
}

The quantum circuit $U_L$ can be implemented efficiently (and similarly for $U_R$), as we now show. We use $d$ calls to $\P_A^{pos}$ to recover the values $ j_\nu \equiv j(i,\nu)$ for $\nu \in \{1, \ldots, d\}$, i.e., the positions of all the (potentially) non-zero entries of $A_i$. This corresponds to implementing the following isometry (i.e., a unitary circuit plus the possibility of adding ancillas):
\al{
\label{eq:comp1}
	\ket{i}
	\overset
	{d \times \P_{\!A}^{pos}}	
	{\mapsto}
	\ket{i} \ket{j_1, \ldots, j_d} .
}
Next, we use $d$ calls to $\P_A^{val}$ to recover all the values $A_{i,j}$, i.e.:
\al{
\label{eq:comp2}
	\eqref{eq:comp1}
	\overset
	{d \times \P_{\!A}^{val}}	
	{\mapsto}
	\ket{i}\ket{j_1, \ldots, j_d}\ket{A_{ij_1}, \ldots, A_{ij_d}} .
}
We then use reversible (classical) computation to calculate the numerical values of all the amplitudes $\boldsymbol{\psi}^{(i)} := (\sqrt{-A_{ij_1}}^{\,*}, \ldots, \sqrt{1-A_{ii}}^{\,*}, \ldots, \sqrt{-A_{ij_d}}^{\,*}, \sqrt{r_i}^{\,*})^T$:
\al{
\label{eq:comp3}
	\eqref{eq:comp2}
	\overset
	{\t{compute}}	
	{\mapsto}
	\ket{i}\ket{j_1, \ldots, j_d}\ket{A_{ij_1}, \ldots, A_{ij_d}} 
	| \boldsymbol{\psi}^{(i)} \rangle \,.
}
Then, we use a general state preparation quantum circuit which, given a classical description of the amplitudes of a $(d+1)$-dimensional quantum state, effectively prepares the corresponding state:
\al{
\label{eq:comp4}
	\eqref{eq:comp3}
	\overset
	{\t{prepare}}
	{\mapsto}
	\ket{i}\ket{j_1, \ldots, j_d}\ket{A_{ij_1}, \ldots, A_{ij_d}} 
	| \boldsymbol{\psi}^{(i)} \rangle 
	\sum_{\nu=1}^{d+1} 
	\psi_\nu^{(i)} \ket{\nu}	.
}
Next, we use a single call to $\P_A^{pos}$ in quantum superposition to map $\ket{i,\nu} \mapsto \ket{i,j(i,\nu)} = \ket{i,j_\nu}$ and thus we obtain
\al{
\label{eq:comp5}
	\eqref{eq:comp4}
	\overset
	{\P_{\!A}^{pos}}
	{\mapsto}
	\ket{i}\ket{j_1, \ldots, j_d}\ket{A_{ij_1}, \ldots, A_{ij_d}} 
	| \boldsymbol{\psi}^{(i)} \rangle 
	\sum_{\nu=1}^{d+1} 
	\psi_\nu^{(i)} \ket{j_\nu}
}
and now note that, adopting the definition $j(i,d+1) := N+1$, on the right-hand side we have obtained 
$
	\sum_{\nu=1}^{d+1} 
	\psi_\nu^{(i)} \ket{j_\nu}
	=
	\ket{\psi_i^*}
$,
the complex conjugate (w.r.t.\ the computational basis) of the state defined in Eq.~\eqref{eq:psi_state}. Finally, we ``uncompute'' the intermediate registers $\ket{j_1, \ldots, j_d}$, $\ket{A_{ij_1}, \ldots, A_{ij_d}}$ and $|\boldsymbol{\psi}^{(i)}\rangle$, mapping them to $\ket{0^b}$ (where $b$ is the number of left-over ancillas), performing the steps \eqref{eq:comp1} to \eqref{eq:comp4} in reverse. With a final swapping of $\ket{i}$ and $\ket{0^b}$, we have implemented the circuit $U_L$ given in Eq.~\eqref{eq:U_L}.

We can now estimate the query and gate cost of implementing $U_L$ (and the cost of implementing $U_R$ is the same). Going through the derivation, we see that $4d+1$ oracle calls to $\P_A = (\P_A^{pos},\P_A^{val})$ are required, that is the query complexity is $Q[\P_A] \in \O(d)$. Regarding gate complexity, step \eqref{eq:comp3} requires $\O\big(d \,\t{poly}(p)\big)$ Toffoli gates (which are universal for classical reversible computation) to perform the computation up to $p$ digits of precision; moreover, step \eqref{eq:comp4} can be performed using $\O(d)$ control-nots and single-qubit rotations employing the methods of Ref.~\cite{Shende06}. In conclusion, treating the number of digits of precision as a constant and the single-qubit rotations as exact, both the query and the gate complexities are in $\O(d)$ and we thus obtain the following result.

\begin{proposition} [Normalised matrix-block-encoding, diagonally-dominant case]
\label{prop:implement_B1}
Suppose that we have access to a $d$-sparse diagonally-dominant PD matrix $A \in \mathds{C}^{N \times N}$ via $\P_A$ as in \defin{sparse_access}. Then it is possible to implement a normalised block-encoding of $B = \Id - A$ with $\O(d)$ query and gate complexity, assuming exact single-qubit rotations and that all arithmetic operations are performed with a constant number of digits of precision. 
\end{proposition}

We finally remark that, in some cases, it might be possible to implement $U_L$ and $U_R$ with a query and gate complexity in $\O(\sqrt{d})$ instead of linearly in $d$. First, we may assume that we have directly access to an oracle $\P_\psi$ that directly returns the amplitudes $\psi_\nu^{(i)}$ (including the value $\sqrt{r_i}$), instead of needing to compute these vales from the non-zero entries of $A_i$. Second, one can use an algorithm that generalises Grover search to prepare the state $\ket{\psi_i^*}$ using $\O(\sqrt{d})$ accesses to $\P_\psi$~\cite{Grover00}, and an even more efficient implementation can be realised using the method of Ref.~\cite{Sanders19}, which avoids synthesising arithmetical operations and brings about an improvement of two orders of magnitude over prior works for realistic levels of precision.

\subsubsection{Sum of positive semi-definite Hamiltonians}

We now consider the case where $A\in \mathds{C}^{N\times N}$ is given by the sum of positive semi-definite Hamiltonian terms; i.e., we consider the case 
\al{
\label{eq:A_here}
	A =
	\sum_{j=1}^J
	H_{(j)} 
	\qquad 
	\forall j \ H_{(j)} \ \textup{is positive semi-definite,} 
}
which is similar to the case presented in \sec{hamiltonians}, but here we allow the Hamiltonian terms to have eigenvalue zero. We assume that the number $J$ of Hamiltonian terms scales polynomially in $n = \lceil \log_2 N \rceil$ and that each Hamiltonian term $H_{(j)}$ is local, i.e., that it acts upon a small number of qubits; specifically, we require that the set $\S_j \subseteq \{1,\ldots, n\}$ of qubits upon which $H_j$ acts non-trivially satisfies $|\S_j| \leq s \in \O(\log n)$ for all $j$. Each $H_{(j)}$ can thus be expressed as\footnote{
The correct expression for an operator $H_{(j)}$ acting on a set $\S_j \subseteq \{1,\ldots, n\}$ of $s$ qubits is
\al{
\label{eq:footnote}
	H_{(j)} 
	= 
	\sum_{\substack{a_1 \cdots a_s \in \{0,1\}^s \\ 
		            b_1 \cdots b_s \in \{0,1\}^s}}
	h^{(j)}_{a_1 \cdots a_s,b_1 \cdots b_s}
	\bigotimes_{r=1}^n \;
	\mathfrak{O}_r
	\qquad
	\t{with} ~ \mathfrak{O}_r =
	\begin{cases}
	I = \proj{0} \!+\! \proj{1} & \t{if} ~ r \notin \S_j \\
	\ket{a_r} \! \bra{b_r}      & \t{if} ~ r \in    \S_j
	\end{cases} 
	\;.
} 
}
\al{
\label{eq:H_j_here}
	H_{(j)} = 
	h_{(j)} \otimes I_{\neg \S_j}
}
where $h_{(j)}$ is a positive-definite matrix of dimension $2^{s} \times 2^{s}$, which can be fully specified with $2^{\O(\log n)} = \O( \t{poly} \, n)$ parameters. We also impose $\norm{h_{(j)}} \leq 2$ for all $j$.

Now we define $w_{(j)} := \Id - h_{(j)}$ and note that it is a small matrix, of at most $\O(\t{poly}\, n)$ size, with $\norm{w_{(j)}} \leq 1$. Then, we can rapidly compute, with classical algorithms, a unitary extension $u_{(j)}$ for each $w_{(j)}$, each requiring only one ancilla qubit. Specifically, we define 
\al{
\label{eq:u_j}
	u_{(j)} := 
	\left(
	\begin{array}{cc}
	w_{(j)} & - \sqrt{\Id - w_{(j)}^2}\\
	\sqrt{\Id - w_{(j)}^2}  & w_{(j)}
	\end{array}
	\right)
}
and then we (implicitly) define $W_{(j)}\in\mathds{C}^{N\times N}$ and the unitary $U_{(j)} \in \mathds{C}^{2N \times 2N}$
\al{
\label{eq:W_j}	
	W_{(j)} & = 
	w_{(j)} \otimes I_{\neg \S_j}\\
\label{eq:U_j}
	U_{(j)} & = 
	u_{(j)} \otimes I_{\neg \S_j}
}
where the interpretations are the same as in Eq.~\eqref{eq:H_j_here}. Note that each $U_{(j)}$ can be efficiently compiled as a quantum circuit: it is sufficient to determine the gate decomposition of the $(s+1)$-qubit matrix $u_{(j)}$ and then embed the circuit in a $(n+1)$-qubit quantum register. We assume that the ancilla qubit used for the extension given in Eq.~\eqref{eq:u_j} is one and shared across all circuits $U_{(j)}$.

We then employ the LCU lemma~\cite{Childs12} as follows. Given access to the controlled version of each circuit $U_{(j)}$, it is possible to efficiently implement the multi-controlled unitary
\al{
	U_\textsc{Select} = \sum_{j=1}^J \proj{j}_c \otimes U_{(j)} 
}
where the subscript $c$ denotes the control register. Then, defining a unitary $\t{Had}$ that acts as $\t{Had}\ket{0}_c = \sum_{j=1}^J \frac{1}{\sqrt{J}} \ket{j}_c$, we obtain:
\al{
	\big(\bra{0}_c \t{Had}^\dag \otimes \Id\big) \,
	U_\textsc{Select}\,
	\big(\t{Had}\ket{0}_c \otimes \Id \big) 
	& =
	\sum_{j=1}^J
	\frac{1}{J}
	U_{(j)}
}
and further post-selecting to the ``top-left'' corner of each $U_{(j)}$, according to Eq.~\eqref{eq:u_j} we obtain:
\al{
	\big(\bra{0} \otimes \Id\big)\,
	\sum_{j=1}^J
	\frac{1}{J}
	U_{(j)}\,
	\big(\ket{0} \otimes \Id\big)\,
	& =
	\sum_{j=1}^J
	\frac{1}{J}
	W_{(j)}
	= 
	\sum_{j=1}^J
	\frac{1}{J}
	\left(\Id - H_{(j)}\right) 
	=
	\Id - \frac{1}{J} A\,.
}

In conclusion, 
$	\U_B :=
	(\t{Had}^\dag \otimes \Id ) 
	\, U_\textsc{Select}\,
	(\t{Had} \otimes \Id )
$ 
is a normalised matrix-block-encoding of $B = \Id - \eta A$, where the factor $\eta =  1/J$ scales, by assumption, polynomially in $n$. The gate complexity of $\U_B$ can be estimated as follows. Each $u_{(j)}$ requires $\O(2^{2s})$ elementary gates to be implemented~\cite{Shende06} and thus the gate complexity of $U_{(j)}$ is also in $\O(2^{2s})$, assuming that two-qubit gates can be applied among arbitrary pairs of qubits. Then, $U_\textsc{Select}$ has gate complexity scaling as the sum of the complexities of the individual $U_{(j)}$~\cite{Berry14} and is thus in $\O(J2^{2s})$. The complexity of $\t{Had}$ is $\O(\log J)$, a sub-leading additive term that can be neglected in the asymptotic gate complexity of $\U_B$. Hence we have the following result.

\begin{proposition} [Normalised matrix-block-encoding, Sum-of-Hamiltonians case]
\label{prop:implement_B2}
Suppose that we have an explicit classical description of positive semi-definite matrices $h_{(j)} \in \mathds{C}^{2^s \times 2^s}$ with $j \in \{1, \ldots, J\}$ and consider the positive semi-definite Hamiltonian terms $H_{(j)}$ each obtained by applying $h_{(j)}$ to a subset of qubits $\S_j \subseteq \{1,\ldots,n\}$, as given in Eq.~\eqref{eq:H_j_here}. Consider then a coefficient matrix $A \in \mathds{C}^{2^n \times 2^n} $ as given in Eq.~\eqref{eq:A_here}. Then it is possible to implement a normalised block-encoding of $B = \Id - \frac{1}{J} A$ with a gate complexity in $\O(J2^{s})$, assuming exact single-qubit rotations.
\end{proposition}

\subsection{From matrix inversion to solving the quantum linear system problem}
\label{sec:from-to}

Suppose now that we have a matrix-block-encoding of 
\al{
	\hat{P}_{2\ell-1,\kappa}(B) 
	\approx 
	\frac{1}{K} \frac{1}{\Id-B}
	=
	\frac{A^{-1}}{\eta\,K}	
}
where $K \in \Theta(\kappa/\eta)$ is the normalization factor of the matrix-block-encoding (recall the rescaling of $\kappa$ to $\kappa/\eta$), upper bounded by $\O(\kappa/\eta)$ as we show in \app{polynomial}. This is equivalent to say that we have implemented the unitary
\al{
	\U_{A^{-1}} 
	\approx 
	\left(
	\begin{array}{cc}
	A^{-1} / (\eta\,K) & * \\
	*		  & *
	\end{array}
	\right)
}
where the left-upper block corresponds to having $a$ ancilla qubits in $\ket{0^a}$. Then, one can directly solve a QLS by applying the unitary quantum circuit $\U_{A^{-1}}$ to the vector $\ket{0^a}\ket{\b}$ and then post-select the outcome $\ket{0^a}$ on the ancilla system; however, post-selection might introduce large overheads, since we have
\al{
	\U_{A^{-1}}: \ket{0^a}\ket{\b}
	\quad \mapsto \hspace{8mm}
	\frac{1}{\eta\, K} \hspace{7mm}
	& \ket{0^a} A^{-1}\ket{\b} + \sqrt{1-\frac{1}{\eta^2K^2}}
	\, \big|\Psi^\perp\big\rangle \\
	= \, 
	\frac{\norm{A^{-1}\ket{\b}}}{\eta\, K}
	& \ket{0^a} \ket{A^{-1}\b} + \sqrt{1-\frac{1}{\eta^2K^2}}
	\, \big|\Psi^\perp\big\rangle 	
}
where $\big|\Psi^\perp\big\rangle$ is a state perpendicular to all states of the form $\ket{0^a,\psi}$. Therefore, the probability of successfully obtaining the state $\ket{A^{-1}\b}$ when post-selecting on the ancilla measurement is
\al{
	p_\t{succ} = 
	\frac{\norm{A^{-1}\ket{\b}}^2}{\eta^2 K^2} \,.
}

A PD-QLS solver that prepares a matrix-encoding of $A^{-1}$ and obtains $\ket{A^{-1}\b}$ via post-selection requires $\O( 1 / p_\t{succ})$ accesses to $\U_\b$ and $\O( (2\ell-1) / p_\t{succ})$ accesses to $\U_B$. Recall, $2\ell-1$ is the degree of $\hat{P}_{2\ell-1,\kappa}(B)$ and 
$
	\ell \in 
	\Theta \! \left( \sqrt{\frac{\kappa}{\eta}} 
	\log \frac{\kappa}{\eta \varepsilon} \right)
$.
The query complexities can be quadratically improved to $\O( 1 / \sqrt{p_\t{succ}})$ and $\O( (2\ell-1) / \sqrt{p_\t{succ}})$, respectively, using amplitude amplification~\cite{Brassard02}. Having implemented an approximate matrix-encoding of $\widetilde{A}^{-1}$ that is $\varepsilon$-close to $A^{-1}$ (using \defin{U_A}), the output state $\sket{\widetilde{A}^{-1}\b} = \widetilde{A}^{-1}\ket{\b} / \snorm{\widetilde{A}^{-1}\ket{\b}\!}$ satisfies~\cite[Proposition 9]{Childs15} 
\al{
	\norm{\big| \widetilde{A}^{-1}\b\big\rangle - \ket{A^{-1}\b}}
	\leq 
	4 \,\varepsilon \,.
}
The same inequality holds in trace distance because of~\eqref{eq:trace_ineq} and it is true both when using post-selection and when using amplitude amplification. In conclusion, we have the following results.

\begin{proposition} [Complexity of the PD-QLS solver]
\label{prop:Algorithm1}
Suppose that we have access to a $(1,b,0)$-matrix-block-encoding $\U_B$ of $B = \Id - \eta \, A$, where $\eta \in (0,1]$ and $A \in \mathds{C}^{N\times N}$ is a PD matrix with eigenvalues contained in the interval $\D_A = \big[\frac{1}{\kappa}, 2\big]$ for some known value $\kappa>1$; see e.g.\ \prop{implement_B1} and \prop{implement_B2} for explicit constructions.

First, using the QSP method of \thm{Signal} and the polynomial approximation given in Eq.~\eqref{eq:poly_approx}, it is possible to implement a $(K,b+2, \varepsilon)$-matrix-block-encoding of $A^{-1}$, where $K \in \Theta(\kappa/\eta)$, and the method has a query complexity 
\al{
	Q[\U_B] 
	& \in 
	\O\!\left(\sqrt{\frac{\kappa}{\eta}} \, 
	\log \frac{\kappa}{\eta \,\varepsilon} \right) ,
}
where $Q[\U_B]$ denotes the number of accesses to $\U_B$. Moreover, the algorithm is gate-efficient, i.e.\ it requires $\O\big(\t{poly}(n) \, Q[\U_B]\big)$ extra elementary quantum gates.

Second, suppose that we want to solve a PD-QLS as in \defin{PD-QLS}, where access to $A$ is given (indirectly) by $\U_B$ and access to $\b$ via a state preparation oracle $\U_\b$ as in \defin{state_prep}. Then, the QLS can be solved up to precision $\O(\varepsilon)$ using the $(K,b+2,\varepsilon)$-matrix-block-encoding of $A^{-1}$ and employing amplitude amplification to perform matrix-vector multiplication with constant success probability. The total query complexities, in terms of accesses to $\U_B$ and $\U_\b$, are
\al{
	Q[\U_\b] 
	& \in
	\O\!\left(
	\frac{\kappa}{\norm{A^{-1} \ket{\b}}}
	\right) \\
	Q[\U_B] 
	& \in
	\O\!\left(
	\sqrt{\frac{\kappa}{\eta}}
	\frac{\kappa}{\norm{A^{-1} \ket{\b}}}
	\log \frac{\kappa}{\eta\,\varepsilon} \right)
}
and the algorithm is gate efficient, that is, the gate complexity is in $\O\big(Q \,\t{poly}(\log Q,\log N)\big)$. A quadratic speed-up in $\kappa$ (up to polylogarithmic factors) is achieved over general QLS solvers when $\norm{A^{-1} \ket{\b}} \in \O(\kappa)$.
\end{proposition}

We now proceed to a worst-case, average case, and best-case scenario analysis of a PD-QLS solver as given in the previous Proposition.

\paragraph*{Worst-case scenario:}
In the worst case we have $\norm{A^{-1}\ket{\b}} \in \O(1)$ and consequently the post-selection success probability is $p_\t{succ} \in \Omega(1/\kappa^2)$. One can in alternative use $\O(\kappa)$ rounds of amplitude amplification to reach a constant success probability. Using amplitude amplification, the overall query complexity is in $\O(\kappa^{3/2})$, which is an improvement compared to the $\O(\kappa^{2})$ runtime achieved by the HHL algorithm, but still falls shorts of the $\widetilde{\O}(\kappa)$ runtime that can be achieved using more advanced methods such as Variable-Time Amplitude Amplification (VTAA)~\cite{Ambainis10} or eigenpath transversal~\cite{Subasi18}.

\paragraph*{Average-case scenario:}
We now look at the distribution of runtimes that arises when using a randomly chosen vector $\ket{\b}$. As observed in the work by Subaşı and Somma~\cite[Section III.B]{Somma20} one could model the eigenvalues of a positive-definite matrix $A$ as uniformly distributed over the interval $[1/\kappa,1]$ while if $\ket{\b}$ is chosen from the outcome of a random quantum circuit its amplitudes are sampled according to a Porter-Thomas distribution; then, $\norm{A^{-1}\ket{\b}} \in \O(\sqrt{\kappa})$ almost surely in the regime $1 \ll \kappa \ll N$. This implies that a randomly sampled PD-QLS problem (according to the  specified distribution) can be solved almost surely with a query complexity in $\O(\kappa)$, employing amplitude amplification in the post-selection step. This method then matches (actually, improves by a polylog($\kappa$) factor) the asymptotic runtime of more sophisticated methods (such as those that employ VTAA or adiabatic evolution) when considering ``typical'' instances of PD-QLS.

\paragraph*{Best-case scenario:}
The largest value that $\norm{A^{-1}\ket{\b}}$ can reach is $\kappa$. In such case the post-selection probability is constant and the overall query complexity is in $\O(\sqrt{\kappa})$, even without employing amplitude-amplification. This is a quadratic improvement for PD-QLS solving over competing methods working for indefinite QLS, since implementing a block-encoding of $A^{-1}$ for indefinite matrices already requires $\O(\kappa)$ oracle calls to $\U_A$~\cite{Childs15}. We note these best-case problems almost never occur under the probabilistic model described before, but real-world problems have intrinsic structure that could make them depart from the Porter-Thomas distribution and thus have $\norm{A^{-1}\ket{\b}} \gg \sqrt{\kappa}$: for instance, this is the case if $\ket{\b}$ has constant overlap with the eigenvector relative to the largest eigenvalue of $A^{-1}$. Note, finally, that it is not required that we know in advance how large the success probability is, since by definition $\ket{0^a}$ heralds the success.

\begin{figure}[t!]
\includegraphics[scale=.52,trim=18mm 15mm 0 0]{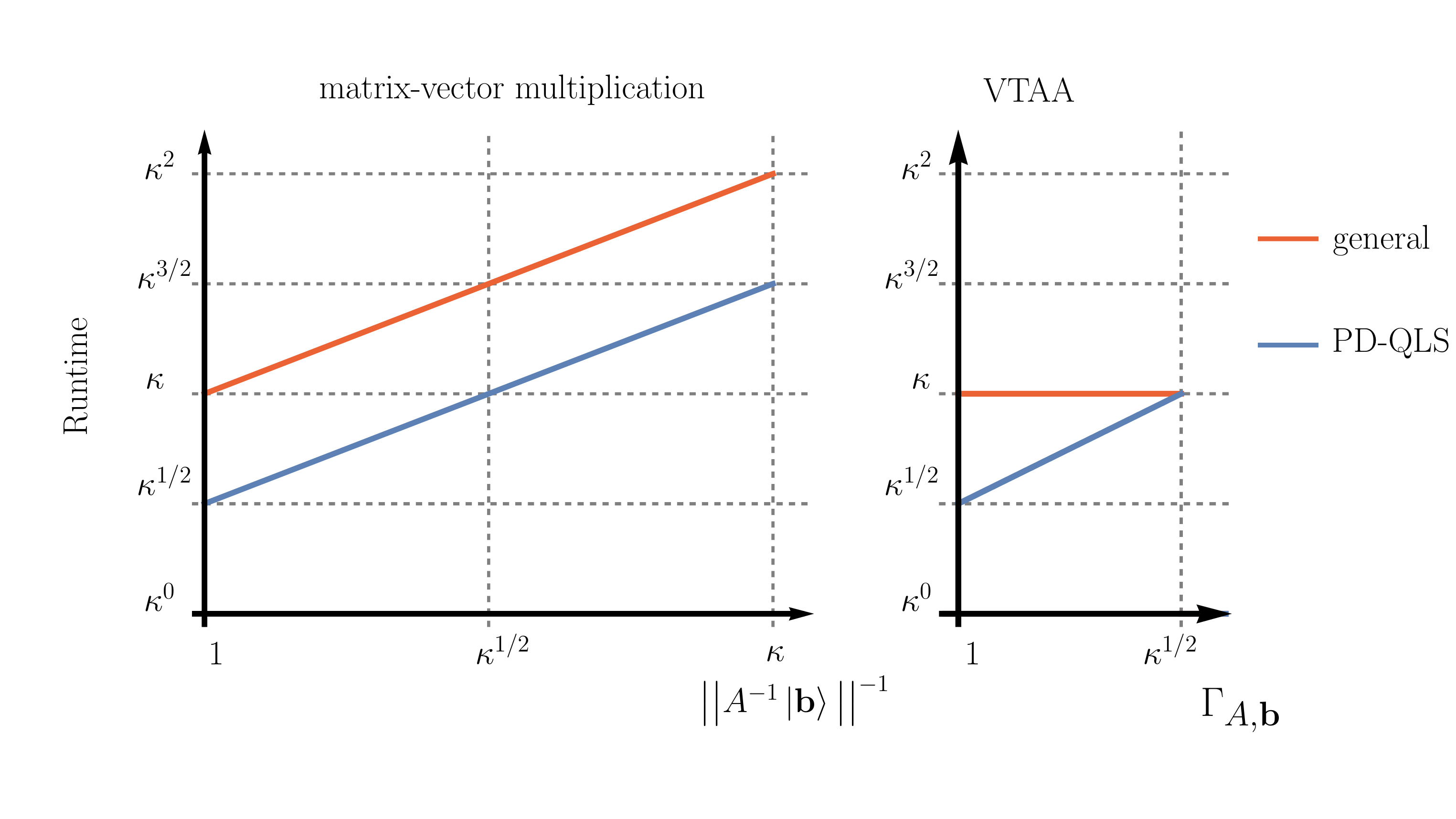}
\caption{Bi-logarithmic plots of the asymptotic runtimes of general QLS solvers and of our PD-QLS solvers. The plot on the left represents the scaling of the query complexity for methods based on direct matrix-vector multiplication (together with amplitude amplification) in terms of the variable $\norm{A^{-1}\ket{\b}}^{-1} \in [\,1,\kappa\,]$. The plot on the right represents the scaling of the query complexity for methods based on VTAA in terms of the variable $\Gamma_{A,\b} := \sqrt{\kappa} \,\frac{\norm{A^{-1/2}\ket{\b}}}{\norm{A^{-1}\ket{\b}}} \in [\,1,\sqrt{\kappa}\,]$. In both plots we ignore poly-logarithmic multiplicative factors and we express all the variables in units of powers of $\kappa$, assuming $\kappa \gg 1$ as constant. See the main text for details.}
\label{fig:Runtimes}
\end{figure}

\subsection{Optimization using Variable-Time Amplitude Amplification} 

In \app{VTAA} we show how to use VTAA to obtain a PD-QLS solver having improved asymptotic query complexities. We proceed as in Ref.~\cite[Section 5]{Childs15} and Ref.~\cite[Section 3]{Chakraborty18}: first we reformulate our algorithm as a variable-stopping-time quantum algorithm and afterwards we apply the VTAA optimisation to improve its runtime. There is, however, a technical hurdle to overcome: all previous VTAA-based QLS solvers use a phase estimation subroutine having a $\O(\kappa)$ runtime; its use would then preclude us from achieving a runtime sub-linear in $\kappa$. The main new idea we introduce is to replace phase estimation with efficient ``windowing functions''  whose implementation via QSP requires only $\widetilde{\O}(\sqrt{\kappa})$ accesses to $\U_B$.

More precisely, in Ref.~\cite{Childs15} the so-called Gapped Phase Estimation (GPE) method is introduced, with the purpose of reliably selecting eigenvalues of $A$ that are larger than some value $\delta$. For these eigenvalues it is possible to implement an approximate inverse at a reduced cost, scaling as $\O(1/\delta)$ instead of $\O(\kappa)$. Then, a sequence of increasingly small values of $\delta$ is considered, until $\delta \leq 1/\kappa$, and VTAA is employed to enhance the success probability. Since GPE has a query complexity in $\Omega(1/\delta)$ and the required precision is $\delta \leq 1/\kappa$, its complexity is in $\Omega(\kappa)$.

Instead, we introduce a ``windowing function'' $W_{\epsilon,\delta}(\lambda)$ to select eigenvalues of $B = \Id -\eta A$ that satisfy $\lambda \in [-1+2\delta, 1-2\delta]$ and to reject eigenvalues $\lambda \in [-1, -1 +\delta] \cup [1-\delta,+1]$, except for a small error probability $\epsilon$. Thus, $W_{\epsilon,\delta}(x)$ is a polynomial $\epsilon$-close to 1 in the center of the interval $[-1,+1]$ and $\epsilon$-close to 0 near the edges of the interval, with a steep fall around the points $\pm (1- 1.5 \delta)$. The intervals where the function derivative is large (of order $1/\delta$) are very close to the extrema of the interval $[-1,+1]$. According to Bernstein's inequality~\cite{Schaeffer41} it is not prohibited that a windowing function could be implemented with a polynomial having a degree $\ell \in \widetilde{O}(1/\sqrt{\delta})$, a quadratically smaller degree compared to case where the large derivative is near the center of the interval. In \app{VTAA} we then show, with an explicit construction, that windowing polynomial with degree $\ell \in \widetilde{O}(1/\sqrt{\delta})$ indeed can be implemented. We defer to the Appendix for further details.

The end result is summarized in the following Proposition.

\begin{proposition}[Complexity of PD-QLS with VTAA]
\label{prop:VTAA}
Consider a PD-QLS where we have access to a normalised matrix-encoding of $B = \Id - \eta\,A$ and to a state preparation unitary for $\b$. Then, there is a VTAA-based solver having target precision $\varepsilon$, constant success probability, and query complexities given by
\al{
	& Q[\U_\b] 
	\in 
	\O\!\left(
	\sqrt{\log(\kappa)} +
	\frac{ \kappa}{\norm{A^{-1}\ket{\b}}}
	\right) \\
	& Q[\U_B] 
	\in 
	\O\!\left(
	\sqrt{\frac{\kappa}{\eta}} \, 
	\Gamma_{A,\b} \,
	\t{polylog}(\kappa, {\tilde \epsilon}^{-1}, \eta^{-1})
	\right) \\
	\t{with} \quad &
	\Gamma_{A,\b} 
	:= 
	\sqrt{\kappa} \,
	\frac{\norm{A^{-1/2}\ket{\b}}}{\norm{A^{-1}\ket{\b}}} ,
	\qquad
	{\tilde \epsilon} 
	\in 
	\O\!\left(
	\frac{\varepsilon}{\kappa \sqrt{\log \kappa}}
	\right)
	\\[1mm]
	\t{and} \quad &
	\t{polylog}(\kappa, {\tilde \epsilon}^{-1}, \eta^{-1})
	=
	\log^{2} (\kappa) 
	\log^{7/4}({\tilde \epsilon}^{-1}) 
	\log^{3/2} (\eta^{-1})
	.
}
Moreover, the algorithm is gate efficient. 
\end{proposition}

Now we discuss the runtime of this VTAA PD-QLS solver.

\begin{itemize}

\item We always have $Q[\U_B] \geq Q[\U_\b]$, thus the $\U_B$-complexity is the dominant factor.

\item Compared to \prop{Algorithm1}, the $\U_\b$ query complexity increases here only by an additive $\sqrt{\log(\kappa)}$ factor, while the $\U_B$ complexity typically (i.e., for almost all values of $\Gamma_{A,\b}$) has a polynomial improvement. To prove it, note that $\norm{A^{-1/2}\ket{\b}}^2 = \bra{\b}A^{-1} \ket{\b} \leq \kappa$ and thus $\Gamma_{A,\b} = \sqrt{\kappa}\,\frac{\norm{A^{-1/2}\ket{\b}}}{\norm{A^{-1}\ket{\b}}} \leq  \frac{\displaystyle \kappa}{\norm{A^{-1}\ket{\b}}}$. 

\item Compared to the general VTAA-based QLS solver of Ref.~\cite{Childs15}, our PD-QLS solver has a polynomial speed-up for almost all values of $\Gamma_{A,\b}$, see the right plot in \fig{Runtimes}. This is a consequence of \lem{Gamma}, where we prove that $\Gamma_{A,\b} \in [\,1,\sqrt{\kappa}\,]$.

\end{itemize}

\section{Method based on a quadratic reduction of the condition number via matrix decomposition}
\label{sec:hamiltonians}

In this section we start giving some preliminary considerations on the approach that we are going to present (\sec{preliminary}) and then give a formal statement of the Sum-QLS problem that we solve (\sec{prob_def}). Next, we describe a classical pre-processing step that quadratically reduces the condition number (\sec{positive}) and then present the quantum algorithm solving the pseudo-inversion problem that originates from the preconditioning (\sec{pseudo-inverse}). Finally, we estimate the gate complexity of the resulting Sum-QLS solver (\sec{runtime_estimation}).

\subsection{General considerations}
\label{sec:preliminary}

We present now the general features of any algorithm that solves PD-QLS exploiting a decomposition of the form $A = LL^\dag$, as already summarised in \sec{overview2}. Note that such $L$ exists for any PD matrix $A$ and that it may not be unique, especially if we allow $L$ to be non-square. The key property is that any $L$ such that $A = LL^\dag$ satisfies $\kappa_\t{eff}(L) = \sqrt{\kappa(A)}$, hence a system of the form $L^\dag\x = \b'$ (for any $\b'$) is quadratically better conditioned than the original system $A\x = \b$.

The method we introduce is based on finding matrices $L \in \mathds{C}^{N\times M}$ and $L^g \in \mathds{C}^{M \times N}$, with $M > N$, such that the decomposition $A = L L^\dag$ holds and $L^g$ satisfies $LL^g = \Id$, i.e.\ it is a right pseudo-inverse; moreover, $L^g$ is rectangular with more rows than columns and thus $L^gL \neq \Id$, i.e.\ $L^g$ cannot be a left pseudo-inverse. We make then the following observations. \vspace{-1mm}

\begin{enumerate} [itemsep=-.5\parsep]
\item $L^g L L^\dag \x = L^g \b$ is a linear system equivalent to the original one, having the vector $\x = A^{-1}\b$ as the unique solution, but with no guarantee that the condition number of $L^g L L^\dag$ is small.
\item $L^\dag \x = L^g \b$ is a over-constrained linear system that may be inequivalent to the original one (since $L^g L \neq \Id$) and typically has no proper solution $\x$.
\item Finding $\t{argmin}_\x \norm{\, L^\dag \x - L^g \b\,}$ is a problem equivalent to the original system. The unique solution is $\x = (L^\dag)^+ L^g \b$, therefore using\footnote{By assumption $A$ is invertible, hence $L^\dag \in \mathds{C}^{M \times N}$ is full-rank, and thus using the SVD $L^\dag= W \Sigma^\dag V^\dag$ we get $(L L^\dag)^{-1} L = (V \Sigma \Sigma^\dag V^\dag)^{-1} V \Sigma W^\dag = V (\Sigma^\dag)^+ W^\dag = (L^\dag)^+$.} $(L^\dag)^+ = (L L^\dag)^{-1} L = A^{-1} L$ and $LL^g = \Id$ we get the required result $\x = A^{-1} LL^g \b = A^{-1}\b$.
\end{enumerate} \vspace{-1mm}

The goal is thus to convert the linear system $A\x=\b$ into the linear regression problem $\t{argmin}_\x \norm{\, L^\dag \x - L^g \b\,}$, having solution $\x = (L^\dag)^+ L^g \b$. This is a non-trivial task, as we need, given access to $A$ via sparse-matrix oracle or via some succinct description, to construct a suitable access to $L^\dag$ and, moreover, given access to $\b$, to construct a suitable access to $\b' := L^g\b$. The latter requirement seems particularly worrisome, since it involves a pseudo-inversion of the exponentially large matrix $L$. We show, however, that for the Sum-QLS problem, whereby $A$ is provided as a sum of local PD terms, one can find a suitable $L^g$ for which a compact classical description can be efficiently computed.

We also remark that a pseudo-inversion problem can be interpreted as a regular matrix inversion on the subspace where the matrix $L^\dag$ is full-rank; thus, solving pseudo-inversion entails a larger runtime compared to solving a standard QLS, since the appropriate subspace has to be selected via projection or via amplitude amplification. More formally, the operator $(L^\dag)^+\!: \mathds{C}^M \rightarrow \mathds{C}^N$ (with $M > N$) has rank equal to $N$ and thus we have the orthogonal decomposition
\al{
	\mathds{C}^M 
	= 
	\t{supp}\big((L^\dag)^+\big) +
	\t{ker }\big((L^\dag)^+\big) 
	\qquad
	\begin{cases} 
	\dim \t{supp}\big((L^\dag)^+\big) = N \\
	\dim \t{ker }\big((L^\dag)^+\big) = M - N
	\end{cases}	
}
where the support is by definition the subspace orthogonal to the kernel. Then, calling $\Pi$ and $\Pi^\perp = \Id - \Pi$ the orthogonal projectors on the support and on the kernel of $(L^\dag)^+$, respectively\footnote{Using the identity $(L^\dag)^+ = A^{-1} L$ we obtain $\t{supp}((L^\dag)^+) = \t{supp}(L)$ and moreover $\t{ker}((L^\dag)^+) = \t{ker}(L)$.}, we obtain the identity
\al{
	(L^\dag)^+ \ket{\b'}
	=
	(L^\dag)^+ \, \Pi       \ket{\b'} +
	(L^\dag)^+ \, \Pi^\perp \ket{\b'}
	=
	(L^\dag)^+ \, \Pi \ket{\b'}	
}
since by definition we have $(L^\dag)^+\, \Pi^\perp = 0$. It is then evident that only the component $\Pi \ket{\b'}$, which is in general a sub-normalised quantum state, plays a role in the pseudo-inversion algorithm, while the orthogonal component $\Pi^\perp \ket{\b'}$ can be arbitrary. Therefore, any quantum pseudo-inversion algorithms implicitly requires the amplification of the $\Pi \ket{\b'}$ component, which therefore entails a gate complexity in $\O(1/\sqrt{\gamma})$, for some known lower bound $\sqrt{\gamma} \leq \norm{\Pi \ket{\b'}}$.

\subsection{Problem statement}
\label{sec:prob_def}

In the Sum-QLS we assume that the coefficient matrix $A\in \mathds{C}^{N\times N}$ is given by an \emph{explicit classical description}, rather than via oracular access. This allows us to evade the lower bounds given in \prop{lower_bound}, since those bounds are formulated for relativising (i.e.\ oracular) algorithms. We assume, specifically, that $A$ has the form~\cite{Chowdhury16}
\al{
\label{eq:A_sum}
	A = \sum_{j=1}^J H_{(j)}
	\qquad 
	\forall j \ H_{(j)} \ \textup{is positive definite.} 
}
Here we impose that each $H_{(j)}$ is strictly positive definite (rather than semi-definite) because of a technical condition that will become clear later; in essence, the expression in Eq.~\eqref{eq:gamma_main} can diverge if any $H_{(j)}$ is singular, resulting in an infinite runtime. Each term $H_{(j)}$ is a local Hamiltonian, i.e.\ it can be expressed as [see also Eq.~\eqref{eq:footnote} in the footnote]
\al{
\label{eq:H_j}
	H_{(j)} = 
	h_{(j)} \otimes I_{\neg \S_j}
}
where each $h_{(j)}$ is an operator acting on a subset $\S_j \subseteq \{1,\ldots, n\}$ of at most $s$ qubits, corresponding to a matrix of size at most $2^{s} \times 2^{s}$. We assume that $J$, the number of Hamiltonian terms, and $s$ are ``small'', i.e.\ we take $J \in \O(\t{poly}\,n)$ and $s \in \O(\log n)$, and that we have a complete classical description of each operator $h_{(j)}$ and of each subset $\S_{j}$. The matrix $A$ is fully specified with at most $J 2^{2s}$ real parameters and with $J n$ boolean values (defining the sets $\S_j$), i.e.\ the number of parameters is $J 2^{2s} + Jn \in \O(\t{poly}\,n)$.

We require moreover that the known-term vector $\b$ is sparse, containing at most $d_\b$ non-zero entries in the computational basis, where $d_\b$ also scales polynomially in $n$. This implies that a preparation circuit $\U_\b$ can be given as an explicit small quantum circuit. This leads to the following definition for the Sum-QLS problem.

\begin{definition}[Sum-of-Hamiltonians Quantum Linear System]
\label{def:Sum-QLS}
A Sum-QLS problem is a PD-QLS as in \defin{PD-QLS} with the following restrictions. The coefficient matrix $A \in \mathds{C}^{N \times N}$, for $N = 2^n$, is provided as the sum of PD Hamiltonian terms $A = \sum_{j=1}^J H_{(j)}$, where each $H_{(j)}$ acts on at most $s$ qubits; each $H_{(j)}$ is fully specified as a PD matrix $h_{(j)}$ of size $2^{s_j} \times 2^{s_j}$ with $s_j \leq s$, together with the set $\S_j$ of $s_j$ qubits on which $H_{(j)}$ acts upon. The vector $\b \in \mathds{C}^N$ is $d_\b$-sparse and the value and position of each of the $d_\b$ non-zero entries is provided. 

\end{definition}

In \app{non-dequantization} we show that the Sum-QLS problem is $\mathsf{BQP}$-hard, by adapting a proof given in HHL~\cite{HHL}. That is, we show that any polynomial-time quantum computation (in the $\textsf{BQP}$ class) can be re-formulated as a Sum-QLS problem for some artfully constructed coefficient matrix $A$ and known-term vector $\b$; therefore no polynomial-time classical probabilistic algorithm (in the $\textsf{BPP}$ class) can solve the Sum-QLS problem\footnote{In this context, a classical probabilistic algorithm ``solves'' a QLS problem if it outputs a value $n\in \{1, \ldots, N\}$ with probability approximately equal to $|x_{n}|^2$, the square of the $n$-th amplitude of the quantum state $\ket{\x} = \ket{A^{-1}\b}$.} (unless $\textsf{BPP} = \textsf{BQP}$).

\subsection{Classical pre-processing step}
\label{sec:positive}

In this section we describe the classical pre-processing step, providing a quadratic improvement of the condition number. We decompose each matrix $h_{(j)}$ as 
\al{
	h_{(j)} = l_{(j)} l_{(j)}^\dag \,,
}
which can be accomplished for example via the Cholesky decomposition~\cite{Cholesky}, in which case $l_{(j)}$ is a lower triangular matrix. Notice that each matrix $h_{(j)}$ is a small matrix of size $2^s \times 2^s \in \O(\t{poly}\, n)$ and thus the Cholesky decomposition can be performed numerically on a classical computer using $\O(\t{poly}\,n)$ operations; specifically, Cholesky factorisation of a $m\times m$ matrix requires $m^3/3$ arithmetic operations, $m^3/6$ additions and $m^3/6$ multiplications~\cite{Cholesky}. The total number of Hamiltonian terms is $J \in \O(\t{poly}\,n)$, implying that the total runtime for performing the Cholesky decomposition for all Hamiltonian terms is $\O(2^{3s} J)$, which also is polynomial in $n$ under our assumptions.

We save the Cholesky decompositions $l_{(j)}$ in a classical memory, storing $\O(J 2^2s) = \O(\t{poly}\,n)$ complex values, for later use. These decompositions implicitly define the operators 
\al{
	L_{(j)} =  l_{(j)} \otimes I_{\neg \S_j}
}
where each $L_{(j)}$ is of size $2^n \times 2^n$ and where the interpretation of this equation is the same as in Eq.~\eqref{eq:H_j}. We now introduce the rectangular matrix $L \in \mathds{C}^{N \times JN}$, with $N=2^n$, given by
\al{
\label{eq:L}
	L
	:=
	\left(
	L_{(1)} \ \bigg| \ \cdots \ \bigg| \ L_{(J)}
	\right)
}
and thus we have the required decomposition
\al{
	L L^\dag 
	= 
	\sum_{j=1}^J L_{(j)}L_{(j)}^\dag 
	= 
	\sum_{j=1}^J  H_{(j)} 
	= 
	A \,.
}
We can efficiently implement quantum circuits $\P_{L}$ ($\P_{L^\dag}$) that provide sparse-matrix access to $L$ ($L^\dag$). To this end, it is sufficient to convert the classical random access memory that stores the positions and values of the entries of $L$ ($L^\dag$) into a qRAM~\cite{Giovannetti08}. The scheme presented in Ref.~\cite[Section 6.3.5]{Gilyen14} implements this qRAM with a gate complexity in $\O(n J 2^{2s})$ and a circuit depth in $\O\big(\log (J 2^{2s})\big) = \O\big(s + \log J\big)$.

Since $h_{(j)}$ is by assumption non-singular, each $l_{(j)}$ is a non-singular lower-triangular matrix and its inverse $l_{(j)}^{-1}$ is an upper-triangular matrix which we can efficiently compute and store in a classical memory using a polynomial amount of space. These matrices implicitly define operators $L_{(j)}^{-1} =  l_{(j)}^{-1} \otimes I_{\neg \S_j}$ such that $L_{(j)} L_{(j)}^{-1} = \Id$. We can then define the matrix
\al{
\label{eq:L^g}
	L^g := 
	\frac{1}{J}
	\left(
	\begin{array}{c}
	L_{(1)}^{-1} \\[1mm]
	\hline
	\vdots \\[1mm]
	\hline 
	L_{(J)}^{-1}
	\end{array}
	\right) .
}
which is a generalised right pseudo-inverse of $L$, i.e.\ $L^g$ satisfies the equation 
\al{
	L \, L^g = \frac{1}{J} \sum_{j=1}^J  \Id = \Id
}
and thus using $(L^\dag)^+ = A^{-1} L$ we get the required relation $(L^\dag)^+ L^g = A^{-1} L \, L^g = A^{-1}$. Using a qRAM the gate complexity of $\P_{L^g}$ is equal to that of $\P_{L}$ and is thus in $\O(n J 2^{2s})$.

We finally introduce the quantum state
\al{
\label{eq:b'_definition}
	\ket{\b'} :=
	\ket{L^g \, \b} .
}
By assumption, $\b$ is sparse and has $d_\b \in \O(\t{poly}\, n)$ non-zero entries, hence the vector $\b'$ has sparsity $d_{\b'} \leq d_\b J 2^s \in \O(\t{poly}\,n)$. It is then possible efficiently classically compute the positions and the values of all the non-zero entries of $\b'$, with a gate complexity in $\O(n \, d_{\b'})$, and thus also compute the normalisation factor $\norm{\b'}$, with a gate complexity in $\O(d_{\b'})$. Using the method described in Ref.~\cite{Malvetti20}, we can then efficiently compile a quantum circuit $\U_{\b'}$ that prepares $\ket{\b'}$ and has a gate complexity in $\O(n\, d_{\b'}) = \O(n\,d_\b J 2^s)$, assuming that all single qubit rotations are performed exactly.

Employing the classical pre-processing described up to now, we can efficiently implement quantum circuits $\P_{L^\dag}$ and $\U_{\b'}$ that act as a sparse access to $L^\dag$ and state preparation circuit for $\b'$, respectively; importantly, the gate complexities of these unitaries are independent from $\kappa$. We thus have at hand the necessary tools to implement the quantum pseudo-inversion algorithm that we present in the upcoming \sec{pseudo-inverse}.

\subsection{Efficient pseudo-inversion quantum algorithm}
\label{sec:pseudo-inverse}

In this section we look into a quantum algorithm for the linear regression problem
\al{
\label{eq:psedo-QLS}
	\underset{\x}{\t{argmin}} 
	\norm{\,L^\dag \x - \b' \,} 
}
having solution $\ket{\x} = \ket{(L^\dag)^+ \b'} = \ket{A^{-1} \b}$. We then employ the quantum pseudo-inversion algorithm of~\cite[Corollary 31]{Chakraborty18} which we report here for completeness.

\begin{proposition}[Complexity of pseudo-inverse state preparation] 
\label{prop:pseudo-inverse}

Suppose $\widetilde{\kappa} \geq 2$, $\mathcal{L} \in \mathds{C}^{N \times N}$ is a Hermitian matrix whose non-zero eigenvalues are contained in the domain $[-1, -1/\widetilde{\kappa}] \cup [1/\widetilde{\kappa},1]$, and $\varepsilon$ is the target precision. Assume that we have access to $\mathcal{U}_\mathcal{L}$, an $(\alpha, a, \delta)$-matrix-block-encoding of $\mathcal{L}$ with $\delta \in \mathfrak{o} \left( \varepsilon / (\widetilde{\kappa}^2 \log^3 \frac{\widetilde{\kappa}}{\varepsilon}) \right)$ and $a \in \Omega(\log N)$, and to a state preparation oracle $\mathcal{U}_\textbf{v}$ for a vector $\textbf{v}$ such that $\norm{\Pi_\mathcal{L} \ket{\textbf{v}}} \geq \sqrt{\gamma}$, where $\Pi_\mathcal{L}$ is the orthogonal projector onto the support of $\mathcal{L}^+$ and $\gamma$ is a known positive parameter. Then, there is a (VTAA-based) quantum algorithm that produces a state $\varepsilon$-close to $\ket{\mathcal{L}^+ \textbf{v}}$ and has:
\al{
\label{eq:1}
	Q\left[\U_\mathcal{L}\right]
	\; & \in \; 
	\O \! \left(
	\frac{\alpha}{\sqrt{\gamma}} \, 
	\widetilde{\kappa} \log^3(\widetilde{\kappa}) \log^2 (1/\varepsilon)
	\right) \\
\label{eq:2}
	Q\left[\U_\textbf{v}\right]
	\; & \in \; 
	\O \! \left(
	\frac{1}{\sqrt{\gamma}} \, \widetilde{\kappa} \log(\widetilde{\kappa})
	\right) 
}
where $Q[\U_\mathcal{L}]$ and $Q[\U_\b]$ are the query complexity in terms of access to $\U_\mathcal{L}$ and $\U_\b$. The algorithm is gate-efficient, only requiring $\O \big(a \, Q\left[\mathcal{U}_\mathcal{L}\right]\big)$ extra elementary gates.

\end{proposition}

We will apply \prop{pseudo-inverse} using as coefficient matrix $\mathcal{L}$ the Hermitian extension of $L^\dag$ and $\widetilde{\kappa} \equiv \sqrt{\kappa}$. That is, we consider the Hermitian matrix $\mathcal{L} \in \mathds{C}^{(J+1)N \times (J+1)N}$ and vector $\textbf{v} \in \mathds{C}^{(J+1)N}$ 
\al{
	\mathcal{L} = 
	\left(
	\begin{array}{cc}
	0 & L^\dag \\
	L & 0
	\end{array}
	\right)
	\quad	
	\textbf{v} =
	\left(
	\begin{array}{c}
	\b' \\ 0
	\end{array}
	\right)
}
which, after pseudo-inversion, results in the vector
\al{
	\x 
	= 
	\mathcal{L}^+ \textbf{v} 
	= 
	\left(
	\begin{array}{cc}
	0 & L^+ \\
	(L^\dag)^+ & 0
	\end{array}
	\right)
	\left(
	\begin{array}{c}
	\b' \\ 0
	\end{array}
	\right)
	=
	\left(
	\begin{array}{c}
	0 \\ (L^\dag)^+\, \b'
	\end{array}
	\right)
}
which encodes the quantum state $\ket{J} \otimes \ket{(L^\dag)^+\, \b'}$ and the solution is obtained discarding the $(J+1)$-level ancilla system in the state $\ket{J}$.

\subsection{Runtime estimation}
\label{sec:runtime_estimation}

In this section we look into methods for estimating of the parameters $\alpha,\kappa$ and $\gamma$ (i.e., the normalisation of the block-encoding of $L^\dag$, condition number, and overlap with the support space) that determine the runtime of the pseudo-inversion algorithm of \prop{pseudo-inverse}, and thus determine the overall complexity of the Sum-QLS solver.

First, we can implement sparse-matrix-accesses $\P_L$ and $\P_{L^\dag}$ using the information stored in a qRAM, as previously explained. Childs' walk operator~\cite{Childs15} then allows to realise a block-encoding $\U_\mathcal{L}$ of $\mathcal{L}$ using $\O(1)$ accesses to $\P_L$ and $\P_{L^\dag}$. Note that $\U_\mathcal{L}$ is also a block-encoding of $L^\dag$ (after a swap of the position of the block). The normalisation factor of this block-encoding is $\alpha = J 2^s \in \O(\t{poly}\,n)$, equal to the sparsity of $\mathcal{L}$.

Second, we can explicitly bound the condition number of $A$ as follows. Positive-definiteness of the Hamiltonian terms $H_{(j)}$ implies positive-definiteness of $A$ and, moreover, the smallest and largest eigenvalues of $A$ satisfy $\lambda_{\min}(A) \geq \sum_{j=1}^J \lambda_{\min}(h_{(j)})$ and $\lambda_{\max}(A) \leq \sum_{j=1}^J \lambda_{\max}(h_{(j)})$. Since each $h_{(j)}$ can be efficiently diagonalised, this means that it is possible to classically compute these values, which then yield the explicit upper bound 
\al{
\label{eq:kappa_bound}
	\kappa(A) 
	\leq 
	\frac{\sum_{j=1}^J \lambda_{\max}(h_{(j)})}
	     {\sum_{j=1}^J \lambda_{\min}(h_{(j)})}
	\equiv \kappa
	\,.
}
Tighter bounds to $\kappa(A)$ could also be obtained via more computationally intensive numerical methods, e.g.\ by first summing together groups of Hamiltonian terms and then diagonalizing each sum of Hamiltonians.

We now move on to lower-bounding the value of the overlap parameter 
\al{
	\bnorm{\, \Pi_\mathcal{L} \ket{\textbf{v}} }
	=
	\bnorm{\,\Pi_L \ket{\b'} } ,
}
where $\Pi_\mathcal{L}$ are the projectors on the supports of $\mathcal{L}^+$ and of $(L^\dag)^+$, respectively, and note that the supports of $L$ and $(L^\dag)^+$ are equal. Using the identity $\Pi_{L} = L^+ L$ we thus obtain
\al{
	\Pi_{L} 
	= 
	L^\dag A^{-1} L 
	& = 
	\sum_{i,j=1}^J
	\ket{i}\! \bra{j} \otimes
	L_{(i)}^\dag A^{-1} L_{(j)} \,.
}
Moreover, we have:
\al{
	\ket{\b'} = 
	\ket{L^g \, \b} = 
	\frac{1}{\sqrt{\N}} 
	\sum_{j=1}^J \ket{j} \otimes L_{(j)}^{-1} \ket{\b}
}
where the normalisation factor $\N$ is given by
\al{
	\N 
	= 
	\sum_{j=1}^J \norm{ L_{(j)}^{-1} \ket{\b} }^2 
	=
	\sum_{j=1}^J \bra{\b}  H_{(j)}^{-1} \ket{\b} .	
}
Then we can compute
\al{
	\Pi_L \ket{\b'}
	& =
	\frac{1}{\sqrt{\N}}
	\sum_{i,j=1}^J
	\ket{i} \otimes L_{(i)}^\dag A^{-1} L_{(j)} L_{(j)}^{-1} \ket{\b} \\
	& =
	\frac{J}{\sqrt{\N}}
	\sum_{i=1}^J
	\ket{i} \otimes L_{(i)}^\dag A^{-1} \ket{\b}
}
and finally we obtain
\al{
\label{eq:gamma1}
	\norm{\Pi_L \ket{\b'}}^{-1}
	& =
	\frac{\sqrt{\N}}{J}
	\left[
	\bra{\b} A^{-1} 
	{\textstyle \sum_{i=1}^J} 
	\big( L_{(i)} L_{(i)}^\dag \big) A^{-1} \ket{\b} 
	\right]^{-1/2}
	\\
\label{eq:gamma2}
	& =
	\frac{1}{J}	\,
	\sqrt{ 
	\frac{ \sum_{j=1}^J \bra{\b} H_{(j)}^{-1} \ket{\b}}
	     {\bra{\b} A^{-1} \ket{\b}} 
	     }
	\;.
}

We now suppose that a value $\gamma > 0$ such that $\norm{\Pi_\mathcal{L} \ket{\textbf{v}}} \geq \sqrt{\gamma}$ is known and we remind that the quantum psuedo-inversion algorithm has a runtime quasi-linear in $1/\sqrt{\gamma}$. We extensively comment on the values that $\gamma$ can take in order to understand in which cases the Sum-QLS solver yields an advantage over competing methods.

\begin{enumerate}

\item We have $\norm{\Pi_L \ket{\b'}} \leq 1$, and the inequality is saturated when $H_{(j)} = A / J$ for all $j$.

\item The bound $\sum_{j=1}^J \bra{\b} H_{(j)}^{-1} \ket{\b} \leq J\,\lambda_*^{-1}$ holds, where $\lambda_* := \min_j \lambda_{\min} (H_{(j)})$. Assuming that $\lambda_* \in \Omega \big( \lambda_{\min}(A)/J \big)$, i.e.\ there is no Hamiltonian term having a minimum eigenvalue significantly smaller than the average minimum eigenvalue, we obtain:
\al{
\label{eq:gamma_bound}
	\norm{\Pi_L \ket{\b'}}^{-1}
	& \in 
	\O \bigg(
	\frac{1}{J}	\,
	\frac{ \sqrt{J \lambda_*^{-1}} }
	     { \sqrt{\bra{\b} A^{-1}\ket{\b}} } 
	\bigg)
	= 
	\O \bigg(
	\frac{ \sqrt{ \kappa(A) } }
	     { \sqrt{\bra{\b} A^{-1}\ket{\b}} } 
	\bigg) \;.
}

\item The numerator in Eq.~\eqref{eq:gamma2} can be explicitly calculated, while the denominator is in general difficult to compute\footnote{One could use techniques related to amplitude estimation to bound $\big|\big|A^{-1/2} \ket{\b}\big|\big|$, but this operation could be as difficult as solving the QLS in the first place.}. However, assuming $\norm{A} \leq 1$, we have the bounds 
\al{
	\sqrt{\bra{\b} A^{-1}\ket{\b}} = \norm{A^{-1/2} \ket{\b}} \in \big[\,1, \sqrt{\kappa}\,\big]
}

\item Importantly, the expression $\bra{\b} A^{-1}\ket{\b}$ appears at the denominator, so that a more ``ill-conditioned'' vector $\b$ results in a larger overlap and thus in a faster Sum-QLS solver.

\item As in the analysis of \sec{from-to} we can study the runtime in an average-case scenario. For randomly chosen $A$ and $\b$ (sampled according to suitable probability distributions) we have $\norm{A^{-1}\ket{\b}} \in \Theta(\sqrt{\kappa(A)})$ almost surely; under the same assumptions, we also have $\norm{A^{-1/2}\ket{\b}} \in \Theta(\kappa(A)^{1/4})$ almost surely. Inserting this estimation in Eq.~\eqref{eq:gamma_bound} we have that
\al{
	\norm{\Pi_L \ket{\b'}}^{-1}
	& \; \in \; 
	\Theta \Big( \kappa(A)^{1/4} \Big)	
}
holds almost surely. We conclude that for an average-case Sum-QLS problem the runtime is in $\widetilde{\O}(\sqrt{\kappa/\gamma}) = \widetilde{\O}(\kappa^{3/4})$, if $\sqrt{\gamma}$ is a tight lower bound for $\norm{\Pi_L \ket{\b'}}$.

\end{enumerate}

Summarising, we have the following result.

\begin{proposition}[Complexity of the Sum-QLS solver]
\label{prop:Algorithm2}
Consider the Sum-QLS problem with parameters $N,J,s,\kappa$ and $d_\b$ as in \defin{Sum-QLS}. There is a classical-quantum algorithm $\mathcal{A}$ solving the Sum-QLS that has the following features.

The first part of $\mathcal{A}$ consists of an efficient classical pre-processing algorithm, which outputs a description of the quantum circuits implementing $\U_{L^\dag}$ and $\U_{\b'}$; here $\U_{L^\dag}$ is a $(2^sJ,1,0)$-matrix-block-encoding of $L^\dag$ [as given in Eq.~\eqref{eq:L}] with gate complexity in $\O(n J 2^{2s})$ and circuit depth in $\O\big(s + \log J\big)$; while $\U_{\b'}$ is a state preparation unitary for $\b'$ [as given in Eq.~\eqref{eq:L^g}] with gate complexity in $\O(n\,d_\b \,J 2^s)$. By definition, $L^\dag$ and $\b'$ satisfy $(L^\dag)^+ \b' = A^{-1}\b$. 

The second part of $\mathcal{A}$ consists of using the quantum pseudo-inversion algorithm given in \prop{pseudo-inverse}, using the unitaries $\U_{L^\dag}$ and $\U_{\b'}$ as sub-routines, in order to produce a $\varepsilon$-close approximation of the ideal output $\ket{\x} = \ket{A^{-1}\b}$. This algorithm requires the knowledge of a value $\sqrt{\gamma}$ that satisfies $\sqrt{\gamma} \leq \norm{\Pi_L \b'}$, where $\Pi_L$ is the projector onto the support of $L$, equivalently:
\al{
	\frac{1}{J^2}
	\frac{ \sum_{j=1}^J \bra{\b} H_{(j)}^{-1} \ket{\b}}
	     {\bra{\b} A^{-1} \ket{\b}} 
	\leq 
	\frac{1}{\gamma} \,.
	\label{eq:gamma_main}
}
Inserting the previously given expressions for the relevant parameters in Eqs.~(\ref{eq:1}-\ref{eq:2}) results in
\al{
	\textup{gate complexity} 
	& \in
	\O\! \left( 
	\big(n\, J 2^{2s}\big) \, \frac{J 2^s}{\sqrt{\gamma}}\, 
	\kappa \log^3(\kappa) \log^2(1/\varepsilon) 
	\,+\,
	\big(n\, d_\b J 2^s\big)
	\frac{1}{\sqrt{\gamma}}\,
	\kappa \log(\kappa)
	\right) 
	\label{eq:intermediate_estimation}\\
	& =
	\O\! \left( 
	\sqrt{\frac{\kappa}{\gamma}} \;
	\t{poly}\Big(n,\log(\kappa/\varepsilon)\Big) 
	\right)
	\label{eq:final_estimation}
}
assuming that $J, 2^s$ and $d_\b$ have polynomial dependence on $n = \log_2 N$. A quadratic speed-up in $\kappa$ (up to polylogarithmic factors) is achieved over general QLS solvers when $\gamma \in \Omega(1)$.

\end{proposition}

We remark that the family of Sum-QLS instances where all the parameters in \prop{Algorithm2} scale polynomially (with the promise, in particular, that Eq.~\eqref{eq:gamma_main} holds for some $\gamma \in \O(\t{poly}\,n)$), can be solved in polynomial time on a quantum computer, as was already shown in Ref.~\cite{Chowdhury16}. This means that the subset of problems having a polynomial scaling of the parameters, which we denote Sum-QLS$_\t{poly}$, is contained in BQP. Moreover, the reduction in \app{non-dequantization} can map polynomial-sized quantum circuits onto an instance of Sum-QLS$_\t{poly}$, thus showing that Sum-QLS$_\t{poly}$ is also \textsf{BQP}-hard. These two inclusion then show that Sum-QLS$_\t{poly}$ is \textsf{BQP}-complete\footnote{While the original QLS problem was proven to be BQP-complete already in Ref.~\cite{HHL}, our contribution is to show that adding the constraint that $A$ is the sum of positive-definite local Hamiltonians does not change the complexity class of the problem. }.

\section{Discussion and outlook}
\label{sec:discussion}

In this work we have presented two algorithms aiming at solving QLS problems in the case where the coefficient matrix is positive definite and having (for certain problem instances) a runtime in $\O(\sqrt{\kappa})$, a quadratic improvement compared to what can be obtained using general QLS solvers. This improvement has the potential of greatly expanding the classes of problems where quantum computation can provide a quantum speed-up. For instance, the discretization of partial differential equations in $D$ dimensions results in PD linear system with $\kappa \in \O(N^{2/D})$~\cite{Montanaro16} and thus having a runtime improvement from $\O(\kappa)$ to $\O(\sqrt{\kappa})$ is crucial to yield a quantum speed-up in the physically relevant cases $D=2$ and $D=3$. As a second example, it is possible estimate the hitting time of a Markov chain solving a QLS for the matrix $A = \Id - S$, where $S$ is related to the discriminant matrix of the Markov chain~\cite{Chowdhury16}; since $A$ is positive definite and decomposable as a sum of PD local Hamiltonian terms~\cite[Appendix A]{Chowdhury16} our second algorithm could be applicable to this problem.

In the spirit of finding in the near future real-world applications of quantum algorithms, we note that there is considerable interest in the possibility of realising QLS solvers in Noisy Intermediate-scale Quantum (NISQ) devices~\cite{Bravo19, Huang19} and we argue that some of our results might be implementable in NISQ devices too. In particular, the crux of our first algorithm is to find a good polynomial approximation for $A^{-1}$ with a degree in $\O(\sqrt{\kappa})$ and then implement it with the quantum signal processing method~\cite{Low16,Gilyen18}. The quadratic reduction in the degree of the polynomials renders their realisation more easily compatible with the next generation of quantum processors.

We note that many further improvements and extensions to our algorithms may be possible. Regarding the first algorithm (\sec{block_encoding}), it would be important to extend the classes of matrices for which a normalised matrix-block-encoding of $B = \Id - \eta\,A$ can be efficiently implemented to make the method more generally applicable. Regarding the second algorithm (\sec{hamiltonians}) we note that the specific choice of the generalised pseudo-inverse $L^g$ in the classical step results in a $\O(1/\sqrt{\gamma})$ multiplicative overhead in the runtime, where $\gamma$ is given in Eq.~\eqref{eq:gamma_main}. An open question is whether a different choice of the pseudo-inverse could improve, or eliminate altogether, this overhead. We also mention the possibility that the decomposition of $A$ as a sum of local PD Hamiltonians could be computed on-the-fly by the solver, instead of being given as an external input. We note that if the sparsity pattern satisfies certain conditions (it is a \emph{chordal} graph) a decomposition $A=L L^\dag$ that does not increase the sparsity of $A$ exists, and the characterisation given in Ref.~\cite[Theorem~2.6]{Jiang17} could be employed to compute it.

We finally mention an open research idea that may be worth investigating. The eigenpath transversal method has been used in some new algorithms to solve the QLS problem with time complexity in $\widetilde{\O}(\kappa)$~\cite{Subasi18, An19, Bravo19, Huang19, Lin19}; this method is simpler than the Variable-Time Amplitude Amplification (VTAA) method and also results in (marginally) improved runtimes, however, it is not directly applicable to solve a pseudo-inversion problem. It would be interesting to find a way to adapt the eigenpath transversal method to make it work also in the case where the coefficient matrix is singular. As a by-product, it could replace the algorithm given in Ref.~\cite{Chakraborty18} as the sub-routine used in our second algorithm to solve the pseudo-inversion problem, therefore making it more practical.


\section*{Acknowledgments}

This work was supported by the Dutch Research Council (NWO/OCW), as part of the Quantum Software Consortium programme (project number 024.003.037). Some ideas present in the paper originated from discussions with Anirban Chowdhury while DO was visiting the Center for Quantum Information and Control (CQuIC). The authors acknowledge discussion with Markus Mieth, Anderas Sp\"orl, and thank Andr\'as Gily\'en for reading the manuscript and for giving us several insightful comments and suggestions.

\section*{Appendices}

\appendix
\renewcommand{\theequation}{\thesection.\arabic{equation}}

\section{Proof of the query complexity lower bound}
\label{app:lowerbounds}
\setcounter{equation}{0}

In this Appendix we give the proof of the query complexity lower bound presented in \sec{lower_bound}, which we present here in a more extended form.

\begin{proposition}[Query complexity lower bound]
\label{prop:extended_bounds}
Consider oracular quantum algorithms that solve the PD-QLS problem as presented in \defin{PD-QLS} for different access models to $A$ and $\b$. Namely, access to $\b$ is given via a state preparation oracle $\U_\b$ (\defin{state_prep}), while access to $A$ is given either via a sparse-matrix oracle $\P_A$ (\defin{sparse_access}) or via a matrix-block-encoding $\U_A$ (\defin{U_A}). Then, PD-QLS solving algorithms reaching a constant precision $\varepsilon \in \O(1)$ have query complexities $Q[\U_\b], Q[\U_A], Q[\P_A]$ all in $\Omega\big(\min(\kappa, N) \big)$. More precisely, we have:
\begin{enumerate}[itemsep=-.8\parsep]
\item 
$Q[\U_b] \in \Omega\big(\min(\kappa, N) \big)$, independently from the access model for $A$;
\item 
$Q[\U_A] \in \Omega\big(\min(\kappa, N) \big)$, independently from the access model for $\b$, when $\U_A$ is a normalised matrix-block-encoding;
\item 
$Q[\P_A] \in \Omega\big(\min(\kappa, N) \big)$, independently from the access model for $\b$, when $A$ is a matrix with constant sparsity.
\end{enumerate}

\end{proposition}

In the main text we prove a weaker result using a reduction to the quantum search problem, which has a query complexity in $\Omega(\sqrt{N/M})$, where $M\in [N] := \{1,\ldots,N\}$ is the number of marked elements; here, we use instead a reduction to a ``promise majority'' problem, which has a query complexity in $\Omega(N/M)$, where $M \in [N]$ is the margin of the majority. As a result, we can prove that solving a PD-QLS has linear scaling of the query complexity in the condition number for all $\kappa \in \O(N)$. To prove \prop{extended_bounds}, we first introduce the \textsc{PromiseMajority}$_M$ problem as follows.

\begin{definition} [\textsc{PromiseMajority}$_M$]
Given a vector $y \in \{0,1\}^N$, a value $M \in [N]$ (we assume for simplicity that $N+M$ is even) and given the promise that we either have 
\begin{itemize}[left=18mm,itemsep=-.8\parsep]
\item[(Case 0)] $y_i = 0$ for $N/2 + M/2$ of the entries 
\item[(Case 1)] $y_i = 1$ for $N/2 + M/2$ of the entries 
\end{itemize}
the \textsc{PromiseMajority}$_M$ problem consists in determining which of the two is the case. 
\end{definition}

We assume that we have access to $y$ via a quantum oracle $\P_y$ that acts as $\P_y \ket{i,z} = \ket{i,z\oplus y_i}$ for all $i\in [N]$ and for $z\in \{0,1\}$. We also remind that the two-sided bounded-error quantum query complexity $\mathcal{Q}_2$ of a boolean function is defined as the minimum number of accesses to the input of the function (i.e., to $\P_y$) that are necessary to correctly output the value of the function with probability at least $2/3$, both for the positive and for the negative instances. Then we have the following Lemma:

\begin{lemma}
The two-sided bounded-error quantum query complexity $\mathcal{Q}_2$ of \textsc{PromiseMajority}$_M$, in terms of accesses to $\P_y$, is $\mathcal{Q}_2(\textsc{PromiseMajority}_M) \in \Omega(N/M)$.
\end{lemma}

\begin{proof}
This follows immediately from Ref.~\cite[Corollary 1.2]{Nayak99}.
\end{proof}

We now show that (relativising) PD-QLS solving algorithms can be used to compute \textsc{PromiseMajority}$_M$; the lower bound on the query complexity of \textsc{PromiseMajority}$_M$ directly translates into a lower bound on the query complexity of the PD-QLS solvers. We will prove separately the three cases of \prop{extended_bounds}, with each proof building upon the previous ones.

\begin{proof}\emph{Case 1.}

We assume that $y\in \{0,1\}^N$ is in the domain of \textsc{PromiseMajority}$_M$, i.e.\ $y$ either contains exactly $N/2+M/2$ zeros or $N/2+M/2$ ones, and we define $\b \in \mathds{C}^{N+1}$:
\al{
\left\lbrace
\begin{array}{l}
	b_i = (-1)^{y_i} \qquad \qquad \t{for} \ i\in [N] \\[2mm]
	b_{N+1} = \sqrt{N+M}
\end{array}
\right.
\label{eq:b_reference}
}
where the value $b_{N+1}$ is fixed to provide a ``phase reference'' and avoid ambiguity on the global sign. We have $\ket{\b} = \b/\sqrt{2N+M}$ and $\ket{\b}$ can be implemented by first preparing a state proportional to $(1,\ldots,1, \sqrt{N+M})$ and then applying the correct phases; this can be done with one oracle call to each of $\P_y$ and $\P_y^\dag$, via the transformations 
\al{
	\ket{i}
	\ \overset{\t{ancilla}}{\mapsto} \ 	 
	\ket{i,0} 
	\ \overset{\P_y}{\mapsto} \ 
	\ket{i,y_i} 
	\ \overset{\Id\otimes Z}{\mapsto} \
	(-1)^{y_i} \ket{i,y_i}
	\ \overset{\P_y^\dag}{\mapsto} \
	(-1)^{y_i} \ket{i,0}
	\ \overset{\t{discard}}{\mapsto} \
	(-1)^{y_i} \ket{i}
	\label{eq:transformations}
}
and extended by linearity to superpositions. Next, we introduce the vector $\boldsymbol{1}_{N} := (1, \ldots, 1)^T$ containing $N$ ones and then define
\al{
	K' \equiv K \oplus 0 := \frac{1}{N} \boldsymbol{1}_{N} \boldsymbol{1}_{N}^{\,T} \oplus 0
	\label{eq:K}
}
as a matrix of size $(N+1) \times (N+1)$, so that $K'^{\,2} = K'$, and finally 
\al{
	A := \Id - (1-\epsilon) K' 
	\label{eq:A_I-K}
}
where $\epsilon$ is a small parameter that we will define shortly. The matrix $A$ can be used as a coefficient matrix for a PD-QLS solver since $A$ is positive definite and $\norm{A}=1$. Moreover, the condition number of $A$ is exactly $\kappa(A) = 1/\epsilon$. 

Next we have:
\al{
	A^{-1} & = \Id + 
	\sum_{t=1}^\infty \big((1-\epsilon) K' \big)^t \\
	& = \Id + \frac{1-\epsilon}{\epsilon} K',
}
where the summation converges since $\norm{(1-\epsilon) K} < 1$. Let's apply $A^{-1}$ to $\b$:
\al{
	A^{-1} \b 
	& =
	\begin{cases}
	\b + \frac{1-\epsilon}{\epsilon} \frac{M}{N}\, \boldsymbol{1}_{N}' 
	& \t{if~} y \t{~has~a~majority~of~0} \\
	\b - \frac{1-\epsilon}{\epsilon} \frac{M}{N}\, \boldsymbol{1}_{N}' 
	& \t{if~} y \t{~has~a~majority~of~1}
	\end{cases} 
	\label{eq:x_ideal} \\
	& =
	\begin{cases}
	\b + \boldsymbol{1}_{N}' 
	& \qquad \ \ \, \t{if~} y \t{~has~a~majority~of~0} \\
	\b - \boldsymbol{1}_{N}' 
	& \qquad \ \ \, \t{if~} y \t{~has~a~majority~of~1}
	\end{cases}
}
where we have introduced $\boldsymbol{1}_N' := (1,\ldots,1,0)^T$ and we have chosen $\epsilon$ so that $\frac{1-\epsilon}{\epsilon} \frac{M}{N} = 1$, giving $\kappa(A) = \frac{1}{\epsilon} = \frac{N+M}{M}$. Introducing a boolean value $f =  \textsc{PromiseMajority}_M(y)$, i.e.\ $f \in \{0,1\}$ is equal to the majority of $y$, we can rewrite the vector $A^{-1}\b$ entry-wise as
\al{
\left\lbrace
\begin{array}{llc}
	{[A^{-1} \b]_i} & = (-1)^f \cdot 2 & 
	\qquad \t{if~} i \t{~is~such~that~} y_i = f \\
	{[A^{-1} \b]_i} & = 0 & \qquad \t{if~} i \t{~is~such~that~} y_i \neq f \\
	{[A^{-1} \b]_{N+1}} & = \sqrt{N+M} & 
\end{array}
\right.
}
Then we have $\ket{A^{-1}\b} = A^{-1}\b /\norm{A^{-1}\b}$, with 
\al{
	\norm{A^{-1}\b}^2 = 
	2^2 \,\frac{N+M}{2} + \sqrt{N+M}^{\,2} = 3(N + M)
}
so that we get
\al{
	\ket{A^{-1}\b} = 
	\sqrt{\frac{1}{3}} \ket{N+1} 
	\, + \, (-1)^f 
	\sqrt{\frac{2}{3}}  
	\sum_{i: y_i = f}\sqrt{\frac{2}{N+M}} \ket{i} .
	\label{eq:two_states}
}
We then perform a projective measurement where one of the possible measurement outcomes is 
\al{
	\ket{``+"} := 
	\sqrt{\frac{1}{2}} \ket{N+1} 
	\, + \,
	\sqrt{\frac{1}{2}}  
	\sum_{i=1}^{N} \frac{1}{\sqrt{N}}\ket{i} .
	\label{eq:plus_state}
}
The cases $f=0$ and $f=1$ in Eq.~\eqref{eq:two_states} can be distinguished with constant advantage, since:
\al{
	\braket{``+"}{A^{-1}\b} 
	& = 
	\sqrt{\frac{1}{6}} 
	\, + \, (-1)^f
	\frac{N+M} {2}  \sqrt{\frac{1}{3}} \sqrt{\frac{2}{N(N+M)}}\\
	& = 
	\frac{1}{\sqrt{6}} 
	\left( 
	1+ (-1)^f \sqrt{1+M/N}
	\right) .
}
Note that the two cases can still be distinguished with constant probability if we replace $\ket{\x} = \ket{A^{-1}\b}$ with any approximation $\rho_\x$ which is sufficiently close to it.

To summarise, suppose we have to solve a \textsc{PromiseMajority}$_M$ problem and that we can exploit as a subroutine an oracular quantum algorithm $\mathcal{A}$ that, given access to $\U_\b$, prepares the state $\ket{A^{-1}\b}$ with sufficiently high precision; suppose moreover that $\mathcal{A}$ has a query complexity $Q[\U_\b] =  g(\kappa)$, for some function $g: \mathds{R}^+ \rightarrow \mathds{N}$. Then, $\mathcal{A}$ can be used to prepare the state in Eq.~\eqref{eq:two_states} and solve \textsc{PromiseMajority}$_M$ with constant distinguishing advantage. The $\P_y$-query complexity of $\mathcal{A}$ is $Q[\P_y] = 2\, Q[\U_\b] = 2\,g\big(\kappa(A)\big) = 2\,g\!\left(\frac{N+M}{M}\right)$. The lower bound $\mathcal{Q}_2(\textsc{PromiseMajority}_M) \in \Omega(N/M)$ then directly implies $g(\kappa) \in \Omega\big(\min (\kappa, N) \big)$.

\end{proof}

\begin{proof}\emph{Case 2.}

We modify the construction given in the previous proof and encode the input $y$ in the entries of the coefficient matrix, with the goal of showing that $Q[\U_A] \in \Omega\big(\min(\kappa, N) \big)$. To this end, we define the vector $\textbf{u} \in \mathds{R}^{N+1}$ and a diagonal matrix $D \in \mathds{R}^{(N+1) \times (N+1)} $
\al{
\left\lbrace
\begin{array}{ll}
	u_i = 1 & \t{for} \ i\in [N] \\[2mm]
	u_{N+1} = \sqrt{N+M} &
\end{array}
\right.
\hspace{1cm}	
\left\lbrace
\begin{array}{ll}
	D_{i,i} = (-1)^{y_i} & \t{for} \ i\in [N] \\[2mm]
	D_{N+1,N+1} = 1 &
\end{array}
\right.
}
and notice the vector $\b$ in Eq.~\eqref{eq:b_reference} satisfies $\b = D \textbf{u}$. We also define the coefficient matrix $A'$
\al{
	A' := D A D
}
where $A$ is given in Eq.~\eqref{eq:A_I-K}. Note that $D$ is unitary and self-inverse, hence $A'$ is positive definite, $\kappa(A') = \kappa(A)$, and moreover $A'^{-1} = D^{-1} A^{-1} D^{-1} = D A^{-1} D$. 

It is possible to implement exactly (i.e., ideally with zero error) a normalised matrix block of $A'$ using at most $4$ calls to $\P_y$. First, we consider the unitary $\t{Had}'$ that prepares the state $\ket{\boldsymbol{1}'} = \ket{(1,1,\ldots,1,0)^T}$, that is $\t{Had}\ket{0} = \ket{\boldsymbol{1}'}$. Then, the matrix 
\al{
	\U_A :=
	\left(
	\begin{array}{cc}
	\!\t{Had}\! & 0 \\
	0 & \!\t{Had}\!
	\end{array}
	\right)
	\left(
	\begin{array}{cc}
	\Id - (1-\epsilon) \proj{0} & -\sqrt{1-(1-\epsilon)^2} \proj{0} \\
	\sqrt{1-(1-\epsilon)^2} \proj{0}    & \Id - (1-\epsilon) \proj{0}
	\end{array}
	\right)
	\left(
	\begin{array}{cc}
	\!\t{Had}^\dag\! & 0 \\
	0 & \!\t{Had}^\dag\!
	\end{array}
	\right)	
}
is a normalised matrix-block-encoding of $A = \Id - (1-\epsilon) K'$. Note that the matrix in the centre can be interpreted as the $\proj{0}$-controlled version of the Pauli-$X$ rotation $e^{i\theta X}$ (with $\cos \theta = - (1-\epsilon)$) and is thus efficiently implementable. The operations in Eq.~\eqref{eq:transformations} correspond to a unitary quantum circuit that can be written as $D \oplus \U$, for some unitary $\U$, and finally we obtain that $\U_{A'} := (D\oplus \U) \, \U_A \, (D\oplus \U^\dag)$ is a matrix-block-encoding of $A'$.

We now consider the linear system $A'\x = \textbf{u}$ and a quantum algorithm that prepares the corresponding solution state
\al{
	\ket{A'^{-1}\,\textbf{u}} 
	=
	\ket{DA^{-1}D\,\textbf{u}}
	=
	D \ket{A^{-1}\,\textbf{b}} .
}
The state $D\ket{A^{-1}\,\textbf{b}}$ can be transformed into $\ket{A^{-1}\,\textbf{b}}$ using the steps given in Eq.~\eqref{eq:transformations}, which only requires two extra accesses to $\P_y$. This state allows to solve \textsc{PromiseMajority}$_M$ with constant probability and thus the same considerations made in the preceding proof yield the result $Q[\U_A] \in \Omega\big(\min(\kappa, N) \big)$.

\end{proof}

\begin{proof}\emph{Case 3.}

We start proving again a lower bound on the query complexity $Q[\U_\b]$, as in Case 1., but for a PD-QLS where the coefficient matrix $A$ is sparse; then, we use the method used in the proof of Case 2.\ to convert it into a lower bound on $Q[\P_A]$.

Given $y \in \{0,1\}^N$ satisfying the \textsc{PromiseMajority}$_M$ condition, we introduce the known term vector $\b \in \mathds{R}^{N+1}$
\al{
\left\lbrace
\begin{array}{l}
	b_i = (-1)^{y_i} \qquad \qquad \qquad \t{for} \ i\in [N] \\[2mm]
	b_{N+1} = \sqrt{N}\, c_0
\end{array}
\right.
\label{eq:b_reference2}
}
where $c_0$ is a positive constant that we will fix later, and we have $\ket{\b} = \b \Big/ \sqrt{N (1+ c_0^2)}$. 

Next, we define $B' \in \mathds{R}^{(N+1) \times (N+1)}$ as 
\al{
	B' = B \oplus 0  
}
where $B \in \mathds{R}^{N \times N}$ is a symmetric sparse matrix, having $d$ non-zero entries in each row and column which are all equal to $\frac{1}{d}$, for some constant $d$. Since each row and column of $B$ sums to one, $B$ can be interpreted as the adjacency matrix of a Markov chain on a $d$-sparse graph. We require that the graph corresponding to $B$ is not bipartite and that the spectral gap of $B$ is large (i.e., the Markov chain is ergodic and rapidly mixing). These properties guarantee that $B^t$ quickly converges to $K = \frac{1}{N} \boldsymbol{1}_N\boldsymbol{1}_N^{\,T}$ for $t \rightarrow \infty$. Since $B$ is symmetric, its spectrum is real and, because of the Perron-Froebenius theorem~\cite{Pillai05}, the spectrum is contained in the interval $[-1,+1]$ and includes an eigenvalue $\lambda = 1$ with multiplicity one; the fact that $B^t$ converges to $\frac{1}{N}\boldsymbol{1}_N \boldsymbol{1}_N^{\,T}$ implies that $-1$ cannot be an eigenvalue of $B$. We then define the spectral gap $\delta(B)$ as the positive parameter $\delta(B) := \min_{\lambda\neq 1} \{1 - |\lambda|\}$, where the minimum is taken over the eigenvalues of $B$. 

There are families of so-called \emph{expander graphs} such that both the sparsity and the spectral gap are constant, see~\cite[Chapter 21]{Arora09}. We assume that $B$ belongs to one of these expander families and thus, in particular, there is a (known) positive constant $c_1$ such that
\al{
	\frac{1}{\delta(B)} \leq c_1
}
for all sizes $N \in \mathds{N}$. Moreover, $\boldsymbol{1}_N$ is the unique $+1$ eigenvector and hence, from the spectral decomposition of $B$, we can write
\al{
	B 
	= \frac{1}{N}\boldsymbol{1}_N \boldsymbol{1}_N^{\,T} 
	+ \sum_{\lambda\neq 1} \lambda \, \textbf{v}_\lambda \textbf{v}_\lambda^{\,\dag} 
	=  K + R \,.
}
Here $K = \frac{1}{N} \boldsymbol{1}_{N} \boldsymbol{1}_{N}^{\,T}$, $\textbf{v}_\lambda$ are normalised eigenvectors of $B$, and thus $R \in \mathds{R}^{N\times N}$ is a matrix of rank $N-1$ with $\norm{R} = 1-\delta(B)$ and $KR = RK = 0$. Then we define:
\al{
	A := \Id - (1-\epsilon) B'
	\label{eq:A_I-K'}
}
for some $\epsilon>0$ that we will define shortly. The matrix $A$ can be used as a coefficient matrix in a PD-QLS solver since it is positive definite and with norm one. Moreover, $A$ is $(d+1)$-sparse (where $d$ is the sparsity of $B$, which is constant) and the condition number of $A$ is exactly $\kappa(A) = 1/\epsilon$.

Next we have
\al{
	A^{-1} 
	& = 
	\Id 
	+ \sum_{t=1}^{\infty} \big((1-\epsilon) B' \big)^t
	\equiv \Id + \mathcal{B}
}
and then defining $\mathcal{K} := \sum_{t=1}^\infty \big((1-\epsilon) K' \big)^t = \frac{1-\epsilon}{\epsilon} K'$ we get
\al{
	\norm{\mathcal{B} - \mathcal{K}} 
	\label{eq:B-K}
	& \leq 
	\sum_{t=1}^\infty 
	\norm{\big[(1-\epsilon) B \big]^t - \big[(1-\epsilon) K\big]^t} \\
	& = 
	\sum_{t=1}^\infty 
	(1-\epsilon)^t \norm{ (K+R)^t - K^t} \\
	& = 
	\sum_{t=1}^\infty 
	(1-\epsilon)^t \norm{ (K+R^t) - K} \\
	& \leq 
	\sum_{t=1}^\infty 
	(1-\delta(B))^t	\\
	& \leq
	\frac{1}{\delta(B)}
	\leq c_1
	\label{eq:B-K_last}
}
where we have used $\norm{R} = 1-\delta(B)$, $K^2 = K$ and $KR = RK = 0$. 

Applying $A^{-1} = \Id+\mathcal{B}$ to $\b$ we thus obtain:
\al{
	A^{-1} \b 
	& =
	\big[ (\Id + \mathcal{B} - \mathcal{K}) + \mathcal{K} \big] \b \\
	& = 
	(\Id + \mathcal{B} - \mathcal{K}) \b' + 
	\begin{cases}
	b_{N+1}\textbf{e}_{N+1} + \frac{1-\epsilon}{\epsilon} \frac{M}{N}\, \boldsymbol{1}_{N}' 
	& \ \t{if~} y \t{~has~a~majority~of~0} \\
	b_{N+1}\textbf{e}_{N+1} - \frac{1-\epsilon}{\epsilon} \frac{M}{N}\, \boldsymbol{1}_{N}' 
	& \ \t{if~} y \t{~has~a~majority~of~1}
	\end{cases} 
	\label{eq:x_with_error} \\
	& = 
	(\Id + \mathcal{B} - \mathcal{K}) \b' + 
	\begin{cases}
	b_{N+1}\textbf{e}_{N+1} + c_0 \boldsymbol{1}_{N}' 
	& \qquad \; \t{if~} y \t{~has~a~majority~of~0} \\
	b_{N+1}\textbf{e}_{N+1} - c_0 \boldsymbol{1}_{N}' 
	& \qquad \; \t{if~} y \t{~has~a~majority~of~1}
	\end{cases}
}
where $\b'$ is a vector equal to the first $N$ entries of $\b$ (and $b'_{N+1} =0$), $\boldsymbol{1}_N' := (1,\ldots,1,0)^T$ while $\textbf{e}_{N+1}$ is the vector having a one in position $N+1$; moreover, we choose $\epsilon$ such that $\frac{1-\epsilon}{\epsilon} \frac{M}{N} = c_0$, where $c_0$ was introduced in the definition of $\b$ in Eq.~\eqref{eq:b_reference2}. Then, fixing the constant $c_0$ as $c_0 = 100 c_1$ and using the triangle inequality, we obtain the following upper bound
\al{
	\mathcal{\sqrt{N}} = \norm{A^{-1}\b}
	& \leq 
	\Big(
	\overbrace{c_0^2 N }^{|b_{N+1}|^2}
	+ \!\!\!
	\overbrace{c_0^2 N }^{\norm{c_0 \boldsymbol{1}_{N}'}^2}
	\hspace{-1.5mm}\Big)^{1/2}
	+ 
	\overbrace{\sqrt{N}\, [1+c_1 ]}^{\geq\,\norm{(\Id + \mathcal{B} - \mathcal{K}) \b' }}  
	\\
	& \leq 
	\left(\sqrt{c_0^2 + c_0^2} +  2\,c_1 \right) 
	\sqrt{N} 
	\ = \
	\left(100\sqrt{2} + 2\right) c_1 \, \sqrt{N} \\ 
	& \leq
	144\,c_1 \, \sqrt{N} ,
}
where we have assumed $c_1 \geq 1$. We also have the lower bound
\al{
	\sqrt{\mathcal{N}} = \norm{A^{-1}\b} 
	& \geq 
	\Big(
	\overbrace{c_0^2 N }^{|b_{N+1}|^2}
	+ \!\!\!
	\overbrace{c_0^2 N }^{\norm{c_0 \boldsymbol{1}_{N}'}^2}
	\hspace{-1.5mm}\Big)^{1/2}
	- 
	\overbrace{\sqrt{N}\, [1 + c_1 ]}^{\geq\,\norm{(\Id + \mathcal{B} - \mathcal{K}) \b' }}  
	\\
	& \geq		
	\left(100\sqrt{2} - 2\right) c_1 \, \sqrt{N} \\ 
	& \geq
	139\,c_1 \, \sqrt{N} .
}
Then we have, using $f = \textsc{PromiseMajority}_M(y)$,
\al{
	\ket{A^{-1}\b} 
	= 
	\frac{A^{-1}\b}{\norm{A^{-1}\b}} 
	& = 
	\frac{c_0\,\sqrt{N}}{\sqrt{\N}}
	\ket{N+1} 
	\, + \, 
	(-1)^f 
	\frac{c_0\,\sqrt{N}}{\sqrt{\N}}
	\sum_{i = 1}^N\frac{1}{\sqrt{N}} \ket{i} 
	+ \ket{\psi} 
	\label{eq:line1}\\
	& = 
	\frac{100}{144} \ket{N+1} 
	\, + \, (-1)^f 
	\frac{100}{144}
	\sum_{i = 1}^N\frac{1}{\sqrt{N}} \ket{i} 
	+ \ket{\psi'} ,
	\label{eq:line2}
}
where $\ket{\psi} := (\mathcal{B} - \mathcal{K}+\Id) \b'/\mathcal{\sqrt{N}} $ is a sub-normalised perturbation vector with $\big\vert\big\vert \ket{\psi} \big\vert\big\vert \leq \frac{2}{139}$, while subtracting line~\eqref{eq:line2} from line~\eqref{eq:line1} one obtains
\al{
	\Big\vert\Big\vert\,
	\ket{\psi'} - \ket{\psi}\,\Big\vert\Big\vert \leq 
	\sqrt{2}
	\left( 
	\frac{100}{139} - \frac{100}{144} 
	\right)
	\leq 0.04
}
and then we have:
\al{
	\Big\vert\Big\vert\,\ket{\psi'}\,\Big\vert\Big\vert 
	\leq 0.04 + \frac{2}{139} \leq 0.06 \,.
}

The cases $f=0$ and $f=1$ in Eq.~\eqref{eq:line2} can be distinguished with constant advantage using the swap test with the state $\ket{``+"}$ defined in Eq.~\eqref{eq:plus_state}, since we have:
\al{
	& ~~\;\braket{``+"}{A^{-1}\b} 
	\, = \,
	\frac{1}{\sqrt{2}} \frac{100}{144} 
	\, + \,
	(-1)^f \frac{1}{\sqrt{2}} \frac{100}{144}
	\, + \, 
	\braket{``+"}{\,\psi'\,} \\
	\Longleftrightarrow \quad &
	\begin{cases}
	\big|\, \braket{``+"}{A^{-1}\b} \,\big| \,\geq\, 0.92 
	& \t{if~} y \t{~has~a~majority~of~0} \\
	\big|\, \braket{``+"}{A^{-1}\b}\, \big| \,\leq\, 0.06 
	& \t{if~} y \t{~has~a~majority~of~1}\,.
	\end{cases}
}

Next, we proceed as in the proof of Case 2.\ and define an equivalent PD-QLS where the vector $y$ is encoded in the entries of the coefficient matrix. We thus define the vector $\textbf{u} \in \mathds{R}^{N+1}$ and the diagonal matrix $D \in \mathds{R}^{(N+1) \times (N+1)} $
\al{
\left\lbrace
\begin{array}{ll}
	u_i = 1 & \t{for} \ i\in [N] \\[2mm]
	u_{N+1} = \sqrt{N}\, c_0 &
\end{array}
\right.
\hspace{1cm}	
\left\lbrace
\begin{array}{ll}
	D_{i,i} = (-1)^{y_i} & \t{for} \ i\in [N] \\[2mm]
	D_{N+1,N+1} = 1 &
\end{array}
\right.
}
and thus the identity $\b = D \textbf{u}$ holds. We then introduce $A'$, given by
\al{
	A' := D A D
}
where $A$ is as in Eq.~\eqref{eq:A_I-K'}. The matrix $D$ is unitary and self-inverse, hence $\kappa(A') = \kappa(A)$, and $A'^{-1} = D^{-1} A^{-1} D^{-1} = D A^{-1} D$. 

Notice that the position of the non-zero entries of $A'$ are the same as in $A$ and thus independent from $y$, while the sign of a non-zero entry $A'_{i,j} = (-1)^{y_i + y_j}$ can be obtained querying $\P_y$ once with input $\ket{i}$ and once with input $\ket{j}$. These queries can be performed in quantum superposition and thus two accesses to $\P_y$ are sufficient to implement a quantum sparse-matrix-access $\P_A'$. Therefore it is possible to prepare, with the same $\P_y$-complexity as discussed previously, the state
\al{
	\ket{A'^{-1}\,\textbf{u}} 
	=
	\ket{DA^{-1}D\,\textbf{u}}
	=
	D \ket{A^{-1}\,\textbf{b}}\,.
}
Finally, we can obtain $\ket{A^{-1}\,\textbf{b}}$ using two extra accesses to $\P_y$ by applying the transformations given in Eq.~\eqref{eq:transformations}. This proves that a quantum algorithm that solves the PD-QLS having access to $\P_A'$ necessarily has a query complexity $Q[\P_{A'}] \in \Omega\big(\min(\kappa, N) \big)$.

\end{proof}

\section{Scaling of the normalisation factor of the matrix-block-encoding}
\label{app:polynomial}
\setcounter{equation}{0}

In this Appendix, we consider the polynomial
\al{
\label{eq:app_polynomial}
	P_{2\ell-1,\kappa}(x) 
	:= 
	\frac{1}{1-x} \left[1 - \hat{\T}_{\ell,\kappa}(x) \right]^2
}
as was defined in Eq.~\eqref{eq:polynomial} and where we have
\al{
	\hat{\T}_{\ell,\kappa}(x)
	:=
	\frac{\T_\ell \left(\frac{x + \frac{1}{2\kappa}}{1-\frac{1}{2\kappa}} \right)}
	     {\T_\ell \left(\frac{1 + \frac{1}{2\kappa}}{1-\frac{1}{2\kappa}} \right)} \;.
}
We prove that the normalisation factor $K:= 2\,\max_{x\in [-1,+1]} P_{2\ell-1,\kappa}(x)$ satisfies $K \in \Theta(\kappa)$, provided that $\ell \geq c \sqrt{\kappa}$ for some constant $c$ that we will determine later.

Notice that by construction $P_{2\ell-1,\kappa}(1) = 0$ and the polynomial is positive for $x\in[-1,+1)$. To study the properties of the local maxima of $P_{2\ell-1,\kappa}(x)$ in the interval $[-1,+1]$ we compute the derivative of $P(x)$ using the property $\frac{\partial}{\partial x} \T_\ell(x) = \ell\,\U_{\ell-1}(x)$, where $\U_\ell(x) \in \mathds{R}_\ell[x]$ is a Chebyshev polynomial of the second kind. We have:
\al{
	\frac{\partial P_{2\ell-1,\kappa}(x)}{\partial x}
	=
	\frac{\left[1 - \hat{\T}_{\ell,\kappa}(x) \right]^2}{(1-x)^2} 
	\; - \;
	2 \ell \, 
	\frac{1 - \hat{\T}_{\ell,\kappa}(x)}{1-x}
	\frac
	{
		\U_{\ell-1} 
		\left(
		\frac{x + \frac{1}{2\kappa}}
             {1 - \frac{1}{2\kappa}}
        \right)
	}
	{
		\left( 1-\frac{1}{2\kappa} \right)
		\T_{\ell}
		\left(
		\frac{1 + \frac{1}{2\kappa}}
			 {1 - \frac{1}{2\kappa}}
        \right)
	}
	\;.
}
We set the derivative equal to $0$ and simplify the expression assuming $1-x \neq 0$ and $1-\hat{\T}_{\ell,\kappa}(x) \neq 0$:
\al{
	\left( 1-\frac{1}{2\kappa} \right)
	\left[	
	\T_{\ell}\!
	\left(
	\frac{1 + \frac{1}{2\kappa}}
		 {1 - \frac{1}{2\kappa}}
    \right)
	- 
	\T_{\ell}\!
	\left(
	\frac{x + \frac{1}{2\kappa}}
		 {1 - \frac{1}{2\kappa}}
    \right)
	\right] 
	\; - \;
	2 \ell \, 
	(1-x) \,
	\U_{\ell-1} \!
	\left(
	\frac{x + \frac{1}{2\kappa}}
         {1 - \frac{1}{2\kappa}}
    \right)
	= 0
}
Then we use the change of variables $y(x)$ and $\delta(\kappa)$ given in Eqs.~\eqref{eq:change} and their inverses $\kappa(\delta) = \frac{1}{\delta} + \frac{1}{2}$, $x(y) = \frac{y - \delta/2}{1+ \delta/2}$ to rewrite the previous equation as
\al{
	\left( \frac{1}{1+\delta/2} \right)
	\big[	
	\T_{\ell}( 1 + \delta )
	- 
	\T_{\ell}( y )
	\big] 
	\; - \;
	2 \ell \, 
	\left( \frac{1 + \delta - y}{1+\delta/2} \right)
	\,
	\U_{\ell-1} ( y )
	= 0
}
which is equivalent to:
\al{\boxed{
\label{eq:main_eq}
	\T_{\ell}( 1 + \delta )
	- 
	\T_{\ell}( y )
	\; = \;
	2 \ell \, 
	(1 + \delta - y)
	\,
	\U_{\ell-1} ( y ) 
}}
for $x \in [-1,+1]$ or, equivalently, $y \in [-1,1 + \delta]$.

The polynomial $P_{2\ell-1,\kappa}(x)$ is by construction non-negative on the domain $x\in[-1,+1]$ and its derivative in $x=1-\frac{1}{\kappa}$ (corresponding to $y=1$) is positive, provided that $\ell \in \Omega(1/\sqrt{\delta})$. In fact, the derivative of $P_{2\ell-1,\kappa}(x)$ is positive if and only if the left hand side of Eq.~\eqref{eq:main_eq} is larger than the right hand side; using $\T_{\ell}(1+\delta) \geq \frac{1}{2}e^{\ell\sqrt{\delta}}$ for $0 \leq \delta \leq 3 - 2\sqrt{2}$, $\T_{\ell}(1) = 1$, $\U_{\ell-1} (1) = \ell$ we obtain from~\eqref{eq:main_eq} the inequality $\frac{1}{2}e^{\ell \sqrt{\delta}} - 1  \overset{!}{>} 2\, \ell^2 \delta$, which is satisfied for $\ell \geq 4.36 /\sqrt{\delta}$. Since we have, moreover, $P_{2\ell-1,\kappa}(1) = 0$, the function $P_{2\ell-1,\kappa}(x)$ does not have any local maximum in $\left[\, -1 , 1- \frac{1}{\kappa},\right]$ and must have one or more local maxima $x_* \in \left(\, 1\! -\! \frac{1}{\kappa}, +1\,\right]$; equivalently, Eq.~\eqref{eq:main_eq} must have at least one solution $y_* \in (1,1+\delta]$. 

Any local maximum $y_* = y(x^*)$ satisfies:
\al{
\label{eq:main_eq2}
	\T_{\ell}( y_* )
	=
	\T_{\ell}( 1 + \delta )
	- 
	2 \ell \, 
	(1 + \delta - y_*)
	\,
	\U_{\ell-1} ( y_* ) 
}
and substituting $\T_{\ell}( y_* )$ in the definition~\eqref{eq:app_polynomial} gives
\al{
	P_{2\ell-1,\kappa}(x_*)
	& =
	\frac{1+\delta/2}{1+\delta-y_*}
	\left[
	1 - \frac{\T_\ell(y_*)}{\T_\ell(1+\delta)} 
	\right]^2 \\
	\label{eq:U_l-1}
	& =
	4\ell^2 (1+\delta/2) (1+\delta - y_*) \,
	\frac{\U_{\ell-1}(y_*)^2}{\T_\ell(1+\delta)^2} \;.
}
From Eq.~\eqref{eq:main_eq} we directly have
\al{
	\U_{\ell-1}(y_*)
	\leq
	\frac{\T_\ell(1+\delta)}{2\ell (1+\delta -y_*)}
}
and inserting this inequality in Eq.~\eqref{eq:U_l-1} we have that any local maximum $x_*$ satisfies
\al{
\label{eq:P_ineq}
	P_{2\ell-1,\kappa}(x_*)
	& \leq
	\frac{1 + \delta / 2}{1+\delta - y_*} \\
	& \leq  
	\frac{3/2}{1+\delta - y_*} \;.
}

In summary, it will be sufficient to show $1+\delta - y_* \in \Omega(\delta)$ to prove that the normalisation constant satisfies $K=2\,\max_{\{x_*\}} |P_{2\ell-1,\kappa}(x_*)| \in \O(\kappa)$, where the maximisation is over the set of (potentially multiple) local maxima $x_* \in \left[\, 1\! -\! \frac{1}{\kappa}, +1\,\right]$.

To this end, we rewrite Eq.~\eqref{eq:main_eq} as:
\al{
	\frac{
	\T_\ell(1+\delta) - \T_\ell(y_*)}
	{(1+\delta) - y_*}
	\; = \;
	2\, \frac{\partial \T_\ell}{\partial y} (y_*) \;.
}
Since both the first and the second derivative of $\T_\ell(y)$ are positive for $y\geq 1$ (i.e., the derivative of $\T_\ell(y)$ is monotonically increasing), this equation can be satisfied only if\footnote{In this appendix we employ the notation $\overset{!}{\leq}$ to mark inequalities that we still need to prove and with $\leq$ an inequality that has already been proven.}
\al{
	2 \frac{\partial \T_\ell}{\partial y} (y_*)
	\overset{!}{\leq} 
	\frac{\partial \T_\ell}{\partial y} (1+\delta) 
}
or equivalently, defining $1+\delta_* := y_*$
\al{
\label{eq:main_ineq}
	2\,\U_{\ell-1} (1+\delta_*)
	\, \overset{!}{\leq} \, 
	\U_{\ell-1} (1+\delta) \;.
}
A Chebyshev polynomial of the second kind can be written as
\al{
	\U_{\ell-1}(y) = 
	\frac{(y + \sqrt{y^2 -1})^\ell - (y + \sqrt{y^2 -1})^{-\ell}}
	{2\sqrt{y^2 - 1}}
}
hence we have
\al{
	\U_{\ell-1} (1+\delta_*) 
	& \leq 
	\frac{(1+\delta_* + \sqrt{2\delta_* + \delta_*^2})^\ell}
	{2\sqrt{2\delta_* + \delta_*^2}} \\
	& \leq 
	\frac{(1 + 1.1\sqrt{2\,\delta_*})^\ell}
	{2\sqrt{2\delta_*}}	
}
for sufficiently small $\delta_*$ (the constant $1.1$ is somewhat arbitrary), while we have
\al{
	\U_{\ell-1} (1+\delta) 
	& = 
	\frac
	{(1+\delta + \sqrt{2\delta + \delta^2})^\ell
	 \big[ 1- (1+\delta + \sqrt{2\delta + \delta^2})^{-2\ell} \big] }
	{2\sqrt{2\delta + \delta^2}} \\
	& \geq 
	\frac{(1+\sqrt{2\delta})^\ell \times 0.8}
	{2.4\sqrt{2\delta}}	\\
	& = \frac{2}{3}
	\frac{(1+\sqrt{2\delta})^\ell }
	{2 \sqrt{2\delta}} \\
	& \geq \frac{2}{3}
	\frac{(1+\sqrt{2\delta})^\ell }
	{2 \sqrt{3\delta_*}} 
}
which holds for $\ell \geq 1/\sqrt{2\delta}$ and in the last step we have assumed $\delta_* \geq\frac{2}{3}\delta$ (notice that in the opposite case $\delta_* < \frac{2}{3}\delta$ we would already have $1+\delta-y_* = \delta-\delta_* > \frac{1}{3}\delta$). Thus, the inequality in~\eqref{eq:main_ineq} is implied by
\al{
	2 \, \frac{(1 + 1.1\sqrt{2\,\delta_*})^\ell}
	{2\sqrt{2\delta_*}}	
	\; & \overset{!}{\leq} \; 
	\frac{2}{3}
	\frac{(1+\sqrt{2\delta})^\ell }
	{2 \sqrt{3\delta_*}} \\
\Longleftrightarrow \hspace{5mm}
	(1 + 1.1\sqrt{2\,\delta_*})^\ell
	\; & \overset{!}{\leq} \; 
	\frac{\sqrt{2}}{3\sqrt{3}} \,
	(1+\sqrt{2\delta})^\ell  \\
\Longleftrightarrow \hspace{7mm}
	1 + 1.1\sqrt{2\,\delta_*}
	\ \ \; & \overset{!}{\leq} \; 
	e^{-\frac{1}{\ell} \log (3\sqrt{3/2})} \,
	(1+\sqrt{2\delta}) \;.
}
From the inequality $e^{-x} \geq 1-x$ for $x\geq 0$ we see that the inequality above is implied by
\al{
	1.1\sqrt{2\,\delta_*}
	\; & \overset{!}{\leq} \; 
	\big(1 - 1.31 / \ell \big) \,
	(1+\sqrt{2\delta}) - 1 \\
	\; & = \; 
	\sqrt{2\delta} - \frac{1}{\ell} 1.31\, (1 +\sqrt{2\delta})  
}
with $1.31 \geq \log (3\sqrt{3/2})$. We now make the assumption that $\frac{1}{\ell} 1.31 \,(1 +\sqrt{2\delta}) \leq \frac{1}{10}\sqrt{2\delta}$, which is implied by $\ell \geq 13.1 + 9.27/\sqrt{\delta}$ and is compatible with the requirement $\ell \in \Omega(\delta^{-1/2})$. Thus it is sufficient to impose:
\al{
	1.1\sqrt{2\,\delta_*}
	\; & \overset{!}{\leq} \;
	\frac{9}{10} \sqrt{2\delta} \\
	\Longleftrightarrow \qquad \qquad
	\delta_*
	\; & \overset{!}{\leq} \;
	\left(\frac{9}{11}\right)^2 \delta \;.
	\label{eq:last}
}
We also remark that $(9/11)^2 \geq 2/3$, so that this last inequality is compatible with the assumption $\delta_* \geq\frac{2}{3}\delta$ we made earlier.

Plugging the bound in Eq.~\eqref{eq:last} into~\eqref{eq:P_ineq}, together with the definition of $\delta(\kappa)$, results in the explicit bound $K \leq 6.05\, \kappa$, i.e.\ $K \in \O(\kappa)$, provided that $\ell \geq 13.1 + 9.27 \sqrt{\kappa - 1/2}$, i.e.\ $\ell \in \Omega(\sqrt{\kappa})$.

\section{VTAA optimization improving the runtime values of the first algorithm}
\label{app:VTAA}
\setcounter{equation}{0}

In this Appendix we use VTAA to speed-up the asymptotic runtime of the algorithm based on polynomial approximations of $1/(1-x)$. Preliminarily, we introduce a Lemma showing that this VTAA algorithm, as summarised in \prop{VTAA}, does provide an asymptotic query complexity improvement (compared to general QLS solvers) for most values of $\Gamma_{A,\b}$, see the right plot in \fig{Runtimes}.

\begin{lemma}
\label{lem:Gamma}
Defining 
$
	\Gamma_{A,\b} :=
	\sqrt{\kappa} \,
	\frac{\norm{A^{-1/2}\ket{\b}}}{\norm{A^{-1}\ket{\b}}}
$ 
we have $\Gamma_{A,\b} \in [1,\sqrt{\kappa}\,]$, under the usual assumption that the spectrum of $A$ is contained in $[1/\kappa, 1]$.
\end{lemma}

\begin{proof}
We expand $\ket{\b}$ in a basis of eigenvectors of $A$, that is $\ket{\b} = \sum_\lambda \beta_\lambda \ket{\lambda}$ with $A\ket{\lambda} = \lambda \ket{\lambda}$, $\braket{\lambda}{\lambda'} = \delta_{\lambda,\lambda'}$, and $\sum_\lambda |\beta_\lambda|^2 =1$. Then, we can write
\al{
	\frac{\norm{A^{-1/2}\ket{\b}}^2}{\norm{A^{-1}\ket{\b}}^2}	
	=
	\frac{\sum_\lambda |\beta_\lambda|^2 f_\lambda}
	     {\sum_\lambda |\beta_\lambda|^2 g_\lambda}
	\ & \in \ 
	\big[\min_\lambda \{f_\lambda / g_\lambda\} \, , \,
	     \max_\lambda \{f_\lambda / g_\lambda\} \big] 
	\ \subseteq \
	\big[ \kappa^{-1} , 1 \big]
}
for $f_\lambda = \lambda^{-1}$, $g_\lambda = \lambda^{-2}$ and thus $f_\lambda / g_\lambda = \lambda \in [\kappa^{-1},1]$. We take the square root of both extrema and multiply by $\sqrt{\kappa}$ to conclude.

\end{proof}

\subsection{Variable-time amplitude amplification review}

In this Section we review the general VTAA method, mainly following Ref.~\cite[Section 3]{Chakraborty18}.

\begin{definition}[Variable-stopping-time quantum algorithm]
\label{def:VST}
A \emph{variable-stopping-time quantum algorithm} $\A = \A_m \cdot \ldots \cdot \A_1 \cdot \A_0$ is given by the application of $m + 1$ sub-algorithms $\A_j$ in sequence, acting on the Hilbert space $\hil = \hil_C \otimes \hil_F \otimes \hil_S$ where $S$ is a ``system register'' of arbitrary size, $F$ is a ``flag qubit'' that heralds success and $\hil_C = \bigotimes_{j=0}^m \hil_{C_j}$ is a ``clock register'' containing $m+1$ qubits $C_0, C_1, \ldots, C_m$. $\hil$ is initially prepared in the all-zero state $\ket{0}_\t{all} = \ket{0,0,0}_{C,F,S}$. Each $\A_j$ acts on $\hil_{C_j} \otimes \hil_F \otimes \hil_S$ and $\ket{1}_{C_j}$ indicates that the algorithm stops after the application of $\A_j$; i.e., each $\A_j$ is a controlled algorithm that acts if and only if the previous $j$ qubits in the clock register are in the state $\ket{0}^{\otimes j} \in \bigotimes_{i=0}^{j-1} \hil_{C_i}$. We assume that all branches of the computation end by step $m$. The successful branches of the algorithm are those where the flag is in the state $\ket{1}_F$ and thus we define
\al{
	\label{eq:VST_succ}
	p_\t{succ} := \bnorm{\,\Pi_\checkmark \A \ket{0}_\t{all}}^2
	\qquad \t{and} \qquad
	\ket{\psi_\t{succ}} :=
	\frac{\Pi_\checkmark \A \ket{0}_\t{all}}{\bnorm{\,\Pi_\checkmark \A \ket{0}_\t{all}}}
}
where $\Pi_\checkmark := \Id_C \otimes \proj{1}_F \otimes \Id_S$. 
\end{definition}

By construction, a variable-stopping-time quantum algorithm produces a quantum state of the form $\A \ket{0}_\t{all} = \sum_{j=0}^m \alpha_j \ket{1_j}_C \ket{\Psi_j}_{F,S}$, where 
\al{
	\ket{1_j}_C := \ket{0}^{\otimes m-j} \ket{1} \ket{0}^{\otimes j}
}
is a state having the qubit $C_j$ in $\ket{1}$ and the other clock qubits in $\ket{0}$, while $\ket{\Psi_j}_{F,S}$ are normalised quantum states in $\hil_F \otimes \hil_S$ with $\sum_{j=0}^m |\alpha_j|^2 = 1$.

\begin{definition}[Stopping times]
\label{def:VST_prop}
We introduce a sequence of \emph{stopping times} $t_\t{min} \equiv t_0 < t_1 < t_2 < \ldots < t_m \equiv t_\t{max}$, where each $t_j \in \mathds{N}$ is the complexity of the sub-algorithm $\A_{\leq j} := \A_j \cdot \ldots \cdot \A_0$. The probability $p_j$ of stopping at time $t_j$ (i.e., after the execution of algorithm $\A_{\leq j}$) and the  $\ell_2$-\emph{average stopping time} are defined as
\al{
	p_j := \bnorm{\,\Pi_{C_j} \A_{\leq j} \ket{0}_\t{all}}^2 
	\qquad \quad
	t_\t{avg} := \sqrt{\sum_{j=0}^m p_j \, t_j^2} 
}
with $\Pi_{C_j} := \proj{1_j}_C \otimes\Id_F \otimes\Id_S$. We define $\Pi_{\t{stop} \leq j} := \sum_{i=0}^j \Pi_{C_i}$ and
\al{
	\Pi_\t{bad}^j 
	& :=
	\Pi_{\t{stop} \leq j}
	\cdot
	\big(\Id_C \otimes \proj{0}_F \otimes \Id_S \big) \\
	\Pi_\t{mg}^j 
	& :=
	\Id - 
	\Pi_\t{bad}^j
}
which project onto ``bad'' and ``maybe good'' subspaces after the application of $\A_{\leq j}$. The ``maybe good'' subspace at step $j+1$ is contained in the ``maybe good'' subspace at step $j$ 
and, by construction, $\Pi_\t{mg}^m \A\ket{0}_\t{all} = \Pi_\checkmark \A\ket{0}_\t{all}$.
\end{definition}

In the definition above the $t_j$ can represent any complexity measure (e.g., query complexity, gate complexity, circuit depth). In this Appendix we only compute the query complexity $Q$ of the sub-algorithm $\A_{\leq j}$, but we remark that all the algorithms we consider here are gate-efficient, i.e., the gate complexity is in $\O\big(Q \,\t{poly}(\log Q,\log N)\big)$.

\begin{definition}[Variable-Time Amplitude Amplification]
\label{def:VTAA}
Given a variable-stopping-time algorithm $\A = \A_m \cdot \ldots \cdot \A_0$ as in \defin{VST} and given a sequence of $m+1$ non-negative integers $(k_0, k_1, \ldots, k_m)$, we recursively define a variable-time amplification $\A' = \A_m' \cdot \ldots \cdot \A_0'$ as follows. Setting $\A_{-1} := \Id$, each $\A_j'$ implements a standard $k_j$-step amplitude amplification that uses $\A_j \A_{j-1}'$ and its inverse a total of $2k_j + 1$ times, where the input state is $\sket{\psi_\t{in}^j} = \A_j\A_{j-1}' \ket{0}_\t{all}$ and the target state is $\Pi_\t{mg}^j \sket{\psi_\t{in}^j}$. That is, $\A_j'$ begins preparing the input state
\al{
	& ~\sket{\psi_\t{in}^j} 
	= \
	\A_j \A_{j-1}' \ket{0}_\t{all} 
	\, = \, 
	\sin(\theta_j) \, \sket{\psi_\t{mg}^j} +
	\cos(\theta_j) \,\sket{\psi_\t{bad}^j} 
	\\[1mm]
	\t{with} ~ &
	\begin{cases}
	\sket{\psi_\t{mg}^j} \,\propto \; \Pi_\t{mg}^j  \sket{\psi_\t{in}^j} \\
	\sket{\psi_\t{bad}^j}  \propto \; \Pi_\t{bad}^j \sket{\psi_\t{in}^j}
	\end{cases}
	\qquad
	\theta_j := 
	\arcsin\!\big(\norm{\Pi_\t{mg}^j \A_j \A_{j-1}' \ket{0}_\t{all}}\big) 
	\; \in \; [0, \pi/2]
}
and then uses $k_j$ accesses to the reflections $R_\t{out}^j := \Id - 2\, \Pi_\t{mg}^j$ and $R_\t{in}^j := 2\, \sket{\psi_\t{in}^j}\sbra{\psi_\t{in}^j} - \Id$, where $R_\t{in}^j$ is implemented using $\A_j \A_{j-1}'$ and $(\A_j \A_{j-1}')^\dag$ once, to obtain the output state
\al{
	\sket{\psi_\t{out}^j} 
	= \
	& \A_{j}' \ket{0}_\t{all} 
	\, = \, 
	\sin\!\big[(2k_j + 1)\theta_j\big] \, \sket{\psi_\t{mg}^j} +
	\cos\!\big[(2k_j + 1)\theta_j\big] \,\sket{\psi_\t{bad}^j} \,.
}
\end{definition}

Note that if we set $k_0 = \ldots = k_m = 0$ we have $\A' \equiv \A$. Also, by the recursive structure of VTAA the first sub-algorithm $\A_1$ is used a total of $\prod_{j=0}^m (2k_j +1)$ times in $\A'$, which grows exponentially in $m$ if $k_j \geq 1$. Nonetheless, VTAA can provide a speed-up when the amplification parameters $(k_0, k_1, \ldots, k_m)$ are chosen appropriately. More precisely, the following result can be derived from Ref.~\cite[Lemma~22]{Chakraborty18}.

\begin{proposition}[Result of VTAA]
\label{prop:VTAA_analysis}
Using the notation of the previous definitions, let $\A'$ be a variable-time amplification such that each $\A'_j$ uses $k_j$ steps of amplitude amplification, where
\al{
	\label{eq:VTAA_condition}
	& \frac{\pi}{8\,\theta_j}-\frac{1}{2}
	\, \leq \,
	k_j 
	\, \leq \, 
	\frac{\pi}{4\,\theta_j}-\frac{1}{2}
}
Then, with the definitions in Eq.~\eqref{eq:VST_succ}, $\A'$ outputs the state 
\al{
	\A' \ket{0}_\t{all}
	\, = \,
	\sqrt{p_\t{succ}'} \,\sket{\psi_\t{succ}}_{C,F,S} 
	\, + \,
	\sqrt{1-p_\t{succ}'} \,\sket{\psi^\perp}_{C,F,S} 	
}
where success is heralded by $\ket{1}_F$, the success probability satisfies $p_\t{succ}' \in \Theta(1)$, and the global query complexity is
\al{
	\label{eq:VTAA_result}
	Q' \in 
	\O \!
	\left(
	t_\t{max} \sqrt{m}
	+ 
	\frac{t_\t{avg}}{\sqrt{p_\mathrm{succ}}}
	\sqrt{m \, \log(t_\t{max} / t_\t{min})}
	\right) .
}
\end{proposition}

We can compare Eq.~\eqref{eq:VTAA_result} with standard amplitude amplification, which has a query complexity $Q \in \O(t_\t{max} / \sqrt{p_\t{succ}})$. In our algorithm, $m$ will scale logarithmically in $t_\t{max}$, so the total runtime is $\widetilde{\O}(t_\t{max} + t_\t{avg} / \sqrt{p_\t{succ}})$. Hence, VTAA can provide a speed-up when the average runtime is much shorter than the maximum runtime.

We remark that a sequence of good values $(k_0, \ldots, k_m)$ such that conditions in Eq.~\eqref{eq:VTAA_condition} are satisfied can be obtained efficiently by means of an iterative classical-quantum pre-processing algorithm. Specifically, suppose that some values $(k_0, \ldots, k_j)$ satisfying~\eqref{eq:VTAA_condition} have been already found; then, one can compile the corresponding VTAA algorithm $\A_j'$ and use phase estimation on the state $\sket{\psi_\t{in}^{j+1}} = \A_{j+1}\A_j'\ket{0}_\t{all}$ to obtain an estimate of $\sin(\theta_{j+1})$ up to constant multiplicative precision; this allows to find a value $k_{j+1}$ which satisfies~\eqref{eq:VTAA_condition} with probability $1 - p_\t{fail}$ with a cost $\O\big(\frac{1}{\theta_{j+1}} \log(1/p_\t{fail}) \big)$ \cite{Brassard02}, which is asymptotically equal to the query complexity of $\A_{j+1}'$, apart from a multiplicative $\log(1/p_\t{fail})$ overhead. Once $k_{j+1}$ has been precomputed one can directly compile $\A_{j+1}'$, without the need of performing phase estimation ``online''. The total failure probability is $p_\t{fail}^\t{tot} \leq m \,p_\t{fail}$ and the query cost of the hybrid classical-quantum pre-processing algorithm has only a $\O\big(\log (1/p_\t{fail}) \big)$ multiplicative overhead compared to the expression in Eq.~\eqref{eq:VTAA_result}.

\subsection{Efficient implementation of windowing polynomials}

In this Section we introduce the ``windowing functions'' that we use to replace the Gapped Phase Estimation (GPE) subroutine, introduced in Ref.~\cite{Childs15}. Using GPE gives an additive $\O(\kappa)$ runtime overhead, which is too costly for our purposes.

\begin{lemma}[Efficient windowing polynomials]
\label{lem:window}
Given $\epsilon, \delta \in (0, 1/2]$, there exists an even polynomial $W_{\epsilon, \delta}(x) = W_{\epsilon, \delta}(-x)$ of degree $\ell \in \O\Big(\frac{1}{\sqrt{\delta}} \, \t{polylog}(\delta^{-1}, \epsilon^{-1})\Big)$, where  
$
	\t{polylog} (\delta^{-1}, \epsilon^{-1})
	\equiv
	\log^{1/4}(\epsilon^{-1}) 
	\big[\log(\epsilon^{-1}) + \log(\delta^{-1}) \big]
$,
satisfying the inequalities
\al{
	W_{\epsilon, \delta}(x) \in
	\begin{cases}
	[1 - \epsilon, 1     ] & \t{if} ~ x \in [0, 1-2\delta] \\
	[- 1, + 1] & \t{if} ~ x \in (1-2\delta, 1-\delta) \\
	[- \epsilon, + \epsilon] & \t{if} ~ x \in [1-\delta, 1] \;.
	\end{cases}
	\label{eq:W_inequalities}
}
The windowing function $W_{\epsilon, \delta}(x)$ can be computed efficiently with classical algorithms.
\end{lemma}

A family of polynomials satisfying the inequalities in~\eqref{eq:W_inequalities} is already given in Ref.~\cite[Lemma 29]{Gilyen18}, but they have a degree $\ell \in \widetilde{\O}(\delta^{-1})$. We seek to achieve the same result with a quadratically smaller degree, $\ell \in \widetilde{\O}(\delta^{-1/2})$. Also, it is crucial that the polynomial approximation has even parity: using QSP one can implement matrix polynomials when the polynomial $P$ satisfies $|P(x)| \leq 1$ for $|x| \leq 1$ and $P$ has definite parity, while for a polynomial $P'$ without definite parity the more restrictive constraint $|P'(x)| \leq 1/2$ has to be satisfied (see \thm{Signal}).

\begin{proof}
The proof proceeds in two steps: first, we find an analytic function $f(x)$ that satisfies the inequalities in~\eqref{eq:W_inequalities} within error $\epsilon/2$, and then show that the Chebyschev expansion converges very quickly to it, so that choosing a degree $\ell \in \widetilde{\O}\big(\delta^{-1/2}\big)$ is sufficient to be within $\epsilon/2$-distance from $f(x)$ for all $x$. \vspace{2mm}

\emph{First part:}
We introduce the normal distribution cumulative function
\al{
	\Phi(x) := 
	\frac{1}{\sqrt{2\pi}}
	\int_{-\infty}^x e^{-\frac{t^2}{2}} \,dt
	\qquad
	x \in \mathds{R}
}
normalised so that $0 \leq \Phi(x) \leq 1$. We then define
\al{
	\mathcal{W}_{\sigma,\delta}(x) := 
	\Phi \! \left( \frac{ x + 1 - 1.5 \, \delta}{\sigma} \right)
	\Phi \! \left( \frac{-x + 1 - 1.5 \, \delta}{\sigma} \right)
}
where $\sigma > 0$ has to be chosen so that $\mathcal{W}_{\sigma,\delta}(1-2\delta) \geq 1-\epsilon/2$ and $\mathcal{W}_{\sigma,\delta}(1-\delta) \leq \epsilon/2$. Using $\Phi(-x) = 1- \Phi(x)$ and the monotonicity of $\Phi$, both inequalities are implied by 
$
	\Phi \! \left(- \frac{ 0.5 \, \delta}{\sigma} \right)
	\leq 
	\epsilon/4
$. 
Using the bound
$
	\Phi(-x) 
	\leq 
	\frac{1}{\sqrt{2\pi}} \int_{-\infty}^{-x} \frac{-t}{x} e^{-t^2/2} \, dt 
	=
	\frac{e^{-x^2/2}}{\sqrt{2\pi} x}
	\leq
	\frac{e^{-x^2/2}}{\sqrt{2\pi} }
$
for $x \geq 1$, it is sufficient to choose 
$
	\sigma 
	\in 
	\Theta \big(\delta / \sqrt{\log (\epsilon^{-1})}\big)
$
to obtain the desired result. \vspace{2mm}

\emph{Second part:}
We consider the Chebyschev series associated to $\mathcal{W}_{\sigma,\delta}(x)$. Note that the function $\mathcal{W}_{\sigma,\delta}(x)$ has even parity, hence also the associated Chebyschev series has even parity \cite{Mason02}. Truncating the Chebyschev series at degree $\ell$, we get a sequence of polynomials $[S_\ell \mathcal{W}_{\sigma,\delta}](x)$ converging uniformly to $\mathcal{W}_{\sigma,\delta}(x)$ for $\ell \rightarrow \infty$. More precisely, we have the following result~\cite[Theorem~5.16]{Mason02}. 

\begin{lemma} 
Suppose $f(x)$ can be extended to an analytic function on the ellipse $E_r \subseteq \mathds{C}$
\al{
	E_r 
	=
	\Big\{ \
	a \cos \theta + i \, b \sin \theta
	\ \Big| \ 
	a = \tfrac{1}{2}(r + r^{-1}), \
	b = \tfrac{1}{2}(r - r^{-1}), \
	\theta \in [0, 2 \pi)
	\ \Big\}
}
for some $r>1$; then, the $\ell$-th degree truncation of the Chebyschev series is a polynomial $[S_\ell f] (x)$ that satisfies for all $x \in [-1,+1]$
\al{
	& \Big|
	f(x) - [S_\ell f](x)
	\Big|
	\leq 
	\frac{M}{r^\ell(r-1)}
	\qquad
	with ~
	M =
	\sup \big\{ |f(z)| : z \in E_r \big\} .
	\label{eq:ChebSeries}
}
\end{lemma}

We apply this theorem to $f(x) = \mathcal{W}_{\sigma,\delta}(x)$, which is analytic in the whole complex plane. The main technical hurdle is to upper bound $M$. We choose $r = \frac{1 + \sqrt{\sigma}}{1 - \sqrt{\sigma}} \geq 1 + 2\sqrt{\sigma}$, so that we have $a = \frac{1+\sigma}{1-\sigma}$ and $b = \frac{2 \sqrt{\sigma}}{1-\sigma}$, where $a$ and $b$ are the semi-axis of the ellipse $E_r$ along the real and imaginary axis, respectively. We have, for $x,y \in \mathds{R}$:
\al{
	|\Phi(x+iy)| 
	& = 
	\tfrac{1}{\sqrt{2\pi}} 
	\Big| 
	\int_{-\infty}^x e^{-\frac{t^2}{2}} dt + 
	\int_x^{x+iy} e^{-\frac{t^2}{2}} dt 
	\Big| \\
	& \leq 
	1 +
	\tfrac{1}{\sqrt{2\pi}} 
	\Big| 
	\int_0^{iy} e^{-\frac{(x+i\tau)^2}{2}} d\tau 
	\Big| \\
	& \leq 
	1 +
	\tfrac{1}{\sqrt{2\pi}} 
	\, e^{-x^2 /2}	
	\int_0^{|y|} e^{\frac{\tau^2}{2}} d\tau 
	\\
	& \leq 
	1 +
	\tfrac{1}{\sqrt{2\pi}} 
	\, e^{-x^2/2} 
	\, e^{y^2/2} \, |y|
	\\
	& \leq 
	1 +
	\tfrac{1}{\sqrt{2\pi}} e^{(-x^2+y^2)/2} 
	\label{eq:Phi_bound}
}
where the last inequality holds for $|y| \leq 1$. We now upper bound
\al{
	M_\t{half}
	:=
	\sup \left\{
	\left|	
	\Phi \! \left( \tfrac{ z + 1 - 1.5 \, \delta)}{\sigma} \right) 
	\right| :
	z \in E_r
	\right\}
}
so that we will have $M = \sup_{E_r} |\mathcal{W}_{\sigma,\delta}(z)| \leq M_\t{half}
^2$. From \eqref{eq:Phi_bound} we only need to maximize the expression $1+ \frac{1}{\sqrt{2\pi}} e^{(-x^2 + y^2)/2}$ for
\al{
	x = \frac{a}{\sigma} \cos \theta + \frac{1-1.5\delta}{\sigma} \qquad
	y = \frac{b}{\sigma} \sin \theta \,.
}
Equivalently, we can maximize $-x^2 + y^2$. Substituting $\cos \theta = u$ and $\sin^2 \theta = 1-u^2$ we get
\al{
	-x^2 + y^2 
	& =
	\frac{1}{\sigma^2} 
	\big[
	- (a u + 1 - 1.5\delta)^2 + b^2 (1- u^2)
	\big]
	\\
	& \leq
	\frac{b^2}{2\sigma^2} 
	\left[
	1 - \frac{(1-1.5\delta)^2}{a^2 + b^2}
	\right] .
}
We now have $b \in \O(\sqrt{\sigma}) $, $\frac{1}{a^2 + b^2} = \frac{1-2\sigma+\sigma^2}{1+6\sigma+\sigma^2} = 1 - \O(\sigma)$ and thus
\al{
	\t{sup} \,
	\{ -x^2 + y^2 \} \
	& \in 
	\O\Big(
	\frac{\sigma}{\sigma^2}
	\Big)
	\Big(
	1 - [1-\O(\delta)][1-\O(\sigma)]
	\Big)
	\\
	& =
	\O \big( \sigma^{-1}     \big)
	\O \big( \delta + \sigma \big)
	=  
	\O \big(\sqrt{\log \epsilon^{-1}} \,\big)
}
where we have used $\delta \in \Theta \big( \sigma \sqrt{\log \epsilon^{-1}}\big)$. We therefore have $M_\t{half} \in e^{\O(\sqrt{\log \epsilon^{-1}}\,)} \subset \O\big(\epsilon^{-1}\big)$ and thus also $M \leq M_\t{half}^2\in \O( \epsilon^{-2}\,)$. We thus set $\ell \in \Theta \big( \frac{1}{\sqrt{\sigma}} \log (\epsilon^{-3} \sigma^{-1/2})\big)$ and compute
\al{
	\Big|
	\mathcal{W}_{\sigma,\delta}(x) - [S_\ell \mathcal{W}_{\sigma,\delta}](x)
	\Big|
	& \leq
	\frac{M}{r^{\ell+1}(r-1)}
	\\
	& \leq
	\frac
	{C \, \epsilon^{-2}}
	{(1+2\sqrt{\sigma})^{\frac{D}{\sqrt{\sigma}} 
		\log (\epsilon^{-3}\sigma^{-1/2})}
	 \,2\sqrt{\sigma}}
	\\
	& \leq 
	\frac{C \, \epsilon^{-2}}
	     {e^{\log (\epsilon^{-3} \sigma^{-1/2})} \, 2 \sqrt{\sigma}}
	= 
	(\epsilon^3 \, \sigma^{1/2}) \; 
	\frac{C \, \epsilon^{-2}}
    		 {2 \sqrt{\sigma} }
    \\
    & \leq 
    \epsilon / 5
}
where $C$ and $D$ are positive constants. We recall that $\sigma \in \Theta\big(\delta / \sqrt{\log (\epsilon^{-1})}\big)$ and thus 
\al{
	\ell 
	& \in \Theta 
	\Big(\frac{1}{\sqrt{\sigma}}
	\big[3\log(\epsilon^{-1}) + \frac{1}{2}\log(\sigma^{-1})\big] \Big) \\
	& = \Theta 
	\Big(\frac{1}{\sqrt{\delta}} 
	\log^{1/4}(\epsilon^{-1}) 
	\big[\log(\epsilon^{-1}) + 
	\log(\delta^{-1}) \big]
	\Big)
}
is sufficient to guarantee that $[S_\ell \mathcal{W}_{\sigma,\delta}](x)$ is within $\epsilon/5$ distance from $\mathcal{W}_{\sigma,\delta}(x)$. 

Finally, we renormalise $[S_\ell \mathcal{W}_{\sigma,\delta}](x)$ dividing it by the maximum value attained for $x \in [-1,1]$; by construction, this value is smaller than $1 + \epsilon/5$, thus the maximum distance from $\mathcal{W}_{\sigma,\delta}$ after normalisation is bounded by $1-\frac{1-\epsilon/5}{1+\epsilon/5} < \epsilon/2$. 

\end{proof}

\subsection{Polynomial approximations on increasingly larger domains}

Now we explain how the windowing functions can be used to select domain where a polynomial of low degree can be a good approximation of the function $1/(1-x)$, if it is applied to eigenvalues contained in such domain.

As usual, the spectrum of $A$ is contained in $\D_A = \left[\, \frac{1}{\kappa}, \, 1\right]$ and, correspondingly, the spectrum of $B= \Id - \eta\,A$ is contained in $\D_B = \left[ 1-\eta, 1- \frac{\eta}{\kappa} \, \right] \subseteq \left[ 0, 1- \frac{\eta}{\kappa} \, \right]$. We set $m := \lceil \log_2 \kappa \rceil + 1$, $\delta_j := \eta\,2^{-j}$ for $j = \{1,\ldots,m\}$ and fix ${\tilde \epsilon} > 0$, a parameter related to the target precision $\varepsilon$. According to Eq.~\eqref{eq:P_bound} and Eq.~\eqref{eq:W_inequalities} we can find polynomials $P_{{\tilde \epsilon},\delta_j}(x)$ and $W_{{\tilde \epsilon},\delta_j}(x)$ such that
\al{
	& \Big|
	P_{{\tilde \epsilon}, \delta_j} (x) -
	\frac{1}{1-x}
	\Big|
	\leq {\tilde \epsilon}
	\qquad 
	\forall x \in  [-1, 1-\delta_j] \\
	& 
	W_{{\tilde \epsilon}, \delta_j} (x) 
	\qquad \quad
	\t{satisfies~Eq.}~\eqref{eq:W_inequalities}~\t{for~each~} \delta_j
}
with degrees in $\O\big(\delta^{-1/2} \log({\tilde \epsilon}^{-1} \delta^{-1})\big)$ and $\O\big(\delta^{-1/2} \log^{1/4}({\tilde \epsilon}^{-1}) \log({\tilde \epsilon}^{-1} \delta^{-1}) \big)$, respectively. We also define a normalisation factor 
\al{
	K := 2 \max_{j} \max_{x\in[-1,1]} |P_{{\tilde \epsilon}, \delta_j} (x)|
}
which can be set to coincide with the factor $K$ defined in Eq.~\eqref{eq:polynomial} and introduce the shorthands 
\al{
	P_j(x) := P_{{\tilde \epsilon}, \delta_j}(x) / K
	\quad \t{and} \quad
	W_j(x) := W_{{\tilde \epsilon}, \delta_j}(x) .
	\label{eq:P_and_W}
}

The windowing function $W_j(x)$ is used to select an interval $[-1+2\delta_j, 1-2\delta_j]$ where $P_j(x)$ is a good approximation of the inverse. For any eigenvalue $\lambda$ of $B$ we have $P_j(\lambda) \approx \frac{1}{K} \frac{1}{1-\lambda}$ when $\lambda \leq 1 - \delta_j$ and $W_j(\lambda) \approx 0$ when $\lambda \geq 1 -\delta_j$. More precisely, we have
\al{
	W_j(B) P_j(B)
	& =
	W_j(B) \frac{1}{K} \frac{1}{\Id - B} + \Delta_j^{\tilde \epsilon} \\
	& =
	W_j(B) \, \frac{A^{-1}}{\eta\, K} + \Delta_j^{\tilde \epsilon}
	\label{eq:Delta_residuals}
}
where $\Delta_j^{\tilde \epsilon}$ are arbitrary matrices having operator norm smaller or equal to ${\tilde \epsilon}$. Note that $W_j(B)$ and $P_j(B)$ commute, since they are both polynomial functions of $B$. Moreover, we can set without loss of generality $W_m(B) \equiv \Id$, since we already have $P_m(B) \approx A^{-1}/(\eta\,K)$ on the entire domain of $A$. Finally, we have $K \in \Theta(\kappa/\eta)$, as derived in the main text.

\subsection{Variable-stopping-time PD-QLS solver: definition}

We now proceed to reformulate our PD-QLS solver as a variable-stopping-time algorithm. For notation convenience, from now on we often drop the dependency of a polynomial from its variable and simply write $P$ instead of $P(x)$.

\begin{algorithm} [Variable-stopping-time linear system solving] 
\label{alg:VTAA_solver}
We preliminarily define two variable-stopping-time algorithms $\A$ and $\B$. $\A$ is the core PD-QLS solving function and $\B$ is used to un-compute the clock registers at the end of the process.

The algorithm $\A  = \A_m \cdot \ldots \cdot \A_1 \cdot \A_0$, with $m := \lceil \log_2 \kappa \rceil + 1$, acts on the Hilbert space $\hil = \hil_C \otimes \hil_F \otimes \hil_{S}$ (respectively clock register, flag qubit, system register with $\dim \hil_S = N$), plus ancillary registers if needed. We set $\A_0 = \Id_C \otimes \Id_F \otimes \U_\b$ (i.e., $\A_0$ prepares $\ket{\b}$ in register $S$) while the other sub-algorithms $\A_j$ use oracular access only to $\U_B$. Each $\A_j$, for $j \geq 1$, consists of the following two steps:
\begin{enumerate}
\item 
Conditioned on the qubits $C_1, \ldots, C_{j-1}$ being in $\ket{0}^{\otimes j-1}$, use QSP and a single Pauli $X$ gate to implement the unitary
\al{
	\label{eq:Uj'}
	\U_j' :=
	\left(
	\begin{array}{cc}
	\sqrt{1 - W_j^2} & - W_j \\
	W_j   & \sqrt{1 - W_j^2}
	\end{array}
	\right)
	\quad  \t{acting~on}  \ \hil_{C_j} \otimes \hil_S ,
}
with the window function $W_j(B)$ given in Eq.~\eqref{eq:P_and_W}.

\item 
Conditioned on $C_j$ being in $\ket{1}_{C_j}$, use QSP and a single Pauli $X$ gate to implement the unitary
\al{
	\label{eq:Uj''}
	\U_j'' :=
	\left(
	\begin{array}{cc}
	\sqrt{\Id - P_j^{2}} & - P_j \\
	P_j & \sqrt{\Id - P_j^{2}}
	\end{array}
	\right)
	\quad  \t{acting~on}  \ \hil_{F} \otimes \hil_S ,
}
with the polynomial approximations $P_j(B)$ given in Eq.~\eqref{eq:P_and_W}. 
\end{enumerate}

The algorithm $\B = \B_m \cdot \ldots \cdot \B_1 \cdot \B_0$ acts on the same Hilbert space as $\A$. The initialisation step is skipped, i.e., we set $\B_0 = \Id$, and $\forall j \geq 1$ we define $\B_j$ as the application of the unitary $\U_j'$ given in Eq.~\eqref{eq:Uj'} controlled on $C_1, \ldots, C_{j-1}$ being in $\ket{0}^{\otimes j-1}$. The unitaries $\U_j''$ are not applied. \vspace{2mm}

The complete PD-QLS solving algorithm is then defined as follows. 
\begin{enumerate}
\item In a initial pre-processing step (that needs to be run only once), a sequence of integers $(k_0, k_1, \ldots, k_m)$ that satisfy Eq.~\eqref{eq:VTAA_condition} is determined using a phase estimation algorithm.

\item The main part of the algorithm consists in running $\A'$, which is a VTAA version of $\A$ where each $\A_j'$ implements a $k_j$-step amplification of $\A_j\A_{j-1}'$.

\item The final post-processing consists in applying $\B^\dag$ to the output of $\A'$ in order to un-compute the clock register. We then measure the flag qubit and when the result is $\ket{1}_F$ (which happens with constant probability) we output the $S$ register (which contains a state close to $\ket{A^{-1} \b}$).
\end{enumerate}

\end{algorithm}

\subsection{Variable-stopping-time PD-QLS solver: correctness}

We will now analyse the correctness of~\alg{VTAA_solver} and we begin introducing a Lemma.

\begin{lemma}
The state after the application of $\A_k$ defined in \alg{VTAA_solver} is given by 
\al{
	\A_{\leq k} \ket{0}_\t{all} 
	= 
	&
	\sum_{j=1}^k
	\ket{1_j}_C
	\left(
	\ket{0}_{F} 
	M_{j-1} W_j \sqrt{\Id - P_j^{2}} \ket{\b}_S 
	+
	\ket{1}_{F} 
	M_{j-1} W_j P_j \ket{\b}_S 
	\right)
	\\
	& 
	+
	M_k
	\ket{0,0,\b}_{C,F,S} 
	\label{eq:induction1}
}
with $M_0 := \Id$, $M_j  :=  \prod_{i=1}^{j} \sqrt{\Id - W_i^2}$ and, using $W_m = \Id$, at the last step we have $M_m = 0$. Moreover, the algorithm $\B$ defined in \alg{VTAA_solver} acts as 
\al{
	\B \ket{0,f,\phi}_{C,F,S} 
	= 
	&
	\sum_{j=1}^k
	\ket{1_j}_C
	\ket{f}_{F} 
	M_{j-1} W_j \ket{\phi}_S 
	\label{eq:inductionB}
}
for any input state $\ket{\phi}_S$ and $f \in \{0,1\}$.
\end{lemma}

\begin{proof}
We proceed by induction over $k$. For $k = 0$ we have $\A_{\leq 0} \ket{0}_\t{all} = \A_0 \ket{0}_\t{all} = \ket{0,0,\b}_{C,F,S}$. Using that $P_j,M_j,W_j$ are functions of $B$ and thus commute, we have at step $k+1$ 
\al{
	\eqref{eq:induction1}
	\ \overset{\U_{k+1}'}{\mapsto} \
	&
	\sum_{j=1}^k
	\ket{1_j}_C
	\left(
	\ket{0}_{F} 
	M_{j-1} W_j \sqrt{\Id - P_j^{2}} \ket{\b}_S 
	+
	\ket{1}_{F} 
	M_{j-1} W_j P_j \ket{\b}_S 
	\right)
	\\
	& +
	\ket{1_{k+1}}_C \ket{0}_F 
	M_k W_{k+1} \ket{\b}_F
	+
	\ket{0,0}_{C,F}
	M_k \sqrt{1-W_{k+1}^2}
	\ket{\b}_{C,F,S} 
	\label{eq:induction2}
	\\
	\eqref{eq:induction2}
	\ \overset{\U_{k+1}''}{\mapsto} \
	&
	\sum_{j=1}^k
	\ket{1_j}_C
	\left(
	\ket{0}_{F} 
	M_{j-1} W_j \sqrt{\Id - P_j^{2}} \ket{\b}_S 
	+
	\ket{1}_{F} 
	M_{j-1} W_j P_j \ket{\b}_S 
	\right)
	\\
	& +
	\ket{1_{k+1}}_C
	\left(
	\ket{0}_{F} 
	M_k W_{k+1}
	\sqrt{\Id - P_{k+1}^{2}} \ket{\b}_S 
	+
	\ket{1}_{F} 
	M_k W_{k+1}
	P_{k+1}
	\ket{\b}_S 
	\right)
	\\[1mm]
	& +
	M_{k+1}
	\ket{0,0,\b}_{C,F,S}  
	.
}

Equation~\eqref{eq:inductionB} can be verified similarly.
\end{proof}

Next, we consider the output state of $\A'$, the VTAA version of $\A$, which features an amplification of the amplitude of the $\ket{1}_F$ component, i.e., it is a state of of the form
\al{
	\label{eq:semifinal_state}
	\A' \ket{0}_\t{all}
	& =
	\sqrt{p_\t{succ}'} \,
	\sket{1}_F
	\sket{\psi_\t{succ}}_{C,S} 	
	\, + \,
	\sqrt{1-p_\t{succ}'} \,
	\sket{0}_F
	\sket{\psi_\t{fail}}_{C,S} 	
}
where the success probability is constant, $p_\t{succ}' \in \Theta(1)$, and where we have
\al{
	\sket{\psi_\t{succ}}_{C,S}  
	& =  
	\frac{1}{\sqrt{\N}}
	\sum_{j=1}^m
	\ket{1_j}_{C}
	M_{j-1} W_j P_j
	\ket{\b}_S 
	\\
	\N & = \sum_{j=1}^m 
	\bnorm{M_{j-1} W_j P_j \ket{\b}}^2 
	\,.
}
We now claim that, for an appropriate choice of the algorithm parameters, the error in $\ket{\psi_\t{succ}}_{C,S}$ is upper bounded by $\O(\varepsilon)$, and we will prove this claim in the next section. Thus we have:
\al{
	\sket{\psi_\t{succ}}_{C,S}
	& =  
	\sum_{j=1}^m
	\ket{1_j}_{C}
	M_{j-1} W_j  
	\ket{A^{-1} \b}_S 
	+
	\O(\varepsilon) 
	\label{eq:VTAA_approx}}
where $\O(x)$ here denotes an arbitrary vector with 2-norm bounded by $\O(x)$. Note that Eq.~\eqref{eq:inductionB} for $\ket{\phi}_S = \ket{A^{-1}\b}_S$ implies that $\sum_{j} \ket{1_j}_{C} M_{j-1} W_j \ket{A^{-1} \b}_S$ is a normalised state.

The final step of the PD-QLS algorithm consists in applying $\B^\dag$ to the state in Eq.~\eqref{eq:semifinal_state}. Using again Eq.~\eqref{eq:inductionB} we obtain:
\al{
	\B^\dag \A' \ket{0}_\t{all}
	=
	\sqrt{p_\t{succ}'} 
	\ket{0,1,A^{-1}\b}_{C,F,S}
	+
	\sqrt{1 - p_\t{succ}'}
	\ket{0,0,\psi'}_{C,F,S}
	+ \O( \varepsilon ) \,.
}
Measuring the flag in $\ket{1}_F$ then results with constant success probability in a vector that has $\O({\tilde \epsilon})$ distance from the ideal output $\ket{A^{-1}\b}$.

\subsection{Variable-stopping-time PD-QLS solver: error bound}

In order to derive Eq.~\eqref{eq:VTAA_approx} we use Eq.~\eqref{eq:Delta_residuals} and write
\al{
	\sket{\psi_\t{succ}}_{C,S}  
	& =  
	\frac{1}{\sqrt{\N}}
	\sum_{j=1}^m
	\ket{1_j}_C
	M_{j-1} W_j  
	\frac{A^{-1}}{\eta \,K} \ket{\b}_S
	+ 
	\frac{1}{\sqrt{\N}}
	\sum_{j=1}^m
	\ket{1_j}_{C}
	\Delta_j^{\tilde \epsilon} \ket{\b}_S
	\\
	& =  
	\frac{\sqrt{\widetilde{\N}}}{\sqrt{\N}}
	\sum_{j=1}^m
	\ket{1_j}_{C}
	M_{j-1} W_j  
	\ket{A^{-1} \b}_S 
	+
	\O({\tilde \epsilon} \sqrt{m / \N})
	\label{eq:VTAA_approx2}
	\\
	& =  \hspace{8mm}
	\sum_{j=1}^m
	\ket{1_j}_{C}
	M_{j-1} W_j  
	\ket{A^{-1} \b}_S 
	+
	\O({\tilde \epsilon} \sqrt{m / \N}) 
}
and $\widetilde{\N}$ is defined below. To prove the last step, recall that $\sum_{j} \ket{1_j}_{C} M_{j-1} W_j \ket{A^{-1} \b}_S$ is a normalised state, hence Eq.~\eqref{eq:VTAA_approx2} implies $\big| \sqrt{\widetilde{\N}/\N} - 1 \big| \in	\O({\tilde \epsilon} \sqrt{m / \N})$. Now we estimate $\N$ by using $\big| \sqrt{\widetilde{\N}} - \sqrt{\N} \big| \in \O({\tilde \epsilon} \sqrt{m})$ and computing
\al{
	\widetilde{\N} 
	& :=
	\sum_{j=1}^m 
	\norm{
	M_{j-1} W_j
	\frac{A^{-1}}{\eta \,K} \ket{\b}
	}^2 
	\\
	& =
	\frac{\snorm{A^{-1} \ket{\b}}^2}{\eta^2 K^2} 
	\sum_{j=1}^m 
	\norm{M_{j-1} W_j \ket{A^{-1} \b}}^2 
	\\
	& =
	\frac{\snorm{A^{-1} \ket{\b}}^2}{\eta^2 K^2} .
}
As proven in the main text, $K \in \O(\kappa/\eta)$. Thus we have $\widetilde{\N} \in \Omega(1 /\kappa^2)$ and choosing
\al{
	{\tilde \epsilon} 
	\in 
	\O\!\left(  
	\frac{\varepsilon}{\kappa \, \sqrt{\log \kappa}}
	\right)
}
we also have $\N\in \Omega(1 /\kappa^2)$ and ${\tilde \epsilon}\sqrt{m / \N} \in \O(\varepsilon)$. This condition is sufficient to guarantee that the output state $\ket{\psi_\t{succ}}$ is within $\varepsilon$-distance from the ideal output.

\subsection{Variable-stopping-time PD-QLS solver: query complexity}

We can now provide a upper bounds the query complexity of the VTAA algorithm using Eq.~\eqref{eq:VTAA_result}:
\begin{align*}
	Q' \in 
	\O \!
	\left(
	t_\t{max} \sqrt{m}
	+ 
	\frac{t_\t{avg}}{\sqrt{p_\mathrm{succ}}}
	\sqrt{m \, \log(t_\t{max} / t_\t{min})}
	\right) .
\end{align*}

We first compute the $\U_\b$-complexity. Since the non-amplified algorithm $\A$ accesses $\U_\b$ only in the first step, we have $t_\t{min} = t_\t{max} = t_\t{avg} = 1$. Using $m \, {\tilde \epsilon}^2 \in \O\big(\varepsilon^2 / \kappa^2 \big)$ we get
\al{
	p_\t{succ} 
		& = \N
		\ \in \ 
		\O\!\left(
		\widetilde{\N} 
		+ m \, {\tilde \epsilon}^2
		\right) 
		= 
		\O\!\left(
		\frac{\snorm{A^{-1} \ket{\b}}^2}{\kappa^2} 
		\right)
}	
and the using $m = \lceil \log_2 \kappa \rceil + 1$ we have
\al{
	Q[\U_\b] 
	\in 
	\O\!\left(
	\sqrt{\log(\kappa)} +
	\frac{ \kappa}{\norm{A^{-1}\ket{\b}}}
	\right) 
	.
}

We now compute the $\U_B$-complexity. We have:
\al{
	t_j 
		& = \sum_{i=1}^j [ \deg(P_i) + \deg(W_i)] 
		\ \in \ 
		\O \Big(
		\delta_j^{-1/2} 
		\log^{1/4}({\tilde \epsilon}^{-1}) \log({\tilde \epsilon}^{-1} \delta_j^{-1})
		\Big) \\[-1mm]
	t_\t{max} 
		& = \sum_{i=1}^m [ \deg(P_i) + \deg(W_i)] 
		\ \in \ 
		\O \Big(
		\sqrt{ \kappa / \eta } \,
		\log^{1/4}({\tilde \epsilon}^{-1}) \log({\tilde \epsilon}^{-1} \kappa / \eta )
		\Big) \\
	t_\t{min} 
		& = \hspace{5mm} \deg(P_1) + \deg(W_1) \
		\ \in \ 
		\O \Big(
		\sqrt{ 1 / \eta } \,
		\log^{1/4}({\tilde \epsilon}^{-1}) \log({\tilde \epsilon}^{-1} / \eta )
		\Big) .
}
Writing $\ket{\b} \equiv \sum_\lambda \beta_\lambda \ket{\lambda}$, where $\ket{\lambda}$ denotes an eigenvector of $B$ relative to the eigenvalue $\lambda$ and $\sum_\lambda |\beta_\lambda|^2 =1$, the probability of stopping at time $t_j$ is
\al{
	p_j 
		& = \bnorm{\,\Pi_{C_j} \A_{\leq j} \ket{0}_\t{all}}^2
		  = \bnorm{\,M_{j-1} W_j \ket{\b}}^2 \\
		& = \sum_\lambda 
		|M_{j-1}(\lambda) W_j(\lambda)|^2 \,
		|\beta_\lambda|^2 \\[-1mm]
		& \in \ 
		\O 
		\Big(
		{\tilde \epsilon} + 
		\sum\nolimits_{\lambda \in (1-4\delta_j,1-\delta_j]}
		|\beta_\lambda|^2
		\Big) .
}
Here, we have used $W_j(\lambda) \leq {\tilde \epsilon}$ for $\lambda \geq 1-\delta_j$, while for $\lambda < 1-2 \delta_{j-1} = 1- 4\delta_{j}$ we have 
\al{ 
	M_{j-1}(\lambda) 
	\leq 
	\sqrt{1-W_{j-1}^2(\lambda)}
	\leq 
	\sqrt{1-(1-{\tilde \epsilon})^2}
	\leq
	\sqrt{2{\tilde \varepsilon}} \,,
}
i.e., the expression $|M_{j-1}(\lambda) W_j(\lambda)|^2 $ is non-negligible only for $\lambda \in (1-4\delta_j,1-\delta_j]$. Now we estimate the $\ell_2$-average runtime $t_\t{avg}= \sqrt{\sum_{j=1}^m p_j t_j^2}$ with
\al{
	t_\t{avg}^2 
	& \in 
	\O\!\left( 
	\sum_{j=1}^m
	\Big[
	{\tilde \epsilon} + 
	\sum\nolimits_{\lambda \in (1-4\delta_j,1-\delta_j]}
	|\beta_\lambda|^2
	\Big]
	\Big[
	\delta_j^{-1/2} 	\log^{1/4}({\tilde \epsilon}^{-1}) \log({\tilde \epsilon}^{-1} \delta_j^{-1})
	\Big]^2
	\right) 
	\\
	& \subseteq
	\O\!\left(
	{\tilde \epsilon} \, t_\t{max}^2 +
	\sum\nolimits_{\lambda}
	\frac{|\beta_\lambda|^2}{1-\lambda} \,
	\log^{1/2}({\tilde \epsilon}^{-1}) \log^2 ({\tilde \epsilon}^{-1} \kappa /\eta)  
	\right) 
}
where we have used that for any positive function $f(\lambda)$
\al{
	\sum_{j=1}^m 
	\sum_{\lambda \in (1-4\delta_j,1-\delta_j]} 
	f(\lambda) \, \frac{1}{\delta_j}
	\ \in \ 
	\Theta\bigg(
	\sum_{\lambda} f(\lambda) \, \frac{1}{1-\lambda}
	\bigg) .
} 
We then write 
$
	\sum_\lambda \frac{|\beta_\lambda|^2}{1-\lambda} 
	= 
	\norm{\sum_\lambda \frac{\beta_\lambda}{\sqrt{\Id - B}} \ket{\lambda}}^2 
	= 
	\norm{(\eta\,A)^{-1/2}\ket{\b}}^2
$
and obtain
\al{
	t_\t{avg}
	& \in 
	\O\!\left( 
	\frac{1}{\sqrt{\eta}}
	\bnorm{A^{-1/2}\ket{\b}} \,
	\log^{1/4}({\tilde \epsilon}^{-1}) \log ({\tilde \epsilon}^{-1} \kappa /\eta) 	
	\right) 
} 
where we used $\norm{A^{-1/2}\ket{\b}} \geq 1$ and thus $\sqrt{{\tilde \epsilon}} \, t_\t{max} \in \O\left(	\frac{1}{\sqrt{\eta}\log^{1/4} \kappa} \log^{1/4}({\tilde \epsilon}^{-1}) \log ({\tilde \epsilon}^{-1} \kappa /\eta)\right)$ is a sub-leading contribution to $t_\t{avg}$. Then, using $1/\sqrt{p_\t{succ}} \in \O\big(\kappa / \norm{A^{-1}\ket{\b}} \big)$ and 
\al{
	\log(t_\t{max}/ t_\t{min}) 
	\ \in \ 
	\O\big(\log(\kappa) +  \log \log ({\tilde \epsilon}^{-1}\kappa / \eta ) \big)
	\ \subseteq \
	\O\big(\log({\tilde \epsilon}^{-1} \kappa /\eta) \big)
}
we finally obtain
\al{
	Q[\U_B] 
	& \in 
	\O\Bigg(
	\sqrt{\frac{\kappa}{\eta}} \, 
	\underbrace{	\sqrt{\kappa} 
	\frac{\norm{A^{-1/2}\ket{\b}}}{\bnorm{A^{-1}\ket{\b}}}}_{:=\Gamma_{A,\b}}
	\underbrace{
	\log^{1/4}({\tilde \epsilon}^{-1}) 
	\log ({\tilde \epsilon}^{-1} \kappa /\eta) 
	\sqrt{ \log(\kappa) \log ({\tilde \epsilon}^{-1}\kappa / \eta ) }
	}_{\t{polylog}(\kappa, {\tilde \epsilon}^{-1}, \eta^{-1})}
	\Bigg) . 
}

\section{Proof that the Sum-QLS problem is BQP-hard}
\label{app:non-dequantization}
\setcounter{equation}{0}

In this Appendix we prove that the Sum-QLS problem is $\mathsf{BQP}$-hard, therefore no efficient classical algorithm can solve the problem (unless $\textsf{BPP} = \textsf{BQP}$). The initial part of the proof is equal to the reduction presented by HHL: it is possible to construct a QLS problem $M\x = \textbf{e}_1$ (where $\textbf{e}_1$ is a canonical vector with a one in the first position) which is $\textsf{BQP}$-hard for a class of sparse indefinite matrices $M$ that are easily constructible. This can be converted to the equivalent QLS problem $M\x = M^\dag \textbf{e}_1$, where the matrix $A = M^\dag M$ is by construction PD and $\ket{M^\dag \textbf{e}_1}$ is easy to prepare. What remains to be proven is that $A$ admits an explicit Sum-QLS structure, i.e.\ that one can construct a decomposition as a sum of PD local Hamiltonian terms and that the overlap with the support (bounded by the parameter $\gamma$) scales polynomially.

Preliminarily, we introduce a couple of useful definitions: the specific quantum circuit model (universal for quantum computation) that we aim to ``simulate'' as a Sum-QLS, and the Feynman-Kitaev clock construction.

\paragraph*{Quantum circuit model:}
We describe an arbitrary quantum computation $\mathcal{C}$ on $n$ qubits as the application of $T$ elementary gates $\{U_0,U_1 \ldots, U_{T-1}\}$ to an initial state $\ket{0}^{\otimes n}$ and the output of the computation is the quantum state $U_{T-1} \cdots U_1 U_0 \ket{0}^{\otimes n}$. We assume that the elementary gate set consists of gates acting either on one or two qubits, e.g.\ arbitrary single-qubit rotations and control-nots. Hence, the $t$-th elementary gate is described by a unitary $U_t \in \mathds{C}^{N \times N}$ with $N = 2^n$ which can be written as $U_t = u_t \otimes I_{\neg \S_t}$, where $u_t$ acts on a subset $\S_t$ of the qubits and is either a $2 \times 2$ ($|\S_t| = 1$) or a $4 \times 4$ ($|\S_t| = 2)$ unitary matrix.

\paragraph*{Feynman-Kitaev clock:}
We introduce the following unitary matrix acting on $\mathds{C}^{3T} \otimes \mathds{C}^n$, which is based on the \emph{Feynman-Kitaev clock} construction  (see e.g.~\cite{Bausch18} and references therein):
\al{
\label{eq:FK-clock}
	U 
	:= &  	
	\sum_{t=0}^{3T-1}
	\ket{   t\!+\!1}\!\bra{   t}_c \otimes U_t 
 	\\
	= & \
	\sum_{t=0}^{T-1} \,
	\ket{       t\!+\!1}\!\bra{       t}_c \otimes U_t \,+\,
	\ket{ T\!+\!t\!+\!1}\!\bra{ T\!+\!t}_c \otimes I   \,+\,
	\ket{2T\!+\!t\!+\!1}\!\bra{2T\!+\!t}_c \otimes U_{T-t-1}^\dag
}
where we have defined $U_t = I$ for $T \leq t \leq 2T-1$ and $U_t = U^\dag_{3T-t-1}$ for $2T \leq t \leq 3T-1$ and the sums in the clock register (denoted by the subscript $c$) are taken modulo $3T$.

\begin{proposition} [Sum-QLS is \textsf{BQP}-hard]
\label{prop:BQP-complete}
Let $\mathcal{C}$ be a given $n$-qubit $T$-gate quantum circuit. Then, there is a Sum-QLS problem (as in \defin{Sum-QLS}) with
\begin{enumerate}[itemsep=-.8\parsep]
\item $n' \in \O(n + \log T)$ \hfill [$N' \in \O(2^{n'})$ is the size of $A$]
\item $\kappa \in \O(T^2)$ \hfill [condition number]
\item $J \in \O(T)$ \hfill [number of PD Hamiltonian terms]
\item $s \in \O(\log T)$ \hfill [locality of the Hamiltonian terms]
\item $d_\b \in \O(1)$ \hfill [sparsity of the known-term vector $\b$]
\item $\gamma \in \Omega(T^{-2})$ \hfill [overlap with the support, Eq.~\eqref{eq:gamma_main}]
\end{enumerate}
that is equivalent to $\mathcal{C}$, i.e., solving this Sum-QLS problem (up to a small constant error $\varepsilon$) allows to obtain the output state of $\mathcal{C}$ with constant probability and constant precision. According to Eq.~\eqref{eq:intermediate_estimation}, the algorithm presented in \sec{hamiltonians} can solve this Sum-QLS problem with a gate complexity in $\O\big(\t{poly}(n,T)\big)$. When $T \in \O(\t{poly}\,n)$, the gate complexity of the Sum-QLS solver is also in $\O(\t{poly}\,n)$, i.e., Sum-QLS problem is \textsf{BQP}-hard. Moreover, the Sum-QLS$_\t{poly}$, defined as the subclass of problems where the six parameters listed above all scale polynomially in $n$, is \textsf{BQP}-complete.
\end{proposition}

\begin{proof}

Starting from the Feynman-Kitaev unitary $U$ encoding a quantum circuit $\mathcal{C}$ as in Eq.~\eqref{eq:FK-clock}, we introduce the matrix
\al{
	 M := \Id - U\, e^{-1/T}
}
which can be written in block form (where each block has size $2^n \times 2^n$) as
\al{
	 M
	 = &
	\left(
	\begin{array}{ccccc}
	\ddots & ~~\ddots~~~~ & 0 & & \\
	\ddots & I & -U_t e^{-1/T} & 0 & \\
	 & 0 & I & -U_{t+1} e^{-1/T} &  \\
	&  & 0 & I & ~~~\ddots~~~ \\
	-U_{3T-1} e^{-1/T} & &  & \ddots & \ddots  \\
	\end{array}
	\right)
}
and moreover we introduce:
\al{
	 \label{eq:A_appendix}
	 A & := M^\dag M  = (1 + e^{-2/T}) \Id - e^{-1/T} (U + U^\dag)
}
where $A$ is by construction Hermitian positive definite and the singular values of $M$ are the square root of the eigenvalues of $A$. From the eigenvalue inequality $-2 \leq \lambda(U + U^\dag) \leq 2$ we obtain 
\al{
\begin{cases}
	\lambda_{\max}(A) 
	\leq 4 
	\\
	\lambda_{\min}(A) 
	\geq 
	\big(1-e^{-1/T}\big)^2 
	\geq 
	\frac{(1-e^{-1})^2}{T^2} 
\end{cases}	
	\Longrightarrow \quad
	\kappa(A) \leq \frac{4}{(1-e^{-1})^2}\,T^2
}
hence $\kappa(M) = \sqrt{\kappa(A)} \in \O(T)$. Using the identity $U^{3T} = \Id$, the inverse of $M$ can be expanded as
\al{
	M^{-1} 
	& =
	\sum_{t'=0}^\infty
	U^{t'} e^{-t'/T}
	 =
	\frac{e^3}{e^3-1}
	\sum_{t=0}^{3T-1}
	U^{t'} e^{-t'/T}	
	\\
	& = 
	\frac{e^3}{e^3-1}
	\sum_{t'=0}^{3T-1}
	e^{-t'/T}
	\sum_{t=0}^{3T-1}
	\ket{t\!+\!t'}\!\bra{t} \otimes U_{[t:t+t']}
}
where the notation $U_{[a:b]}$ indicates the ordered product of unitary operators from $a$ to $b-1$, $U_{[a:b]} := U_{b-1} U_{b-2} \cdots U_a$, with $U_{[t:t]} = \Id$.

As done already by HHL in Ref.~\cite{HHL} step, we consider the QLS problem
\al{
\label{eq:M_sysytem}
	M \x = \textbf{e}_1
}
where $\textbf{e}_1 \in \mathds{C}^{3T2^n}$ is the unit vector with a one in the first position (i.e., it is the vector representation of the state $\ket{0}_c \ket{0}^{\otimes n}$). Using the state $\ket{\x} = \ket{M^{-1} \textbf{e}_1}$ the output state of the quantum computation is obtained whenever upon measurement of the clock register a time $t \in [T:2T-1]$ is returned, which happens with probability $e^{-2}/(1+e^{-2}+e^{-4}) \geq 0.11$. Therefore, any quantum circuit $\mathcal{C}$ having $T$ gates can be restated as a QLS and solved (using e.g.\ the HHL algorithm) with a gate complexity scaling as $\t{poly}(n,T)$. If the circuit $\mathcal{C}$ has $T \in \O(\t{poly}\,n)$, then also the QLS solver has gate complexity in $\O(\t{poly}\,n)$, showing that the QLS problem is $\textsf{BQP}$-hard.

We now consider the PD linear system
\al{
\label{eq:appendix_system}
	A \x = \b
	\qquad
	\t{for} \
	A = M^\dag M \,, \ \;
	\b & = M^\dag \textbf{e}_1 
}
which is equivalent to the system in Eq.~\eqref{eq:M_sysytem} and can be cast as a Sum-QLS problem.

Notice, first, that $\b$ is a sparse vector, with $d_\b \leq 3$: if the elementary gate set consists of control-nots and single-qubit rotations, then the first row of $M$ has at most $3$ non-zero entries.

The matrix $A$ has size $3T2^n \times 3T2^n$ and can be written as 
\al{
	A = \sum_{t=0}^{3T-1} H_{(t)}\;.
}
where each Hamiltonian terms $H_{(t)}$ is positive definite and only acts on the clock register plus either one or two extra qubits (i.e., the same number of qubits on which the gates from the elementary set act). We start writing $A$ in block form
\al{
\label{eq:A=sum}
	A =
	\left(
	\begin{array}{ccccc}
	\ddots & \ddots & 0 & & \\
	\ddots & I(1+e^{-2/T}) & -U_t e^{-1/T} & 0 & \\
	~~~0~~~ & -U_t^\dag e^{-1/T} & I(1+e^{-2/T}) & -U_{t+1} e^{-1/T} & ~~~0~~~ \\
	& 0 & -U_{t+1}^\dag e^{-1/T} & I(1+e^{-2/T}) & \ddots \\
	& & 0 & \ddots & \ddots  \\
	\end{array}
	\right)
}
and then introduce, for each $t \in \{0,\ldots,3T-1\}$, the Hamiltonian term
\al{
\label{eq:H_app}
	H_{(t)} 
	\; = \;
	I\, \delta 
	\; + \;
	e^{-1/T}
	\left(
	\begin{array}{c|cc|c}
	~0~ & 0 & 0 & ~0~ \\
	\hline
	0 & I & -U_t & 0 \\
	0 & -U_t^\dag & I & 0 \\
	\hline
	0 & 0 & 0 & 0  \\
	\end{array}
	\right)
	\begin{array}{l}
	\hspace{1pt}\rbrace ~ \t{size}~t\cdot 2^n \\[2.2mm]
	\Big{\rbrace}       ~ \t{size}~2\cdot 2^n \\[2.5mm]
	\hspace{1pt}\rbrace ~ \t{size}~(3T-t-2)\cdot 2^n
	\end{array}
}
where we have defined
\al{
\label{eq:delta-bound}
	\delta
	\; := \; 
	\frac{1}{3T}
	\big(1+e^{-2/T}-2e^{-1/T}\big) 
	\; \geq \; 
	\frac{(1-e^{-1})^2}{ 3T^3} 
	\; \geq \;
	\frac{1}{7.51 \, T^3} \;,
}
and thus the decomposition in Eq.~\eqref{eq:A=sum} holds. Note that the matrix 
$\left(\begin{smallmatrix} I & -U_t\\ -U_t^\dag & I\end{smallmatrix}\right)$ has eigenvalues $\lambda =0$ and $\lambda =2$, hence is positive semi-definite, while each $H_{(t)}$ has the smallest eigenvalue equal to $\delta$. Moreover, each $H_{(t)}$ is a local Hamiltonian. Specifically, we have $H_{(t)} = h_{(t)} \otimes I_{\neg \S_t}$, where $h_{(t)}$ acts on a set $\S_t$ consisting of $s = \lceil \log_2 3T \rceil + q$ qubits, where $q=1$ if $U_t$ corresponds to a single qubit gate and $q=2$ if it corresponds to a two-qubit gate. Explicitly, we have
\al{
	h_{(t)} 
	\; = \;	
	\left(
	\begin{array}{c|cc|c}
	I \,\delta & 0 & 0 & ~0~ \\
	\hline
	0 & I(e^{-1/T} \!+\!\delta)\! & -u_t e^{-1/T} & 0 \\
	0 & -u_t^\dag e^{-1/T} & \! I(e^{-1/T} \!+\!\delta) & 0 \\
	\hline
	0 & 0 & 0 & I\,\delta \\
	\end{array}
	\right)
	\begin{array}{l}
	\hspace{1pt}\rbrace ~ \t{size}~t\cdot 2^q \\[2.2mm]
	\Big{\rbrace}       ~ \t{size}~2\cdot 2^q \\[2.5mm]
	\hspace{1pt}\rbrace ~ \t{size}~(3T-t-2)\cdot 2^q
	\end{array}
}
where each $u_t$ is a matrix specifying the $t$-th elementary quantum gate, with $U_t = u_t \otimes I_{\neg \S_t}$.

We now need to estimate a lower bound $\gamma$ for the overlap parameter as in Eq.~\eqref{eq:gamma_main}. We have
\al{
	\bra{\b} A^{-1} \ket{\b} 
	& = 
	\bra{M^\dag \textbf{e}_1} M^{-1}M^{-\dag} \ket{M^\dag \textbf{e}_1} \\
	& = 
	\frac{\bra{\textbf{e}_1} M M^{-1}M^{-\dag} M^\dag\ket{\textbf{e}_1}}
	     {\norm{M^\dag\ket{\textbf{e}_1}}^2} 
	\\
	& =
	\frac{1}{\norm{M^\dag\ket{\textbf{e}_1}}^2} 
	= 
	\frac{1}{1+e^{-2/T}}
	\geq 
	1/2
}
while from~\eqref{eq:H_app} and~\eqref{eq:delta-bound} we obtain 
\al{
	\bra{\b} H_{(j)}^{-1} \ket{\b} 
	\,\leq\, 
	\delta^{-1}
	\,\leq\, 
	7.51 \, T^3\,.
}
Using $J = 3T$ we finally get
\al{
	\frac{1}{J^2}\,
	\frac{\sum_{j=1}^J \bra{\b} H_{(j)}^{-1} \ket{\b}}
	     {\bra{\b} A^{-1} \ket{\b}} 
	\,\leq\,
	\frac{1}{(3T)^2}\,
	\frac{ 3T \times 7.51 \, T^3}{1/2}
	\, \leq \, 
	5.01 \,T^2
	\, =: \,
	\frac{1}{\gamma}\,.
}

In conclusion, given the sequence of gates $u_t$ and the set of qubits $\S_t$ to which they are applied in the quantum circuit $\mathcal{C}$, we can efficiently compute the values and the positions of the non-zero entries of $H_{(t)}$, as given in Eq.~\eqref{eq:H_app}. This then explicitly describes $A$ as a sum of PD local Hamiltonian terms, as required by the definition of the Sum-QLS problem. Going through the derivation, we see that the relevant parameters scale as stated in points (1) to (6) in the statement of the Proposition. Using Eq.~\eqref{eq:intermediate_estimation}, it follows that the algorithm presented in \sec{hamiltonians} solves this problem with a gate complexity in $\O\big(\t{poly}(n,T)\big) = \O(\t{poly}\,n)$, in the case where the original circuit is itself a polynomial time quantum circuit (i.e.\ $T \in \O(\t{poly}\,n)$). Finally, note that the $\varepsilon$-error in precision of Sum-QLS solver is amplified at most by a constant factor when post-selecting the clock register to show a time $t \in [T:2T-1]$. Thus, it is possible to obtain the output of $\mathcal{C}$ with an error that is bounded by a constant, as required per the definition of the \textsf{BQP} class. This finally proves that the Sum-QLS problem is \textsf{BQP}-hard; at the same time, it proves that Sum-QLS$_\t{poly}$, the sub-class of problems having polynomially scaling parameters, is solvable in quantum polynomial time and is thus \textsf{BQP}-complete.

\end{proof}

\bibliographystyle{unsrt}

\end{document}